\documentclass[onefignum,onetabnum]{siamonline250106}

\usepackage{lipsum}
\usepackage{amsfonts}
\usepackage{graphicx}
\usepackage{epstopdf}
\usepackage{algorithmic}
\usepackage{tabularx}
\usepackage{float}
\usepackage{caption}
\usepackage{subcaption}
\usepackage{amsmath,amssymb,amsfonts,bm}
\usepackage{color}
\usepackage[letterpaper,margin=1in]{geometry}
\usepackage{bbm}
\usepackage{comment}
\usepackage[textsize=tiny]{todonotes}
\usepackage{mathtools}
\usepackage{empheq} 
\usepackage{commath}
\usepackage[titletoc]{appendix}
\usepackage{textcomp}
\usepackage[normalem]{ulem}
\usepackage{enumitem}
\usepackage{placeins}

\DeclareMathOperator*{\argmax}{arg\,max}
\DeclareMathOperator*{\argmin}{arg\,min}

\newcommand{\indep}{\perp \!\!\! \perp}
\DeclareMathOperator{\Tr}{Tr}

\DeclareMathOperator{\PCA}{PCA}
\DeclareMathOperator{\CCA}{CCA}
\DeclareMathOperator{\CMI}{CMI}
\newcommand{\EIG}{\mathrm{EIG}}
\newcommand{\OPT}{\mathrm{OPT}}
\newcommand{\R}{\mathbb{R}}
\newcommand{\Normal}{\mathcal{N}}
\newcommand{\I}{\mathcal{I}}
\newcommand{\KL}{\text{KL}}
\DeclareMathOperator{\Dkl}{\mathcal{D}_{\KL}}

\definecolor{darkgreen}{rgb}{.15,.55,0}

\newcommand{\rsb}[1]{\textcolor{black}{#1}} 
\newcommand{\fl}[1]{\textcolor{black}{#1}} %
\newcommand{\revise}[1]{\textcolor{black}{#1}}

\graphicspath{{./plots/}}
\ifpdf
  \DeclareGraphicsExtensions{.eps,.pdf,.png,.jpg}
\else
  \DeclareGraphicsExtensions{.eps}
\fi

\usepackage{enumitem}
\setlist[enumerate]{leftmargin=.5in}
\setlist[itemize]{leftmargin=.5in}

\newsiamremark{remark}{Remark}
\newsiamremark{hypothesis}{Hypothesis}
\crefname{hypothesis}{Hypothesis}{Hypotheses}
\newsiamthm{claim}{Claim}
\newsiamremark{assumption}{Assumption}

\headers{EIG estimation using transport maps}{F. Li, R. Baptista, and Y. Marzouk}

\title{Expected information gain estimation via density approximations:\\ Sample allocation and dimension reduction} 

\author{Fengyi Li\thanks{Massachusetts Institute of Technology  (\email{fengyil@mit.edu}, \email{ymarz@mit.edu}).}
\and Ricardo Baptista\footnotemark[2]\thanks{University of Toronto
  (\email{r.baptista@utoronto.ca}) }
\and Youssef Marzouk\footnotemark[1]}

\usepackage{amsopn}

\makeatletter
\newcommand*{\addFileDependency}[1]{
  \typeout{(#1)}
  \@addtofilelist{#1}
  \IfFileExists{#1}{}{\typeout{No file #1.}}
}
\makeatother

\newcommand*{\myexternaldocument}[1]{%
    \externaldocument{#1}%
    \addFileDependency{#1.tex}%
    \addFileDependency{#1.aux}%
}

\myexternaldocument{ex_supplement}

\ifpdf
\hypersetup{
  pdftitle={EIG estimation: Sample allocation and dimension reduction},
  pdfauthor={F. Li, R. Baptista, and Y. Marzouk}
}
\fi

\begin{document}

\maketitle

\begin{abstract}
Computing expected information gain (EIG) from prior to posterior (equivalently, mutual information between candidate observations and model parameters or other quantities of interest) is a fundamental challenge in Bayesian optimal experimental design. We formulate flexible transport-based schemes for EIG estimation in general nonlinear/non-Gaussian settings, compatible with both standard and implicit Bayesian models. These schemes are representative of two-stage methods for estimating or bounding EIG using marginal and conditional density estimates. In this setting, we analyze the optimal allocation of samples between training (density estimation) and approximation of the outer prior expectation. We show that with this optimal sample allocation, the mean squared error (MSE) of the resulting EIG estimator converges more quickly than that of a standard nested Monte Carlo scheme. We then address the estimation of EIG in high dimensions, by deriving gradient-based upper bounds on the mutual information lost by projecting the parameters and/or observations to lower-dimensional subspaces. Minimizing these upper bounds yields projectors and hence low-dimensional EIG approximations that outperform approximations obtained via other linear dimension reduction schemes. Numerical experiments on a PDE-constrained Bayesian inverse problem also illustrate a favorable trade-off between dimension truncation and the modeling of non-Gaussianity, when estimating EIG from finite samples in high dimensions.

\end{abstract}

\begin{keywords}
  Optimal experimental design, expected information gain, mutual information, density estimation, transportation of measure, dimension reduction
\end{keywords}

\begin{AMS}
  62K05, 62B10, 62G07, 62F15
\end{AMS}

\section{Introduction}
Optimal experimental design (OED) problems are ubiquitous in the physical and biological sciences, engineering, social sciences, and beyond \cite{HuanJMacta2024,Attia17,Krause08,clinical_application, social_application}.
The goal of OED is to identify observational or experimental configurations that are---in a specific sense defined by the experimenter---most {useful}, before the associated data are collected. Most OED approaches begin with a parametric statistical model for the observations and seek a design that maximizes, for example, a notion of \emph{information gain} or uncertainty reduction in the model's parameters, 
or another utility that reflects a downstream task in which the parameters are involved. Typically these objectives can be maximized only in expectation over the data.

To make this idea concrete, let the observations associated with some candidate design be represented as a random vector $Y$ taking values in $\R^{n_y}$, such that our parametric statistical model is specified by conditional probability density functions\footnote{Here and throughout the paper, we assume that all distributions are absolutely continuous with respect to Lebesgue measure, so that these Lebesgue densities exist.} of $Y$, $\pi_{Y | X, d}(\cdot \, | x, d)$, for parameter values $x \in \R^{n_x}$ and design parameters $d \in \Xi$. 
We work in the Bayesian setting, such that the parameters $x$ are treated as a random vector $X$ and endowed with a prior density $\pi_X$. The posterior density of $X$ for observations $y$ is given by Bayes' rule as
\begin{equation} \label{eq:Bayes_formula}
    \pi_{X | Y,d}(x | y, d) = \frac{\pi_{Y | X, d}(y|x, d) \, \pi_X(x)}{\pi_{Y | d}(y | d)}.
\end{equation}
Here, the marginal density of the observations $\pi_{Y|d}$ is also called the evidence or marginal likelihood, \revise{and $x \mapsto \pi_{Y | X, d}(y |x, d)$ is the likelihood function.} In the Bayesian paradigm, the prior and posterior distributions represent, respectively, our states of knowledge before and after having observed $Y=y$. We always assume that the prior density of $X$ is functionally independent of the design $d$.

In the special case of linear-Gaussian models, i.e., $Y \vert x, d \sim \Normal(G(d) x, \Sigma_{Y\vert X})$ for some design-dependent matrix $G(d) \in \R^{n_y \times n_x}$ and covariance matrix $\Sigma_{Y\vert X}$, the classical ``alphabetic optimality'' criteria are simply functions of the Fisher information matrix $G(d)^\top \Sigma_{Y\vert X}^{-1} G(d)$ of the model, which is independent of $x$. For example, choosing $d$ to maximize the trace of the Fisher information matrix (A-optimality) is a natural criterion if one would like to minimize the variance of parameter estimates on average for any observation; similarly, maximizing the determinant of the Fisher information (D-optimality) minimizes the volume of the confidence ellipsoid~\cite{BOED_review}. In the \textit{Bayesian} linear-Gaussian setting, these alphabetic optimality criteria can be modified to include the prior: for instance, Bayesian $D$-optimality then seeks to minimize the determinant of the posterior covariance, and so on.

In the nonlinear setting, the Fisher information matrix---and the Hessian of the log-posterior density---depend on the parameters $x$, and thus these simple alphabetic optimality criteria do not directly apply. Many have suggested averaging these criteria over the parameter space, leading to a variety of nonlinear design heuristics \cite{nonlinear_alphabetic_regression, pronzato2013design} which in some cases permit interesting information theoretic interpretations \cite{walker2016,overstall2022,prangle2023}. 
A more general and principled approach, however, is to consider the \emph{expected information gain} (EIG) from prior to posterior, i.e., 
\begin{equation}\label{eig_kl}
   \EIG(d) \coloneqq \mathbb E_{Y|d} \left[ \Dkl (\pi_{X|Y,d} \, || \, \pi_X) \right],   
\end{equation}
where $\Dkl$ denotes the Kullback--Leibler (KL) divergence. This design objective was advocated by~\cite{lindley} as a natural measure of information due to an experiment. It can be justified from a decision-theoretic perspective \cite{Bernardo1979,BOED_review} as the expectation of a utility function based on Shannon information. The \revise{EIG} enjoys other useful properties, such as invariance under reparameterization (inherited from the KL divergence). It is also equivalent to the \emph{mutual information} (MI) between $ X$ and $ Y | d$, denoted $\I( X; Y \vert d)$ \cite{thomas2006elements}. We thus will use these two terms interchangeably.
In the Bayesian linear-Gaussian model, maximizers of the expected information gain from prior to posterior coincide with Bayesian D-optimal designs, i.e., $d^\ast \in \argmin_{d\in \Xi} \log \det \Sigma_{X \vert Y}(d)$, where $\Sigma_{X \vert Y}(d)$ is the design-dependent posterior covariance matrix. In general, however (e.g., outside of the linear-Gaussian setting), closed-form expressions for the EIG are unavailable. Estimating or bounding the EIG is in fact one of the \emph{central challenges of nonlinear Bayesian experimental design.}

This paper proposes and analyzes new \textit{transport-based estimation schemes for EIG in general nonlinear/non-Gaussian settings}. We address settings where the likelihood function and prior density can be evaluated directly, but also ``implicit model'' \cite{cranmer,Kleinegesse2019EfficientBE} settings where one can only simulate (i.e., draw samples from) the prior and/or the data-generating distribution. To improve the estimation of EIG in high dimensions, we then introduce new \textit{dimension reduction} schemes that seek low-dimensional projections of the data and parameters that best preserve the \revise{MI} $\I(X; Y)$. \revise{This dimension reduction approach is adapted from~\cite{Ricarod_dimRed}}.

\subsection{Related work}
To frame our contributions, it is useful to expand the EIG~\eqref{eig_kl} into expectations over marginal and conditional densities of $X$ and $Y$. For notational simplicity, we drop explicit dependence on $d$:
{\allowdisplaybreaks
\begin{align}
    \mathrm{EIG}  
    =& \int\int \pi_{X,Y} ( x, y) \log \frac{\pi_{X, Y} ( x, y)}{\pi_X( x)\pi_Y(y)} \, d x \, d y, \nonumber \\
    =& \, \mathbb E_{\pi_{X,Y}} \left[\log \frac{\pi_{X \vert Y} ( X | Y)}{\pi_X( X)}\right] \label{eq:EIGpost} \\
    =& \, \mathbb E_{\pi_{X,Y}} \left[ \log \frac{\pi_{Y \vert X} ( Y | X)}{\pi_Y( Y)}\right]. \label{eq:EIGmarg}
\end{align}
From the last two expressions, we see that estimating EIG requires either the \textit{normalized} posterior density $\pi_{X \vert Y}$ and the prior density $\pi_X$ \eqref{eq:EIGpost}, or the conditional density $\pi_{Y\vert X}$ (and hence the likelihood) and the evidence $\pi_Y$ \eqref{eq:EIGmarg}. To compute the outer expectation, all of these densities must be evaluated over a range of arguments $x$ and $y$. Depending on what information is available, the problem of EIG estimation can be divided into settings of (roughly) increasing difficulty: 
\begin{enumerate}
    \item One can sample from the joint distribution of parameters and data $\pi_{X,Y}$ (typically by drawing $x^i \sim \pi_X$ and $y^i \sim \pi_{Y | X}(\cdot \,  | x^i)$) and also evaluate the prior density and likelihood function.\label{assump1}
    \item One can sample from joint distribution of parameters and data, but can only evaluate the marginal prior density $\pi_{X}$.\label{assump2}
    \item One can sample from the joint distribution of parameters and data, but cannot evaluate any of its marginal or conditional densities, and hence neither the prior density nor the likelihood. \label{assump3}
\end{enumerate}
The second and third settings above are ``likelihood-free,'' in that likelihood evaluations are unavailable and one can only simulate $Y \vert X = x$ for any $x$. Data-generating models of this form are often called ``implicit'' and are the prototypical target of simulation-based inference (SBI) methods \cite{cranmer} and 
\revise{approximate Bayesian computation (ABC)~\cite{del2012adaptive, sisson2018handbook, OED_ABC}. SBI and ABC methods produce posterior samples that allow further downstream tasks, such as estimating the posterior density or density ratio, to be easily performed~\cite{chakraborty2024likelihoodfreeapproachgoalorientedbayesian}. Our goal here is somewhat different, as we do not seek samples directly from the posterior.}
The third setting is additionally ``prior-free'' and encompasses, for instance, generative prior models whose densities are unavailable. In both cases, the estimation of densities from samples is required.  

In the first setting, a workhorse approach to EIG estimation has been nested Monte Carlo (NMC), which to our knowledge was first proposed by~\cite{ryan}. The core idea is to independently estimate the evidence $\pi_Y(y)$ for different values of $y$ and to embed these estimates in an outer Monte Carlo estimator of the expectation over $Y$, following~\eqref{eq:EIGmarg}. NMC is biased at finite sample sizes, but asymptotically unbiased; it also converges more slowly than standard Monte Carlo (see~\cite{Rainforth2018OnNM} and Section~\ref{sec:num_res}). For more details on NMC and its variants, we refer to \cite[Section 3.1]{HuanJMacta2024}. \revise{~\cite{Goda03072020} introduce a more efficient multilevel Monte Carlo estimator of the EIG and report that, for the same level of accuracy, the required number of samples can be reduced by several orders of magnitude compared to simple NMC (see Section~\ref{subsec:LG} for a detailed comparison).}
In the third setting, a classical method for estimating MI from samples is the nonparametric Kraskov--St\"ogbauer--Grassberger (KSG) estimator~\cite{knn_MI}, which is based on the statistics of nearest-neighbor distances. The error of this estimator scales poorly with dimension \cite{GaoOhViswanath2018}, however, and it does not take advantage of smoothness in the associated densities.



A more modern family of sample-driven approaches involves constructing and optimizing variational bounds for EIG. In general, these approaches involve writing a variational representation of \revise{MI} and optimizing this representation (thus tightening the bound) over a family of so-called ``critic'' functions. Prominent examples include the ``mutual information neural estimation'' approach of \cite{MINE}, which uses the Donsker--Varadhan variational representation of the KL divergence to construct a lower bound for MI, and then maximizes this bound over a family of critic functions represented by deep neural networks. Related methods include the ``tractable unnormalized Barber--Agakov'' bound introduced in \cite{Poole2019OnVB} and the NWJ (Nguyen, Wainwright, Jordan) estimator of \cite{Nguyen2010_MI}. We refer to \cite{Poole2019OnVB} for links among these and other variational bounds as well as trade-offs between the bias and variance of the resulting EIG estimators. One key point is that the optimal critic functions in these bounds typically coincide with the log-conditional density $\log \pi_{Y \vert X}$, the log-evidence $\log \pi_Y$ and/or the log-density ratio $\log \pi_{X \vert Y}  - \log \pi_X  = \log \pi_{Y \vert X} - \log \pi_Y$, depending on the specifics of each construction.

The use of variational bounds is thus linked to density approximation, and an alternative approach is to optimize \fl{the deficit of a bound} over families of density functions (i.e., non-negative functions that integrate to one) rather than generic critic functions. In the context of OED, this approach was first suggested by \cite{Foster2019VariationalBO}, who construct a lower bound for EIG by approximating the posterior density $\pi_{X \vert Y}$ and an upper bound for EIG by approximating the marginal $\pi_Y$. \cite{Foster2019VariationalBO} demonstrate their approach using simple mean-field variational approximations in standard parametric families. Although such approximations are easy to implement, they do not, in general, yield tight bounds (e.g., if the true posterior or evidence is far from the chosen family). 
We also note that other estimators of EIG or \revise{MI} are based on explicit density \textit{ratio} estimation~\cite{qin1998inferences, density_ratio_estimation, bickel2007discriminative,chakraborty2024likelihoodfreeapproachgoalorientedbayesian}. In the machine learning literature, MI estimators that involve learning the density ratio are sometimes called `discriminative,' while those based on learning the relevant densities directly are called `generative'~\cite{letizia}. The comparative performance of these two classes of methods is in general not clear. Here, we will concentrate on the latter.

To this end, recent work has replaced the simpler density approximations of~\cite{Foster2019VariationalBO} with more flexible density approximations based on transportation of measure. For instance, \cite{Herrmann,dong2024variational,ihlerOEDNFs} use conditional normalizing flows \cite{pmlr-v37-rezende15,Papamakarios_NF} trained from samples of $\pi_{X,Y}$ to approximate the density of $\pi_{X \vert Y}$. 
\cite{Koval2024TractableOE} instead constructs functional tensor-train approximations~\cite{dirt} of the triangular Knothe--Rosenblatt (KR) transport maps using direct evaluations of the joint density $\pi_{X,Y}$, which is further extended by \cite{cui2025subspaceacceleratedmeasuretransport} (combined with dimension reduction) to the sequential design setting. The approach we develop in this paper is also based on transport. In contrast with~\cite{Koval2024TractableOE}, we estimate the KR map from samples, and then use the associated plug-in estimator of any desired marginal and/or conditional density to form an EIG approximation. Depending on the problem setting, this approximation may also be an upper or lower bound. Our KR approximation is built on a general representation of monotone triangular maps that employs an invertible rectification operator, proposed in \cite{baptista2020adaptive}; this operator transforms generic (non-monotone) functions $f$ into monotone invertible functions, in particular triangular diffeomorphisms. The resulting EIG approximations are thus parameterized by $f$, and in this sense our EIG approximation comes full circle and is analogous to a critic-based variational bound. Crucially, under appropriate assumptions on the target distribution, we can take advantage of theory in \cite{baptista2020adaptive} that guarantees recovery of the exact KR rearrangement and hence the exact marginals and conditionals. Thus, as the class of functions chosen to represent $f$ is enriched, our bounds on the EIG can become arbitrarily tight. A specific version of EIG estimators explored in this paper was also demonstrated in~\cite{baptista2022bayesian}.


The methods discussed above can be viewed effectively as \emph{two-step procedures:} first, density approximations
%
%
are learned from a set of ``training samples''; then, the outer prior expectation (e.g., in \eqref{eq:EIGpost} or \eqref{eq:EIGmarg}) is estimated via Monte Carlo with a separate set of ``evaluation samples.'' 
%
%
%
In many applications of OED, the generation of samples $(y^i, x^i) \sim \pi_{Y,X}$ is the most computationally expensive element of the problem, as simulating the data $Y \vert x^i$ requires evaluating a computationally intensive model. It is thus natural to ask how to make best use of a given sample budget---i.e., how to balance the training sample size $N$ and evaluation sample size $M$, given a total sample size $L = M+N$, to minimize the mean squared error of the EIG estimate. To our knowledge, this sample allocation question has not yet been analyzed in the literature. Yet it is analogous to one that has previously been posed for NMC estimators, where the inner-loop sample size (for estimating the evidence $\pi_Y(y)$ at any $y$) must be balanced with the outer-loop sample size to achieve an optimal convergence rate of the MSE with $L$ \cite{Tempone_NMC,Rainforth2018OnNM}. 


We also note that most of the aforementioned EIG estimation methods, despite their widespread adoption, have only been applied in relatively low-dimensional regimes. Intuitively, these methods either explicitly 
or implicitly 
involve density estimation, and in the absence of further structure, the number of samples required for non-parametric or semi-parametric density estimation scales exponentially with dimension \cite{McDonald2017MinimaxDE, wang2022minimax}. 
Some papers on density ratio estimation do include high-dimensional examples, but these examples are either Gaussian \cite{choi2021} or the density ratio is parameterized by a simple exponential family \cite{density_ratio_estimation}. 
Yet a recurring structure in high-dimensional Bayesian inference problems is that the change from prior to posterior is well-captured by a \textit{low-dimensional subspace} of the parameter space \cite{LIS,zahm2022certified,brennan2020greedy, Cui_2021, cui2025subspaceacceleratedmeasuretransport}. Similarly, conditioning on a \textit{low-dimensional projection} of the realized observations $y$ might closely approximate the result of conditioning on the full vector $y$ \cite{giraldi,jayanth, li2024nonlinear, li2024new}. An information theoretic perspective on these projections is developed in \cite{Ricarod_dimRed}. In the present paper, we will use this perspective to create dimension reduction methods for OED, with \textit{error guarantees} on the estimated value of the EIG.

Prior work on dimension reduction for OED has taken several forms. \cite{omar_JUQ} uses Laplace approximations of the posterior, averaged over many realizations of the data, to approximate the EIG. These Laplace approximations are simplified using low-rank approximations of the prior-preconditioned Hessian of the log-likelihood, which is equivalent to parameter-space dimension reduction for linear (or linearized) problems~\cite{spantini2015}. \cite{cui2025subspaceacceleratedmeasuretransport} utilizes a likelihood-informed subspace to identify the informative direction, which is subsequently used to reduce the dimensionality of the parameters. Alternatively, in the setting of nonlinear inverse problems with additive noise, e.g., $Y = G(x) + \mathcal{E}$,~\cite{omar_oed_nn} constructs deep neural network approximations of the deterministic parameter-to-observable map $G$ after performing dimension reduction in both the input space (of parameters $X$) and output space (of observations $Y$). The resulting surrogate is then used in an NMC estimator of the EIG. 
%
We employ the same parameter-space projection here, but a different projection of the observations that yields sharper error bounds. Moreover, instead of using NMC, we use both dimension reduction methods in transport-based EIG estimators. 


\subsection{Contributions} 
We summarize our main contributions as follows.
We first present a general formulation for EIG estimation using transport-based density estimates. Our formulation is applicable to each of the three settings identified above, including \revise{implicit Bayesian models~\cite{cranmer,Kleinegesse2019EfficientBE} and goal-oriented variants of the OED problem~\cite{li2019combinatorial, chakraborty2024likelihoodfreeapproachgoalorientedbayesian}}. This formulation encompasses \textit{any} method for transport-based marginal and conditional density estimation, including conditional normalizing flows~\cite{condNFs} or conditional optimal transport~\cite{wang2024COT}, though we will demonstrate it here using an adaptive representation of monotone triangular maps that provides certain approximation guarantees~\cite{baptista2020adaptive}.




As our first technical contribution, we analyze the underlying problem of sample allocation and convergence rates for such EIG estimators. Our results apply to any parametric density estimation scheme used within the EIG estimation framework. We develop an asymptotically optimal allocation of samples, between density estimation and the estimation of the outer expectation, to balance the decay rates of bias and variance and thus maximize the convergence rate of the mean squared error (MSE) of the EIG estimate. With this allocation, we show that this class of estimators converges at a \textit{faster} rate than an optimally-balanced NMC scheme.  
We numerically verify the conditions yielding optimal rates, in a simple linear-Gaussian setting where the exact EIG is available analytically. As a demonstration, we then evaluate the performance and versatility of density-based EIG estimation algorithms for a non-Gaussian problem with various EIG objectives.

Our second technical contribution is to combine dimension reduction with transport-based EIG estimation to construct estimators of EIG in high dimensions. Our proposed dimension reduction scheme relies on minimizing an \textit{upper bound} for the error in EIG induced by projecting both the parameters $X$ and observations $Y$ to lower-dimensional subspaces. The scheme for computing these projections originates in \cite{Ricarod_dimRed}, where it was used for the purpose of posterior approximation, but this information theoretic perspective on dimension reduction is a natural fit for EIG in non-Gaussian settings. It seeks $(r \ll n_x)$-dimensional and $(s \ll n_y)$-dimensional projections $X_r$ and $Y_s$ of $X$ and $Y$, respectively, that minimize an upper bound for the information loss $\mathcal{I}(X; Y) - \mathcal{I}(X_r ; Y_s)$ at any $r$ and $s$. The preceding EIG estimation schemes can then be applied to the joint distribution of $X_r$ and $Y_s$. 
We show that this dimension reduction approach outperforms other linear approaches, such as principal component analysis (PCA) and canonical correlation analysis (CCA). We also demonstrate how this scheme allows one to \textit{trade off dimensionality with the modeling of non-Gaussianity}: in many EIG estimation problems it is better to truncate the representation of the parameters and observations in a suitable basis and capture the non-Gaussianity of the resulting $\pi_{X_r, Y_s}$, than it is to make cruder (e.g., Gaussian) approximations of the full distribution $\pi_{X,Y}$. 

The remainder of the paper is structured as follows. In Section \ref{sec:EIG_est}, we first present a general formulation of transport-based EIG estimation, where block-triangular transport maps are used to approximate the relevant marginal and/or conditional densities in Bayesian problem settings ranging from standard to simulation-based. We also show how maximum likelihood estimation of block-triangular transport maps corresponds to tightening the bounds on EIG in certain cases. In Section~\ref{sec:theory}, we present our results on optimal sample allocations and convergence rates. In Section~\ref{sec:high_dim}, we describe the construction of parameter-space and observation-space projections that minimize the loss of \revise{MI}, and show how these projections are applied to the estimation of EIG in high dimensions. Section \ref{sec:num_res} contains our numerical experiments, where we verify our results on optimal sample allocation using a linear Gaussian example, demonstrate EIG estimation in non-Gaussian settings, and explore the impact of dimension reduction and density approximations on EIG estimates in a high-dimensional PDE-constrained Bayesian inverse problem.

\section{EIG estimation using transportation of measure} \label{sec:EIG_est}

In this section, we first formulate upper and lower bounds for EIG based on density approximations 
in Section~\ref{subsec:bounds}. Second, we show how transport-based density estimates can be used to realize these bounds and elucidate the link between maximum likelihood estimation and tightness of the bounds in Section~\ref{subsec:TM}. Lastly, we analyze the asymptotic convergence rate of this class of EIG estimators and derive an asymptotically-optimal rule for allocating samples between density estimation and the estimation of an outer prior expectation in Section~\ref{sec:theory}. 


\subsection{EIG bounds and estimators} \label{subsec:bounds}
Recall the expressions~\eqref{eq:EIGpost} and~\eqref{eq:EIGmarg} for the EIG from prior to posterior: the first involves the normalized posterior density $\pi_{X|Y}$ and the second involves the evidence $\pi_Y$. In most settings, neither of these densities can be evaluated exactly. Instead we can only evaluate approximations, denoted by $\widetilde \pi_{X|Y}$ and $\widetilde \pi_Y$, respectively. \emph{Any} such approximating densities yield lower and upper bounds, respectively, for the EIG:
\begin{align}
  \mathbb E_{\pi_{X,Y}}\left[\log \frac{\widetilde \pi_{X|Y}( X | Y)}{\pi_X(X)} \right] 
     &\leq \mathbb E_{\pi_{Y}} \left[ \Dkl(\pi_{X|Y}||\widetilde \pi_{X|Y} )\right] + \mathbb E_{\pi_{X,Y}}\left[\log \frac{\widetilde \pi_{X|Y}( X | Y)}{\pi_X( X)}\right]  \nonumber \\ 
     & =\mathrm{EIG} \nonumber\\
     & =\mathbb E_{\pi_{X,Y}}\left[ \log\frac{\pi_{Y|X}( Y | X)}{\widetilde \pi_Y( Y )}\right] - \Dkl\left(\pi_Y||\widetilde \pi_Y\right)  \leq  \mathbb E_{\pi_{X,Y}} \left[\log \frac{\pi_{Y|X}( Y | X)}{\widetilde \pi_Y( Y)} \right]   \label{bounds_EIG}
\end{align}
We refer readers to~\cite{Foster2019VariationalBO,Poole2019OnVB} for the derivation. We note that equality holds in \eqref{bounds_EIG} if and only if the relevant approximation is exact, i.e., if $\widetilde \pi_{X|Y} = \pi_{X| Y}$ or if $\widetilde \pi_Y = \pi_Y$, and that the deficits in the inequalities are precisely $\mathbb E_{\pi_Y}\left[\Dkl( \pi_{X|Y} || \widetilde{\pi}_{X|Y}) \right]$ and $\Dkl(\pi_Y || \widetilde{\pi}_Y)$. Therefore, approximating the EIG by maximizing the lower bound or minimizing the upper bound is equivalent to finding good approximations in this Kullback--Leibler sense. 
Given the appropriate density approximations and $M$ i.i.d.\ joint samples $\{x^i, y^i\}_{i=1}^M \sim \pi_{X,Y}$, we can also construct Monte Carlo estimators of the lower or upper bounds that satisfy the inequalities in \eqref{bounds_EIG} with high probability; we will define such estimators explicitly below.


If we replace the prior or likelihood in equations~\eqref{eq:EIGpost} or~\eqref{eq:EIGmarg} with approximations $\widetilde{\pi}_X$ and $\widetilde{\pi}_{Y|X}$, the expected information gain (EIG) is neither upper-bounded nor lower-bounded \cite{Foster2019VariationalBO}. Specifically, neither $\mathbb{E}_{\pi_{X,Y}} \left[\log \frac{\widetilde{\pi}_{X|Y}(x|y)}{\widetilde{\pi}_X(x)} \right]$ nor $\mathbb{E}_{\pi_{X,Y}} \left[\log \frac{\widetilde{\pi}_{Y|X}(y|x)}{\widetilde{\pi}_Y(y)} \right]$ serves as a bound for EIG. Furthermore, in cases where both the prior \textit{and} likelihood are intractable, it becomes necessary to approximate both the numerator and denominator densities. This requires using Monte Carlo estimators for $\mathbb{E}_{\pi_{X,Y}} \left[\log \frac{\widetilde{\pi}_{X|Y}(x|y)}{\widetilde{\pi}_X(x)} \right]$ and $\mathbb{E}_{\pi_{X,Y}} \left[\log \frac{\widetilde{\pi}_{Y|X}(y|x)}{\widetilde{\pi}_Y(y)} \right]$. While these estimators can be consistent, they do not provide (stochastic) bounds for EIG.

To summarize, evaluating EIG according to the expressions above involves: (i) approximating a conditional and/or a marginal density, depending on what problem information is \textit{a priori} available (e.g., samples versus the ability to evaluate exact densities); and (ii) using samples $\{x^i, y^i\}_{i=1}^M$ from $\pi_{X,Y}$ to approximate the outer expectation. We then have four estimators that are useful in practice:
\begin{align}
    \widetilde{\mathrm{EIG}}_{\mathrm{m}} &\coloneqq \frac 1M \sum_{i = 1}^M\log \frac{\pi_{Y|X}( y^i| x^i)}{\widetilde \pi_Y( y^i)},  \label{eig1}\\
    \widetilde{\mathrm{EIG}}_{\mathrm{pos}} &\coloneqq \frac 1M \sum_{i = 1}^M\log \frac{\widetilde \pi_{X|Y}( x^i| y^i)}{\pi_X( x^i)},  \label{eig2}\\
   \widetilde{\mathrm{EIG}}_{\mathrm{lik}} &\coloneqq \frac 1M \sum_{i = 1}^M\log \frac{\widetilde \pi_{Y|X}( y^i| x^i)}{\widetilde \pi_Y( y^i)},  \label{eig3}\\
    \widetilde{\mathrm{EIG}}_{\mathrm{pr}} &\coloneqq \frac 1M \sum_{i = 1}^M\log \frac{\widetilde \pi_{X|Y}( x^i| y^i)}{\widetilde \pi_X( x^i)}.  \label{eig4}
\end{align}
The first two provide stochastic upper and lower bounds as described in~\eqref{bounds_EIG}. The last three can be used in the likelihood-free setting, as the expressions do not involve evaluating the likelihood function $\pi_{Y|X}$; the second, however, does require evaluating the prior density. \revise{To be precise, all four estimators are applicable under Case~\ref{assump1}. Under Case~\ref{assump2}, the estimators $\widetilde{\mathrm{EIG}}_{\mathrm{pos}}$, $\widetilde{\mathrm{EIG}}_{\mathrm{lik}}$, and $\widetilde{\mathrm{EIG}}_{\mathrm{pr}}$ remain applicable, whereas under Case~\ref{assump3}, only $\widetilde{\mathrm{EIG}}_{\mathrm{lik}}$ and $\widetilde{\mathrm{EIG}}_{\mathrm{pr}}$ are applicable. Note that NMC is only applicable in Case~\ref{assump1}.} In the next section, we will introduce an approach to compute these four EIG estimators using transport maps.

\subsection{Transport maps for density estimation} \label{subsec:TM}
Measure transport provides an expressive and infinitely refinable way of approximating complex probability densities. Here we give a brief and intuitive introduction to this approach; for more details, we refer the reader to \cite{Marzouk2016,spantini2018jmlr,baptista2020adaptive}. To illustrate key ideas, we first set aside the OED problem and the specific probability distributions it entails, returning to them in Section~\ref{subsubsec:marg_cond_est}.

Let $\pi$ denote an intractable target distribution that we would like to characterize, and let $\eta$ denote a reference distribution whose density we can easily evaluate---for example, a standard Gaussian. We would like to find a map $S\colon \mathbb{R}^{n} \rightarrow \mathbb{R}^n$ such that $\eta(A) = \pi(S^{-1}(A))$ for any measurable set $A$. 
A map that satisfies this property is said to \emph{push forward} $\pi$ to $\eta$, or equivalently to \emph{pull back} $\eta$ to $\pi$. The existence of $S$ is guaranteed when both $\pi$ and $\eta$ are absolutely continuous with respect to Lebesgue measure~\cite{Santambrogio}. 
We use subscript and superscript $\sharp$ symbols to denote the pushforward and pullback operations, respectively:
\begin{align*}
    & S_{\sharp} \pi = \eta \\
    & S^{\sharp} \eta= \pi.
\end{align*}
When $S$ is invertible and sufficiently smooth, we can write the target density $\pi$ explicitly in terms of the reference density using the change-of-variables formula\footnote{Here, since all probability measures are assumed to have Lebesgue densities, we abuse notation by using the same symbol to denote either object.}
\begin{align} 
    \pi(z) = S^{\sharp} \eta( z) = \eta \circ S( z) \left| \det \nabla S(z) \right|. \label{pullback}
\end{align}
This expression lets us estimate the density $\pi$ simply by plugging in an estimate for the map $S$~\cite{wang2022minimax}. Moreover, the map allows us to sample from the target distribution by evaluating the inverse map $S^{-1}(z^i)$ at reference samples $z^i \sim \eta$.

For absolutely continuous measures $\eta$ and $\pi$, there are infinitely many transport maps coupling one to the other. For the purpose of EIG estimation, we are particularly interested in marginal and conditional density estimation. A triangular transport map, known as the Knothe--Rosenblatt (KR) rearrangement~\cite{Villani2008OptimalTO,Santambrogio, bogachev}, is particularly well suited to this goal. The KR rearrangement is the unique\footnote{Up to re-labeling/re-ordering of the coordinates.} map $S$ satisfying $S_\sharp \pi = \eta$ that has the \textit{lower-triangular} structure
\begin{align}
    S( z) = 
    \begin{bmatrix*}[l]
    S^1(z_1)\\
    S^2(z_1, z_2)\\
    \vdots \\
    S^n(z_1, z_2, \ldots, z_n)\\
    \end{bmatrix*}, \label{triangular_map}
\end{align}
and where each component function $S^k: \R^k \to \R$ is monotone in its last argument, i.e., $z_k \mapsto S^k(z_1, \ldots, z_k)$ is monotone (increasing) for all $(z_1, \ldots, z_{k-1}) \in \R^{k-1}$, $k=1,\ldots,n$. Component functions of the triangular map correspond to conditionals of $\pi$ \cite{Santambrogio}, and thus this map immediately provides conditional density estimates using an adaptation of~\eqref{pullback}, which we will detail below. Strict monotonicity guarantees that the map $S$ is invertible, and evaluations of the inverse map $S^{-1}$ essentially involve a sequence of $n$ univariate root-finding problems. Moreover, the Jacobian determinant of $S$ is easy to evaluate as $\det \nabla S( z) = \prod_{k = 1}^n \partial_{k} S^k(z_{1:k})$. 

Although we have the freedom to choose $\eta$ to be any distribution on $\R^n$, putting $\eta$ equal to the standard normal distribution $\Normal(0, I_n)$ substantially simplifies the problem of estimating the map $S$, which we discuss next.




\subsubsection{Maximum likelihood estimation of transport maps} \label{subsubsec:TM_opt}
Here we formulate an optimization method for identifying lower triangular transport maps $S$, given a set of samples drawn from $\pi$. One approach is to minimize the KL divergence from $S^\sharp\eta$ to $\pi$ over a class of monotone triangular functions $\mathcal S$~\cite{Marzouk2016}, i.e., 
\begin{align}
    S_{\OPT} = \argmin_{S \in \mathcal S} \Dkl(\pi || S^\sharp \eta). \label{opt_KL}
\end{align}
Choosing $\eta$ to be the standard Gaussian distribution, we can then write the objective in~\eqref{opt_KL} as 
\begin{align*}
    \Dkl(\pi||S^{\sharp}\eta) &= \mathbb E_{z \sim \pi}\left[ \log  \frac{\pi(z)}{S^{\sharp} \eta(z)}\right]\\
     &= \mathbb E_\pi \left[\log \pi \right] + \frac n2 \log(2\pi) + \sum_{k = 1}^n \mathbb E_{z \sim \pi} \left[\frac 12 S^k(z_{1:k})^2 - \log \partial_k S^k( z_{1:k}) \right].
\end{align*}
Since the first two terms are independent of $S$, and each term in the summation depends only on $S^k$, the optimization problem is separable; we can learn each map component independently by minimizing the objective
\begin{align}
    J_k(S^k) \coloneqq \mathbb E_\pi \left[\frac 12 S^k( z_{1:k})^2 - \log \partial_k S^k( z_{1:k}) \right], \label{opt_prob1}
\end{align} 
for $k=1, \ldots, n$. Note that $J_k(S^k)$ is convex in $S^k$ for every $k$. Therefore the optimization problem \eqref{opt_KL}, subject to the linear constraints $\partial_k S^k( z_{1:k}) > 0$ for $k=1, \ldots, n$, is a convex problem. In practice, we replace the expectation in \eqref{opt_prob1} with its empirical approximation,
\begin{align}
    \widehat J_k(S) = \frac 1N \sum_{i = 1}^N \left[\frac 12 S( z^i_{1:k})^2 - \log  \partial_k S( z^i_{1:k}) \right]  \label{opt_prob2},
\end{align}
where $\{z^i\}_{i=1}^N$ are i.i.d.\ samples from $\pi$. Let $\widehat S^k_{\OPT} =  \argmin_{S^k\in \mathcal S^k} \widehat J_k(S^k)$, where $\mathcal{S}^k$ is a space of maps from $\R^k$ to $\R$ satisfying the monotonicity constraint $\partial_k S^k >0$. As explained in \cite{baptista2020adaptive,wang2022minimax}, each $\widehat S^k_{\OPT}$ is a \textit{maximum likelihood} estimate of the corresponding map component $S^k$, and concatenating these functions yields a maximum likelihood estimate of the entire triangular map: $\widehat{S}_{\OPT} = (\widehat S^1_{\OPT}, \widehat S^2_{\OPT}, \ldots, \widehat S^n_{\OPT})$. 

We can then form a plug-in estimate of the density $\pi$: $\widehat \pi( z) \coloneqq \widehat S^\sharp_{\OPT} \eta( z)$. Thus $\widehat \pi$ is a transport-induced estimate of the target density, given $N$ samples.


\subsubsection{Parameterization and the choice of basis functions}\label{subsubsec:ATM}

Solving the maximum likelihood estimation problem above requires parameterizing classes of monotone maps; specifically, we must explicitly construct classes $\mathcal{S}^k$ of map components $S^k \colon \R^k \to \R$ that are strictly increasing in their last argument, i.e., satisfying $\partial_k S^k(x_{1:k}) > 0, \ \forall x_{1:k} \in \R^k$. Here we do so by expressing $S^k$ as the transformation of an arbitrary (non-monotone) function $f \colon \mathbb{R}^k \rightarrow \mathbb{R}$ through the operator $\mathcal{R}_k$:
\begin{equation}
    S^k(z_{1:k}) = \mathcal{R}_k(f)( z_{1:k}) \coloneqq f( z_{1:k-1},0) + \int_0^{ z_k} g(\partial_k f( z_{1:k-1},t)) dt,
\end{equation}
where $g: \R \to \R_{> 0}$ is strictly positive and bijective, and satisfies further regularity conditions detailed in~\cite{baptista2020adaptive}. We
substitute this representation of each $S^k$ into \eqref{opt_prob2} and solve the resulting \textit{unconstrained} minimization problem, $\min_f \widehat{J}_k(\mathcal{R}_k(f)) $ over some linear space of functions $\mathcal{V}_k \ni f$. As explained in~\cite{baptista2020adaptive}, this construction provides important optimization guarantees: there are no spurious local minima (every local minimizer $f^\ast$ is global minimizer) and under certain tail conditions on $\pi$ and $f$, we provably recover the true KR rearrangement.

Here, we parameterize $f$ using a linear expansion of tensor-product Hermite polynomials $\psi_{\alpha}$ as $f( z_{1:k}) = \sum_{\alpha} c_{\alpha} \psi_{\alpha}( z_{1:k})$ and solve for its coefficients $c_{\alpha}$. We refine this polynomial space using a greedy adaptive procedure and apply cross-validation
to select the total number of polynomial basis functions to retain. We refer to~\cite{baptista2020adaptive} for complete details on these numerical procedures. 

We note also that while our focus thus far in Section~\ref{subsec:TM} and Section~\ref{subsubsec:ATM}, and in the numerical experiments reported later in this paper, is on \textit{strictly} triangular maps, we will see in Section~\ref{subsubsec:marg_cond_est} that \textit{block triangular} maps as described in \cite{baptista2024conditional} are also suitable for estimating the marginal and conditional densities relevant to the OED problem. In contrast with the strictly triangular case, there is no unique choice of block-triangular map. One natural choice is a conditional Brenier map, as described in \cite{carlier_brenier,baptista2024conditional,wang2024COT}. Another broad family of choices are given by conditional normalizing flows, which are specific parameterizations of the bottom block of a block-triangular map. We use strictly triangular maps in our numerical experiments because the associated computations are relatively fast, and because they enjoy the optimization guarantees mentioned above; to our knowledge, such guarantees are not available for other continuous map representations.

\subsubsection{From marginal and conditional density estimation to EIG bounds}\label{subsubsec:marg_cond_est}
In the previous discussion, we let $z$ be a generic vector in $\mathbb R^n$. Now let $z = (y,  x)$ where $ x\in \mathbb R^{n_x}$ and $ y \in \mathbb R^{n_y}$. Given samples from the joint distribution $\pi_{X,Y}$, we consider the problem of estimating the densities $\pi_Y$ and $\pi_{X|Y}$. To do so, we construct a \textit{block lower-triangular} map $S: \mathbb R^{n_y + n_x} \rightarrow \mathbb R^{n_y + n_x}$ of the form: 
\begin{align*}
     S(y, x) = \begin{bmatrix*}[l]
    S^{\mathcal Y}(y)\\
    S^{\mathcal X}(y, x)\\
    \end{bmatrix*},
\end{align*}
where $S^{\mathcal Y}\colon \R^{n_y} \rightarrow \R^{n_y}$ and $S^{\mathcal X} \colon \R^{n_y + n_x} \to  \R^{n_x}$. We choose $\eta$ to be standard Gaussian on $\R^{n_y + n_x}$. Our parameterizations ensure that $S^\mathcal{Y}$ is invertible and that $\nabla_y S^\mathcal{Y}$ exists, and similarly that $x \mapsto S^\mathcal{X}(y, x)$ is invertible and that $\nabla_x S^\mathcal{X}(y, x)$ exists for every $y \in \R^{n_y}$; this applies specifically to the numerical representations of triangular maps mentioned above. 
Note that such properties are satisfied almost everywhere by the KR map for absolutely continuous distributions with full support; see \cite[Sec.\ 3]{spantini2018jmlr} and \cite[Sec.\ 2.3]{Santambrogio} for a discussion.

If we select our reference distribution $\eta$ to be \textit{any} distribution whose density factorizes as $\eta(x,y) = \eta_X(x)\eta_Y(y)$ (which clearly includes the standard normal), \cite[Theorem 2.4]{baptista2024conditional} shows that we can express the marginal density of $Y$ as $\pi_Y = (S^{\mathcal Y})^{\sharp} \eta_Y$ and the posterior density of $X$ as $\pi_{X|Y=y} = S^{\mathcal X}( y, \cdot)^\sharp \eta_X$ for any $y \in \R^{n_y}$. 


By reversing the ordering of $x$ and $y$, we obtain another lower-triangular map 
\begin{align*}
     U(x, y) = 
    \begin{bmatrix*}[l]
    U^{\mathcal X}( x)\\
    U^{\mathcal Y}( x, y)\\
    \end{bmatrix*}.
\end{align*}
We can then exchange the roles of $y$ and $x$ and repeat the procedure above to approximate the prior $\pi_X$ and the likelihood $\pi_{Y|X}$. Thus, the two maps with both variable orderings provide all four density estimates, $\widehat{\pi}_X$, $\widehat{\pi}_{X|Y}$, $\widehat{\pi}_Y$, and $\widehat{\pi}_{Y|X}$, as follows:
\begin{align*}
    \widehat \pi_X &= (\widehat U^{\mathcal X})^\sharp \eta_X, \\
    \widehat \pi_{X|Y = y} &= \bigl(\widehat S^{\mathcal X}( y,\cdot)\bigr)^\sharp \eta_X, \\
     \widehat \pi_Y &= (\widehat S^{\mathcal Y})^\sharp \eta_Y, \\
     \widehat \pi_{Y|X = x} &= \bigl(\widehat U^{\mathcal Y}( x,\cdot) \bigr)^\sharp \eta_Y.
\end{align*}
where we use $\widehat{S}$ and $\widehat{U}$ to denote maximum likelihood estimates of $S$ and $U$. These density estimates can then be substituted into any of \eqref{eig1}--\eqref{eig4}. We summarize the proposed EIG estimation scheme in Algorithm~\ref{EIG_alg}.

\begin{algorithm}[!ht]
\caption{EIG estimation using transport maps}
\label{EIG_alg}
\begin{algorithmic}[1]
\STATE{\textbf{Input}: Partition the i.i.d.\ pairs of joint samples $\{(x^i,y^i)\}_{i=1}^L \sim \pi_{X,Y}$ into $N$ training and $M$ evaluation samples, such that $L = N+M$}
\STATE{\textbf{Output}: EIG estimator}
\STATE{Use $N$ training samples to learn the transport maps $\widehat{S}$ and/or $\widehat{U}$, and form plug-in estimates of densities that are not available in closed form, e.g.,  $\widehat{\pi}_Y$, $\widehat{\pi}_{Y|X}$, $\widehat{\pi}_X$, and/or $\widehat{\pi}_{X|Y}$.}
\STATE{Use $M$ pairs of samples to estimate EIG by evaluating the empirical averages $\frac 1M \sum_{i = 1}^M \log \frac{\pi_{Y|X}( y^i| x^i)}{\widehat \pi_Y( y^i)} $, $\frac 1M \sum_{i = 1}^M \log \frac{\widehat \pi_{X|Y}( x^i| y^i)}{\pi_X( x^i)} $, $\frac 1M \sum_{i = 1}^M  \log \frac{\widehat \pi_{Y|X}( y^i| x^i)}{\widehat \pi_Y( y^i)} $, and/or $\frac 1M \sum_{i = 1}^M \log \frac{\widehat \pi_{X|Y}( x^i| y^i)}{\widehat \pi_X( x^i)} $. }
\end{algorithmic}
\end{algorithm}

Now we make the crucial observation that solving the optimization problem \eqref{opt_KL} is equivalent to tightening the bounds in \eqref{bounds_EIG}. Observe that the objective of \eqref{opt_KL} separates:
\begin{align*}
    \Dkl\left( \pi_{Y,X}||S^\sharp \eta_{Y,X}\right) 
     = & \mathbb E_{\pi_{Y,X}} \left[ \log \frac{\pi_{Y,X}(Y,X)}{ (S^\sharp \eta_{Y,X} )(Y,X)}\right]\\
     = & \mathbb E_{\pi_{Y}} \left[ \log \frac{\pi_Y(Y)}{ (S^{{\mathcal{Y}, \sharp }}\eta_Y )(Y)}\right] + 
     \mathbb E_{\pi_{Y,X}}\left[ \log \frac{\pi_{X|Y}(X|Y)}{\left(S^{\mathcal{X}} (Y, \cdot)^\sharp\eta_X \right)(X)}\right]\\
     = & \Dkl\left(\pi_Y || S^{{\mathcal{Y} \sharp }}\eta_Y \right)
     + \mathbb E_{\pi_{Y}} \left[ \Dkl \left(\pi_{X|Y}||S^{\mathcal{X}} (y, \cdot)^\sharp\eta_X \right) \right].
\end{align*}
Therefore, finding $\argmin_{S \in \mathcal{S}} \Dkl\left( \pi_{Y,X}||S^\sharp \eta_{Y,X}\right)$ is equivalent to solving both
\begin{align*}
    \min_{S^\mathcal{X} \in \mathcal{S}^\mathcal{X}}\mathbb E_{\pi_{Y}} \left[ \Dkl \left(\pi_{X|Y}||S^{\mathcal{X}} (y, \cdot)^\sharp\eta_X \right) \right]
\ \ \text{and} \ \     \min_{S^\mathcal Y \in \mathcal{S}^\mathcal{Y}} \Dkl\left(\pi_Y || S^{{\mathcal{Y}}^{\sharp }}\eta_Y \right),
\end{align*}
where $\mathcal{S} = \mathcal{S}^\mathcal{Y} \times \mathcal{S}^\mathcal{X} $. As noted in~\ref{subsec:bounds}, these two terms are \textit{precisely} the deficits of the upper and lower bounds in~\eqref{bounds_EIG}. Therefore, given a particular parameterization of the transport map, minimizing these two quantities is equivalent to maximizing and minimizing, respectively, the lower and upper bounds in~\eqref{bounds_EIG}. Maximum likelihood estimation of the maps $S^\mathcal{X}$ and $S^\mathcal{Y}$ is a sample approximation of these bound-tightening problems. 

\subsection{Convergence analysis and optimal sample allocation}
\label{sec:theory}

Given that the transport maps are estimated using $N$ training samples (see~\eqref{opt_prob2}), and the empirical average is computed using $M$ evaluation samples (see~\eqref{eig1},~\eqref{eig2},~\eqref{eig3},~\eqref{eig4}), we aim to investigate the optimal sample allocation between the number of training samples $N$ and the number of evaluation samples $M$, given a total of $L = N+M$ samples, with the objective of minimizing the MSE of the EIG estimator. Here, we use $\widehat{\mathrm{EIG}}_{M,N}$ to indicate that the EIG estimator is obtained using $N$ training samples and $M$ evaluation samples, \revise{which are assumed to be independent}. For an EIG estimator $\widehat{\mathrm{EIG}}_{M,N}$, the bias is defined as 
\begin{align*}
    \mathrm{bias} = \mathbb E \left[ \widehat{\mathrm{EIG}}_{M,N} - \mathrm{EIG}\right],
\end{align*}
where the expectation is taken over all the non-deterministic terms: the random map obtained through $N$ training samples and Monte Carlo sum calculated using $M$ evaluation samples. We denote the variance of the EIG estimator by 
\begin{align*}
    \mathrm{variance} = \mathbb{V}\left[ \widehat{\mathrm{EIG}}_{M,N}\right].
\end{align*}
Recall that the MSE is defined as 
\begin{align*}
    \mathrm{MSE} = \mathrm{bias}^2 + \mathrm{variance}.
\end{align*}
Now consider the transport map $S$ parameterized by coefficients $\alpha \in \R^p$, with $\alpha^\ast$ being the optimizer that minimizes the exact KL divergence. That is,
\begin{align}
    \alpha^* &= \argmin_\alpha \Dkl \left(\pi || S^\sharp_\alpha \eta \right) = \argmax_\alpha \mathbb E_\pi\left[ \log S_\alpha^\sharp \eta(z)\right],
    \label{eq:bestapprox}
\end{align}
where $\pi$ is a generic probability measure. 

To simplify the exposition, we first consider estimating only the marginal distribution $\pi_Y$. Let $\pi_Y(y)$ denote the true density evaluated at $y$. Let 
\begin{align}\label{eq:parametric_Qy}
    \mathcal{Q}^{\mathcal Y} = \{q_Y( \cdot \, ; \alpha): \R^{n_y} \to \R_+, \  \alpha \in \R^p \} 
\end{align} 
be a parametric class of densities within which we will approximate $\pi_Y$. \fl{For example, in the specific case of transport-based approximations, an element $q_Y \in Q^{\mathcal Y}$ can be written as the pullback of $\eta_Y$ by $S^{\mathcal Y}_\alpha$, i.e., $q_Y( y; \alpha) = ( (S_\alpha^{\mathcal{Y}})^\sharp \eta_Y)(y)$, where $S_\alpha^{\mathcal Y}$ belongs to the parametric class of transport maps $\mathcal S = \{S_\alpha: \R^{n_y} \to \R^{n_y}, \  \alpha \in \mathbb R^p\}$ and $\eta_Y$ is the standard normal distribution.}
The results below, however, apply to any class of densities parameterized by $\alpha$.

Given i.i.d.\ samples $\{y^i\}_{i=1}^N\sim \pi_Y$, the maximum likelihood estimator (MLE) of $\alpha$ follows from simply replacing the expectation in \eqref{eq:bestapprox} by its empirical approximation, 
\begin{align}
\label{eqn:alphahat}
    \widehat{\alpha}_N = \argmax_{\alpha \in \R^p} \frac 1N\sum_{i = 1}^N   \log q_Y( y^i; \alpha),
\end{align}
where the subscript $N$ indicates that $\widehat{\alpha}_N$ is learned using $N$ training samples. \revise{Next we make some assumptions about this estimation problem.}

\begin{assumption}\label{ass:inclass} The target density $\pi_Y$ is in-class for the parametric model, i.e., there exists a parameter $\alpha^*$ such that $q_Y(y ; \alpha^*) = \pi_Y(y)$.
\end{assumption}
\revise{Assumption~\ref{ass:inclass} may appear rather strong, but we note that the parametric families of densities $\mathcal{Q}^{\mathcal Y}$ used in practice (e.g., involving triangular transport maps or normalizing flows) are highly expressive and can effectively capture the true distribution as $p$ becomes large. Below we remark further on issues of approximation error (asymptotic bias) in the density estimate.}


%
\begin{assumption} \label{ass:consist_and_normal} 
The maximum likelihood estimator $\widehat{\alpha}_N$ is consistent and asymptotically normal. That is $\widehat{\alpha}_N$ converges to $ \alpha^*$ in probability as $N \rightarrow \infty$ and $\sqrt N \left(  \widehat{\alpha}_N -  \alpha^* \right)$ converges to $\mathcal N(0, I( \alpha^*)^{-1})$ in distribution, where the entries of $I(\alpha)$ are defined as
\begin{align}\label{eq:fisher}
    I(\alpha)_{ij} = -\mathbb E_{\pi_Y}\left[ \partial_{\alpha_i}\partial_{\alpha_j}\log q_Y(y;\alpha)\right].
\end{align}
\end{assumption}
We refer the reader to Section 10.6.2 of~\cite{stat_inf} for the regularity conditions that guarantee the consistency and asymptotic normality of the MLE.
We now present our main result on the bias and variance of an resulting EIG estimator constructed using an approximate density based on transport maps. 
The proof is given in Appendix~\ref{proof_of_thms}. 

\begin{theorem}\label{thm:EIG_rate}
    Under Assumptions~\ref{ass:inclass} and~\ref{ass:consist_and_normal}, let $g( \alpha) = \mathbb E_{\pi_Y} \left [ \log q_Y( y; \alpha)\right]$, where $q_Y\in \mathcal Q^{\mathcal Y}$.
    Suppose $g$ is continuous and bounded, and that both $\nabla g(\alpha), \nabla^2 g( \alpha)$ exist and are continuous. Then, the EIG estimator 
    $$\frac{1}{M}\sum_{j = 1}^M \log \frac{\pi_{Y|X}( y^j| x^j)}{q_Y( y^j; \widehat{\alpha}_N )},$$
     \revise{where $\{x^i, y^i\}_{i=1}^{N}$ and $\{x^j, y^j\}_{j=1}^{M}$ are drawn i.i.d.\ from $\pi_{X,Y}$,}
    has bias of $\mathcal O(1/N)$ and variance of $\mathcal O(1/M + 1/N^2)$. The MSE of this estimator thus converges to zero at a rate of $\mathcal O(1/M + 1/N^2)$.
\end{theorem}

\revise{Next, we can develop analogous results for EIG estimators involving conditional density estimation.} Define the family of densities
\begin{align*}
    \mathcal Q^{\mathcal Y|X} = \{q_{Y|X}(\cdot \vert x; \beta): \R^{n_y} \to \R_+; \ x \in \mathbb{R}^{n_x},  \beta \in \mathbb R^{t} \},
\end{align*}
where $\beta$ are the coefficients in a parameterization of the likelihood function. In our case, $q_{Y|X} \in \mathcal Q^{\mathcal Y|X}$ is represented as the pullback of $\eta_Y$ by a transport map $U^{\mathcal Y}_\beta(x, \cdot)$. That is, $q_{Y|X}( y| x; \beta) = \bigl(U_\beta^{\mathcal Y}( x,\cdot)^\sharp \eta_Y\bigr) (y)$, where $U_\beta^{\mathcal Y}( x,\cdot)$ is an element in the parametric class of $x$-dependent transport maps $\mathcal U = \{U_\beta(x,\cdot): \R^{n_y} \to \R^{n_y}; \ x \in \mathbb{R}^{n_x}, \beta\in\mathbb R^t \}$. } 
Now let $\widehat \beta_N$ be the MLE of $\beta$,
\begin{align}
\label{eqn:betahat}
    \widehat \beta_N = \argmax_{\beta \in \R^t} \frac 1N \sum_{i=1}^N \log q_{Y|X}(y^i|x^i;\beta).
\end{align}
Analogous to Assumptions~\ref{ass:inclass} and~\ref{ass:consist_and_normal}, we assume the following for estimation over $\mathcal{Q}^{\mathcal{Y}|\mathcal{X}}$. 
%
\begin{assumption}\label{ass:inclass_beta} The target density $\pi_{Y|X}$ is in-class for the parametric model, i.e., there exists a parameter $\beta^*$ such that $q_{Y|X}(y|x ; \beta^*) = \pi_{Y|X}(y|x)$.
\end{assumption}
\begin{assumption} \label{ass:consist_and_normal_beta} 
The maximum likelihood estimator $\widehat{\beta}_N$ is consistent and asymptotically normal. 
\end{assumption}

\begin{theorem}\label{thm:EIG_rate_2}
    Under Assumptions~\ref{ass:inclass_beta} and~\ref{ass:consist_and_normal_beta}, let $h(\beta) = \mathbb E_{\pi_{X,Y}} \left [ \log q_{Y|X}( y|  x; \beta) \right]$, where $q_{Y|X}$ $\in \mathcal Q^{\mathcal {Y|X}}$.
    Suppose $h$ is continuous and bounded and that both $\nabla h(\beta), \nabla^2 h( \beta)$ exist and are continuous. Then, the EIG estimator 
    $\frac{1}{M}\sum_{j = 1}^M \log \frac{q_{Y|X}( y^j | x^j; \widehat{\beta}_N)}{\pi_Y( y^j)}$,
     \revise{where $\{x^i, y^i\}_{i=1}^{N}$ and $\{x^j, y^j\}_{j=1}^{M}$ are drawn i.i.d.\ from $\pi_{X,Y}$,}
    has bias of $\mathcal O(1/N)$ and variance of $\mathcal O(1/M + 1/N^2)$. The MSE of this estimator thus converges to zero at a rate of $\mathcal O(1/M + 1/N^2)$.
\end{theorem}

Building on the two preceding theorems, we can obtain an similar result for EIG estimators that employ simultaneous marginal and conditional density estimates. The proof is deferred to Appendix~\ref{proof_of_thms}. 
\begin{theorem}\label{cor:MSE_rate}
    Under Assumptions~\ref{ass:inclass},~\ref{ass:consist_and_normal},~\ref{ass:inclass_beta} and~\ref{ass:consist_and_normal_beta}, suppose that $g$ and $h$ are both continuous and bounded, and that $\nabla g(\alpha), \nabla^2 g( \alpha)$ and $\nabla h(\beta), \nabla^2 h( \beta)$ exist and are continuous in some neighborhood of $ \alpha^*$ and $\beta^*$, respectively. Then bias of the EIG estimator $$\widehat{\mathrm{EIG}}_{M,N} = \frac{1}{M}\sum_{j = 1}^M \log \frac{q_{Y|X}( y^j| x^j; \widehat{\beta}_N)}{q_Y( y^j;  \widehat{\alpha}_N)},$$ \revise{where $\{x^i, y^i\}_{i=1}^{N}$ and $\{x^j, y^j\}_{j=1}^{M}$ are drawn i.i.d.\ from $\pi_{X,Y}$,} is $\mathcal O(1/N)$ and its variance is $\mathcal O(1/M + 1/N^2)$. Therefore, the MSE of this $\widehat{\mathrm{EIG}}_{M,N}$ is $\mathcal O(1/M + 1/N^2)$. 
\end{theorem}

With this analysis of the convergence rate, we can then explore the optimal allocation between training and evaluation samples.
\begin{corollary}\label{cor:opt_allocation}
    \revise{In the context of the preceding theorems,}
    given $L = M+N$ samples from the joint distribution $\pi_{X,Y}$, the optimal allocation between the training and evaluation samples should be set as $\frac{M}{N} \sim \mathcal O(L^{1/3})$. Correspondingly, the optimal convergence rate of the MSE of $\widehat{\mathrm{EIG}}_{M,N}$ is $\mathcal O(L^{-1})$.
\end{corollary}
The proof of this corollary is provided in Appendix~\ref{proof_of_thms}. We observe that if $\frac{M}{N} \sim \mathcal O(L^{\varpi})$ for some $\varpi > 0$, the mean squared error (MSE) decreases at a rate of $2L^{\varpi-2} + L^{-\varpi-1} + L^{2\varpi-2} + L^{-2} + L^{-1}$. If $\varpi > 1/2$, the leading term has an exponent greater \revise{(i.e., slower)} than $-1$. For all $\varpi \in (0, \frac{1}{2}]$, the \textit{asymptotic} convergence rate remains $\mathcal O(L^{-1})$. The optimal value, $\varpi = 1/3$, is obtained by solving a root-finding problem (see Appendix~\ref{proof_of_thms} for more detail), where the difference lies in the order of the non-leading term.
The rate of convergence is confirmed in the numerical examples presented in Section~\ref{subsec:LG}. 

\revise{While the analysis above assumes the error of the best approximation in our parametric class of densities to be negligible (i.e., zero), there are results in the literature that quantify the bias introduced by transport map approximation. In the case of analytic densities, \cite{zech2022sparse} establish convergence rates for sparse polynomial and neural network approximations of the KR map, as well as for the associated pushforward distributions, with error decreasing exponentially in the size of the neural network. 
In a different context, \cite{baptista2025approximation} study polynomial approximation of transport maps on compact domains and show, for $C^k$ densities, that the KL divergence between the target distribution and the pullback through the approximate map vanishes at a polynomial rate in the degree of the polynomial space (though the size of this space will grow exponentially with dimension).
\cite{wang2022minimax} simultaneously consider approximation error and statistical finite-sample estimation error for density estimation using transport, in a nonparametric setting. They show that Hellinger distance and KL divergence from the target distribution to the pullback through the approximate map decays at the optimal minimax rate. This rate is of course dimension-dependent, though as the smoothness of the target density tends to infinity, it approaches the dimension-independent parametric rate considered here.
As we discuss in Section~\ref{sec:conclusion}, using these results to generalize the analysis of optimal sample allocation presented here would be a useful avenue for future work. 
}

\section{EIG estimation in high dimensions}\label{sec:high_dim}
When the dimension of the parameters and data is high, directly \fl{estimating MI becomes computationally expensive.}
To address this, we aim to reduce the dimensionality of the parameters and data simultaneously, ensuring that the lower-dimensional subspace still preserves the \revise{MI}. In this section, we apply the dimension reduction methods proposed in~\cite{Ricarod_dimRed}. Given two unitary matrices $U = [U_r, U_\perp] \in \R^{n_x \times n_x}$ and $V = [V_s, V_\perp] \in \R^{n_y \times n_y}$, consider the decompositions 
\begin{align*}
    X &= U_r X_r + U_\perp X_\perp, \quad \text{where} \quad X_r = U_r^\top X, \; X_\perp = U_\perp^\top X \\
    Y &= V_s Y_s + V_\perp Y_\perp, \quad \text{where} \quad Y_s = V_s^\top Y, \;
     Y_\perp = V_\perp^\top Y.
\end{align*}
We interpret $X_\perp$ as the non-informed part of $X$, 
and similarly $Y_\perp$ is interpreted as the non-informative part of $Y$. 
Given such a decomposition, 
we define approximate posterior $\pi^*_{X|Y}$ as $\pi^*_{X|Y}(x|y) := \pi_{X_r|Y_s}(x_r|y_s) \pi_{X_\perp|X_r}(x_\perp|x_r)$. While this approximation is equal to the true posterior under the conditional independence conditions $X_\perp \indep Y \mid X_r$ and $X_r \indep Y_\perp \mid Y_s$, these conditions are difficult to satisfy in practice, which generally leads to $\pi^*_{X|Y} \neq \pi_{X|Y}$. Therefore, we aim to find transformations $U$ and $V$ such that the expected KL divergence from the approximate posterior $\pi^*_{X|Y}$ to the true posterior $\pi_{X|Y}$ is small relative to some tolerance $\epsilon > 0$. That is 
\begin{align*}
    \mathbb E_{\pi_Y} \left[ \Dkl\left(\pi_{X|Y}\left(\cdot |y\right)|| \pi^*_{X|Y}\left(\cdot |y\right)\right)\right] \leq \epsilon.
\end{align*}
To find rotation matrices $U$ and $V$ and reduced dimensions $r,s$ that minimize the posterior approximation error and satisfy the bound above, we use the following theorem from~\cite{Ricarod_dimRed}.
\begin{theorem}[Theorem 1 in \cite{Ricarod_dimRed}]
    Let $X$ and $Y$ be random variables 
    whose joint distribution satisfies the logarithmic Sobolev inequality with constant $C(\pi_{X,Y})$ as stated by Definition 2 in~\cite{Ricarod_dimRed}. Then, we have 
    \begin{align*}
     \mathbb E_{\pi_Y} \left[ \Dkl\left(\pi_{X|Y}\left(\cdot |y\right)|| \pi^*_{X|Y}\left(\cdot |y\right)\right)\right] \leq C(\pi_{X,Y})^2\left(\Tr(U_\perp^\top H_X U_\perp) + \Tr(V_\perp^\top H_Y V_\perp) \right),
\end{align*}
for any unitary matrices $U$ and $V$, where
\begin{align}
    H_X = \int \left(\nabla_x \nabla_y \log \pi_{Y|X}(y|x)\right)^\top  \left( \nabla_x \nabla_y \log \pi_{Y|X}(y|x) \right) \pi_{X,Y}(x,y) dxdy \label{eq:diagnostic_X} \\
    H_Y = \int \left(\nabla_x \nabla_y \log \pi_{Y|X}(y|x) \right) \left(\nabla_x \nabla_y \log \pi_{Y|X}(y|x) \right)^\top \pi_{X,Y}(x,y) dxdy \label{eq:diagnostic_Y}.
\end{align}
\end{theorem}
Our goal is then to solve the following optimization problem:
\begin{align*}
    \argmin_{U_\perp^\top, V_\perp^\top}\; \Tr(U_\perp^\top H_X U_\perp) + \Tr(V_\perp^\top H_Y V_\perp),
\end{align*}
subject to the constraints $U_\perp^\top U_\perp = I_{n_x-r}$ and $V_\perp^\top V_\perp = I_{n_y-s}$. Following classical results in linear algebra, $U_\perp$ are the eigenvectors corresponding to the smallest $r$ eigenvalues of the matrix $H_X$. Similarly, $V_\perp$ are the eigenvectors corresponding to the smallest $s$ eigenvalues of the matrix $H_Y$.

\subsection{Gaussian prior and likelihood case}
If we have a Gaussian prior and likelihood, the previous derivation can be further simplified. Consider the following problem:
\begin{align*}
Y = G(X) + \mathcal E,
\end{align*}
where $X \sim \mathcal{N}(0, \Sigma_X)$, $\mathcal{E} \sim \mathcal{N}(0, \Sigma_\mathcal{E})$, and $G$ is the forward model. Let $\widetilde{X} = \Sigma_X^{-1/2}X$ and $\widetilde{Y} = \Sigma_\mathcal{E}^{-1/2}Y$ be pre-conditioned parameter and data variables. The connection between $\widetilde{Y}$ and $\widetilde{X}$ is given by
\begin{align*}
\widetilde{Y} = \Sigma_\mathcal{E}^{-1/2} G(\Sigma_X^{1/2} \widetilde{X}) + \Sigma_\mathcal{E}^{-1/2} \mathcal{E},
\end{align*}
Using that 
$\nabla_{y} \nabla_{x} \pi_{\widetilde{Y}|\widetilde{X}}(y | x) = \Sigma_\mathcal{E}^{-1/2} \nabla G(X) \Sigma_X^{1/2}$, the diagnostic matrices in~\eqref{eq:diagnostic_X} and~\eqref{eq:diagnostic_Y} are given by
\begin{align} 
    H_{\widetilde X} &= \Sigma_{X}^{1/2} \left(\int \nabla G(X)^\top  \Sigma_{\mathcal{E}}^{-1} \nabla G(X)  \pi_{X}(x) dx \right) \Sigma_{X}^{1/2} \label{eq:H_X} \\
    H_{\widetilde Y} &= \Sigma_{\mathcal{E}}^{-1/2} \left(\int \nabla G(X)  \Sigma_{X} \nabla G(X)^\top  \pi_{X}(x) dx \right) \Sigma_{\mathcal{E}}^{-1/2} \label{eq:H_Y}. 
\end{align}
Let $\widetilde{U}$ and $\widetilde{V}$ denote the matrices whose columns are the eigenvectors of $H_{\widetilde{X}}$ and $H_{\widetilde{Y}}$, respectively. Then, we obtain the matrices $U$ and $V$ using
\begin{align*}
U = \Sigma_X^{-1/2} \widetilde{U}, \quad V = \Sigma_\mathcal{E}^{-1/2} \widetilde{V},  
\end{align*}
where $\widetilde{U}$ and $\widetilde{V}$ satisfy the properties $\widetilde{U}^\top \Sigma_X \widetilde{U} = I$ and $\widetilde{V}^\top \Sigma_\mathcal{E}\widetilde{V} = I$. Let $U_r$ and $V_s$ be the matrices composed of the first $r$ and $s$ columns of $U$ and $V$, respectively. The optimal projection, in terms of preserving \revise{MI}, is given by $X_r = U_r^\top X$ and $Y_s = V_s^\top Y$. We can then compute the \revise{MI} $\I(X_r, Y_s)$ in the reduced coordinates, and the error is proportional to \(\sum_{i = r+1}^{n_x} \lambda_i\left(U^\top_\perp H_X U_\perp\right) + \sum_{i = s+1}^{n_y} \lambda_i\left(V^\top_\perp H_Y V_\perp\right)\), where $\lambda_i$ denotes the $i$th largest eigenvalue of the corresponding matrix. For further details, see \cite{Ricarod_dimRed}.

\subsection{Lower bounds of MI and the choice of estimators}
To reduce the dimensionality of the samples, we need access to the gradient of the forward model, specifically $\nabla G(X)$, which is not available in a completely likelihood-free setting. This limits us to using $\widehat{\mathrm{EIG}}_{\mathrm{m}}$ and $\widehat{\mathrm{EIG}}_{\mathrm{pos}}$. In the rotated space, we define:
\begin{align*}
\widehat{\mathrm{EIG}}_{\mathrm{m}}(r,s) &= \frac{1}{M} \sum_{i=1}^M \log \frac{\pi_{Y_s|X_r}(y_s^i | x_r^i)}{\widehat{\pi}_{Y_s}(y_s^i)} , \\
\widehat{\textrm{EIG}}_{\mathrm{pos}}(r,s) &= \frac{1}{M} \sum_{i=1}^M  \log \frac{\widehat{\pi}_{X_r|Y_s}(x_r^i |y_s^i)}{\pi_{X_r}(x_r^i)} .
\end{align*}
Then note that $\mathbb E_{\pi_{X_r,Y_s}}\left[\log \frac{\widehat{\pi}_{X_r|Y_s}(x_r^i |y_s^i)}{\pi_{X_r}(x_r^i)}\right] \leq \I(X_r; Y_s) \leq \I(X; Y)$, where the first inequality follows from the lower bound of the EIG, and the second is obtained from the data processing inequality. Therefore, $\widehat{\mathrm{EIG}}_{\mathrm{pos}}$ serves as a lower bound in the asymptotic sense. However, the relationship between $\widehat{\mathrm{EIG}}_{\mathrm{m}}(r,s)$ and the true EIG remains unclear. Although $\mathbb E_{\pi_{X_r,Y_s}}\left[\log  \frac{\pi_{Y_s|X_r}(y_s^i | x_r^i)}{\widehat{\pi}_{Y_s}(y_s^i)}\right] \geq \I(X_r, Y_s)$, it does not provide a clear upper or lower bound for $\I(X, Y)$. We thus focus on computing $\widehat{\mathrm{EIG}}_{\mathrm{pos}}$ in high-dimensional settings. We summarize the proposed method for EIG estimation using low-dimensional projections in Algorithm~\ref{alg:EIG_alg_highDim}.

\begin{algorithm}
\caption{EIG estimation using low-dimensional projections}
\label{alg:EIG_alg_highDim}
\begin{algorithmic}[1]
\STATE{\textbf{Input}: Partition the $L$ pairs of samples into $N$ training and $M$ evaluation samples, such that $L = N+M$; $\nabla G(X)$; $\Sigma_X$; $\Sigma_Y$; $\Sigma_\mathcal{E}$; joint samples $\{x^i, y^i\}_{i = 1}^L$}
\STATE{\textbf{Output}: EIG estimator}
\STATE{Compute the diagnostic matrices $H_{\widetilde{X}}$ in~\eqref{eq:H_X} and $H_{\widetilde{Y}}$ in~\eqref{eq:H_Y}}
\STATE{Compute the leading $r$ and $s$ eigenvectors $U_r$ and $V_s$ of $H_{\widetilde{X}}$ and $H_{\widetilde{Y}}$, respectively.}
\STATE{Project parameter and data samples: $x^i_r = U_r^\top x^i$, $y^i_s = V_s^\top y^i$.}
\STATE{Use $N$ training samples to learn the transport maps and obtain $\widehat \pi_{X_r|Y_s}(x_r | y_s)$.}
\STATE{ 
  Use $M$ pairs of samples to estimate EIG by evaluating the empirical average $
    \widehat{\textrm{EIG}}_{\mathrm{pos}} = \frac{1}{M} \sum_{i = 1}^M \log \frac{\widehat \pi_{X_r|Y_s}(x_r^i \mid y_s^i)}{\pi_{X_r}(x_r^i)}$.}
\end{algorithmic}
\end{algorithm}

\section{Numerical experiments} \label{sec:num_res}
In this section, we present three numerical examples to illustrate our proposed method. The first is a linear Gaussian example, where we observe that the empirical convergence rate aligns with the theoretical convergence rate in the asymptotic regime. We then use a nonlinear example to demonstrate that our method performs well even when the joint distribution is non-Gaussian and the forward model is nonlinear. Finally, we consider a challenging nonlinear example where both the data and parameters are high dimensional; here we use our proposed dimension reduction technique to enable effective EIG estimation in high-dimensional settings and compare it to alternative reduction methods. \fl{Throughout this section, we use the following notation for our EIG estimators (as described in Algorithm~\ref{EIG_alg}):
\begin{align*}
    \widehat{\mathrm{EIG}}_{\mathrm{m}} &= \frac 1M \sum_{i = 1}^M\log \frac{\pi_{Y|X}( y^i| x^i)}{\widehat \pi_Y( y^i)},\\
    \widehat{\mathrm{EIG}}_{\mathrm{pos}} &= \frac 1M \sum_{i = 1}^M\log \frac{\widehat \pi_{X|Y}( x^i| y^i)}{\pi_X( x^i)},\\
    \widehat{\mathrm{EIG}}_{\mathrm{lik}} &= \frac 1M \sum_{i = 1}^M\log \frac{\widehat \pi_{Y|X}( y^i| x^i)}{\widehat \pi_Y( y^i)},\\
    \widehat{\mathrm{EIG}}_{\mathrm{pr}} &= \frac 1M \sum_{i = 1}^M\log \frac{\widehat \pi_{X|Y}( x^i| y^i)}{\widehat \pi_X( x^i)}.   
\end{align*}
}

\subsection{Sample allocations and convergence rates}\label{subsec:LG}
In this example, we consider a simple linear Gaussian problem where the EIG can be computed in closed form. \rsb{We restrict our transport map to be affine functions of their inputs, which is sufficient to ensure Assumption~\ref{ass:inclass} is satisfied for the Gaussian target distributions}. We first compare the different estimators, including the nested Monte Carlo (NMC) estimator, against the analytical result. \fl{The forward model $G\colon X \rightarrow Y$ is a random finite-dimensional linear map from $\mathbb{R}^{20}$ to $\mathbb{R}^{10}$ with prescribed eigenvalues $\lambda_i = 0.8 \lambda_{i-1}$, where $\lambda_1 = 1$.} The prior for $X$ and additive noise $\mathcal{E}$ are both multivariate Gaussian, i.e., $X \sim \mathcal{N}(0, \Sigma_{X})$ and $\mathcal{E} \sim \mathcal{N}(0, \Sigma_{\mathcal{E}})$. We generate $\Sigma_{X}$ using a squared exponential kernel \fl{given by $\sigma \exp(-\left\Vert z_i - z_j\right\Vert^2/l^2)$, where we set $\sigma = 0.1$ and $l = 0.1$, and $z_i$'s are $20$ equally spaced points between 0 and 1.} We set $\Sigma_{\mathcal{E}} = 0.01 I_{10}$. The true EIG is easily computed to be $\frac 12 \log \frac{\det \left(G \Sigma_{ X} G^\top + \Sigma_{\mathcal{E}}\right) }{\det \left(\Sigma_{ \mathcal{E}}\right)}$~\cite{thomas2006elements}. 

To investigate the effect of sample allocation, we set the ratio between the evaluation samples and training samples to be proportional to $L^\varpi$, i.e., $M/N \sim \mathcal O(L^\varpi)$, with $\varpi = \frac{1}{8}, \frac{1}{3}, \frac{3}{4}$ respectively. Each experiment is repeated 100 times. In the first study, we assume the model is known and the likelihood function can be evaluated exactly, and we estimate only the marginal distribution of $Y$. That is, we use $\widehat{\mathrm{EIG}}_{\mathrm{m}}$. We also compare the performance of transport map-based estimators with that of the NMC estimator, where the NMC estimator is computed using the optimal sample allocation discussed in~\cite{chi}. We then plot the estimated EIG along with their bias, variance, and mean squared error (MSE) for the different estimators in Figure~\ref{fig:violin_marg} and~\ref{fig:conv_marg}. We observe that when we set $\varpi = 1/3$, the MSE decreases most rapidly as the total number of samples increases. The empirical convergence rate also aligns with the theory presented in the previous section. Next, we show in Figures~\ref{fig:violin_post} and~\ref{fig:conv_post} that by fixing the prior and using transport maps to approximate the posterior density, i.e., setting the EIG estimator to $\widehat{\mathrm{EIG}}_{\mathrm{pos}}$, we achieve the same rate.

We then study the likelihood-free case, where we estimate either the likelihood and the marginal densities and obtain $\widehat{\mathrm{EIG}}_{\mathrm{lik}}$ (Figure~\ref{lik_over_marg_likfr}), or the posterior and prior densities and obtain $\widehat{\mathrm{EIG}}_{\mathrm{pr}}$ (Figure~\ref{post_over_pr_likfr}). As shown in Figures~\ref{fig:conv_lik_over_marg_likfr} and~\ref{fig:conv_post_over_pr_likfr}, in both cases, the convergence rate of the MSE of the transport map-based EIG estimator matches the theoretical rate $\mathcal O(L^{-1})$ in the asymptotic regime, which is faster than that of  NMC, \revise{whose rate is $\mathcal O(L^{-2/3})$ under the optimal sample allocation proposed in~\cite{Beck_laplace,chi}. It is worth noting that in the most general setting (not necessarily that of EIG estimation), the convergence rate of NMC is $\mathcal{O}(L^{-1/2})$; see~\cite{Rainforth2018OnNM} for a detailed discussion.} 

We make a few further observations: The convergence rate of the MSE with $\varpi = 1/3$ is of order $L^{-1}$ in the asymptotic region across all four cases. When $\varpi > \frac{1}{2}$, we observe that the convergence rate of the MSE is considerably slower than for $\varpi \leq \frac{1}{2}$ in all four cases, which supports our theoretical results developed in Section~\ref{sec:theory}. From the violin plot, we see that all the transport map-based results have smaller variance than NMC, and with $\varpi = 1/8$ and $\varpi = 1/3$, transport maps produce good results even with small sample sizes. Another noteworthy observation is that $\widehat{\mathrm{EIG}}_{\mathrm{lik}}$ and $\widehat{\mathrm{EIG}}_{\mathrm{pr}}$ exhibit the same values, up to numerical error (see Figures~\ref{lik_over_marg_likfr} and ~\ref{post_over_pr_likfr}). This occurs because we are in the linear Gaussian case, and if we assume the mean has been shifted to 0, then the map can be computed from covariance matrix of the joint samples. In other words, we can obtain $\widehat{\pi}_Y$, $\widehat{\pi}_{Y|X}$, $\widehat{\pi}_X$, and $\widehat{\pi}_{X|Y}$ directly from the sample covariance matrix of the joint distribution. Therefore, we have exactly $\frac{\widehat \pi_{X|Y}( x| y)}{\widehat \pi_X( x)} =\frac{\widehat \pi_{Y|X}( y| x)}{\widehat \pi_Y( y)} $ for any $x$ and $y$.

\revise{We also compare the convergence behavior of our estimators with that of the MLMC estimator proposed in \cite{Goda03072020}. The number of outer samples per level in MLMC is computed following Algorithm~1 in~\cite{Goda03072020}, with the number of levels set to~$\{8, 10, 10, 12\}$ respectively. As illustrated in panel~(c) of Figures~\ref{fig:conv_marg},~\ref{fig:conv_post},~\ref{fig:conv_post_over_pr_likfr}, and~\ref{fig:conv_lik_over_marg_likfr}, the MLMC estimator achieves a faster convergence rate than NMC but converges more slowly than our proposed TM-based estimators under their optimal sample allocation; moreover, it exhibits a larger error for every value of $L$ tested here.}

\begin{figure}[!htbp]
\centering
\begin{subfigure}{.45\textwidth}
  \centering
  \includegraphics[width=\linewidth]{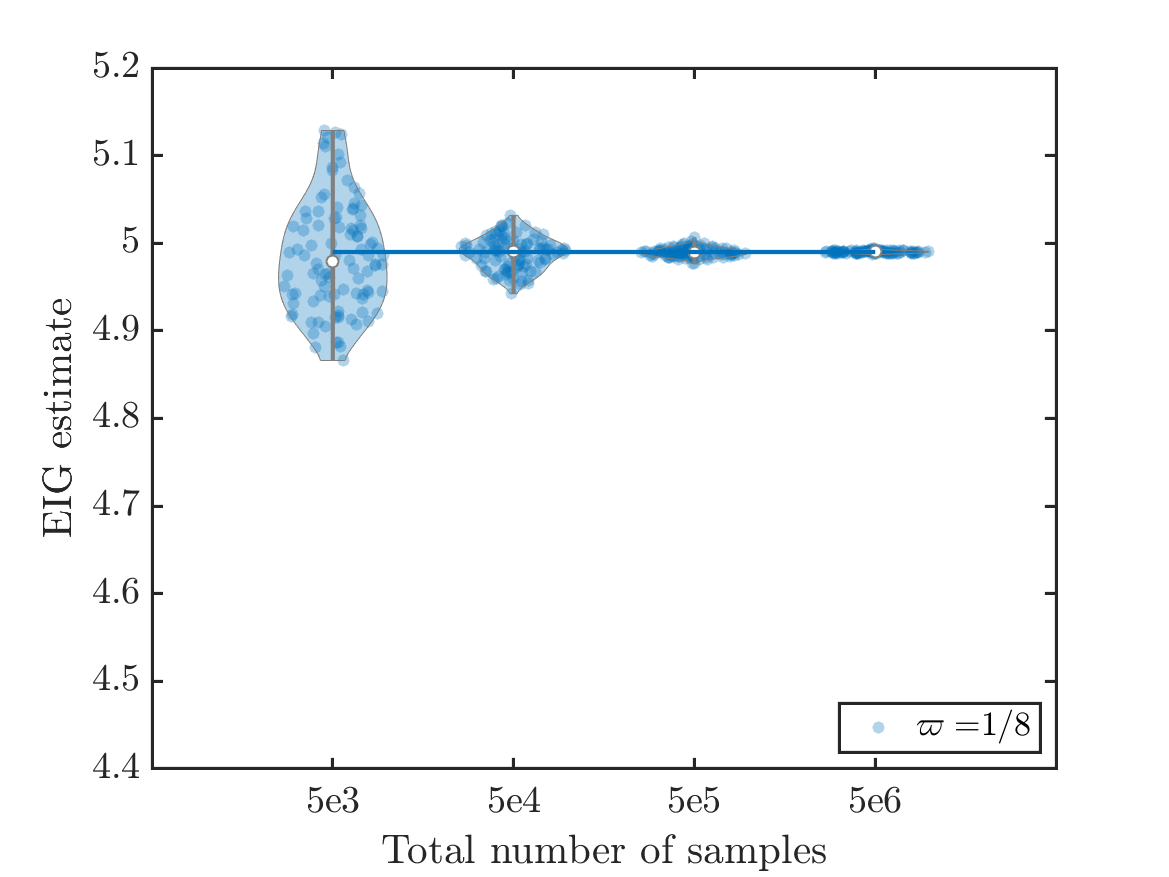}
\end{subfigure}%
\begin{subfigure}{.45\textwidth}
  \centering
  \includegraphics[width=\linewidth]{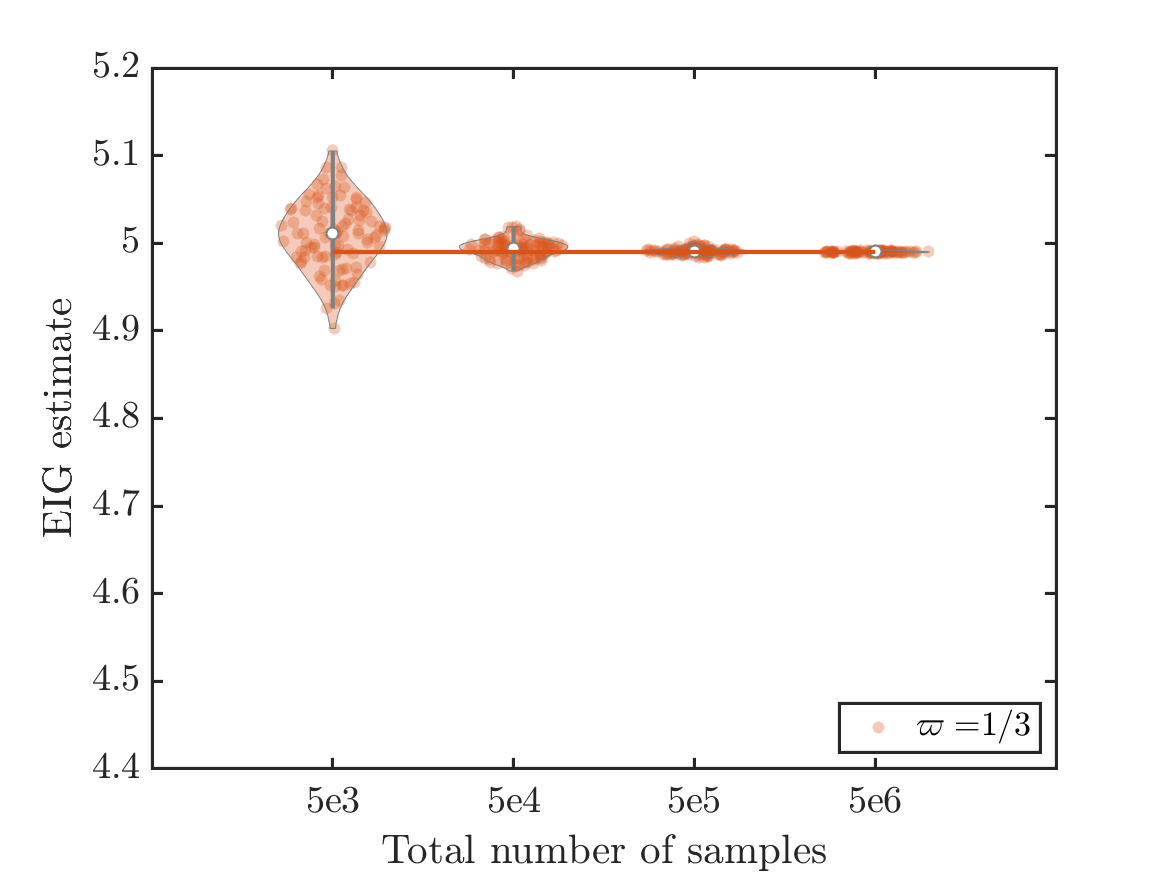}
\end{subfigure}
\\
\begin{subfigure}{.45\textwidth}
  \centering
  \includegraphics[width=\linewidth]{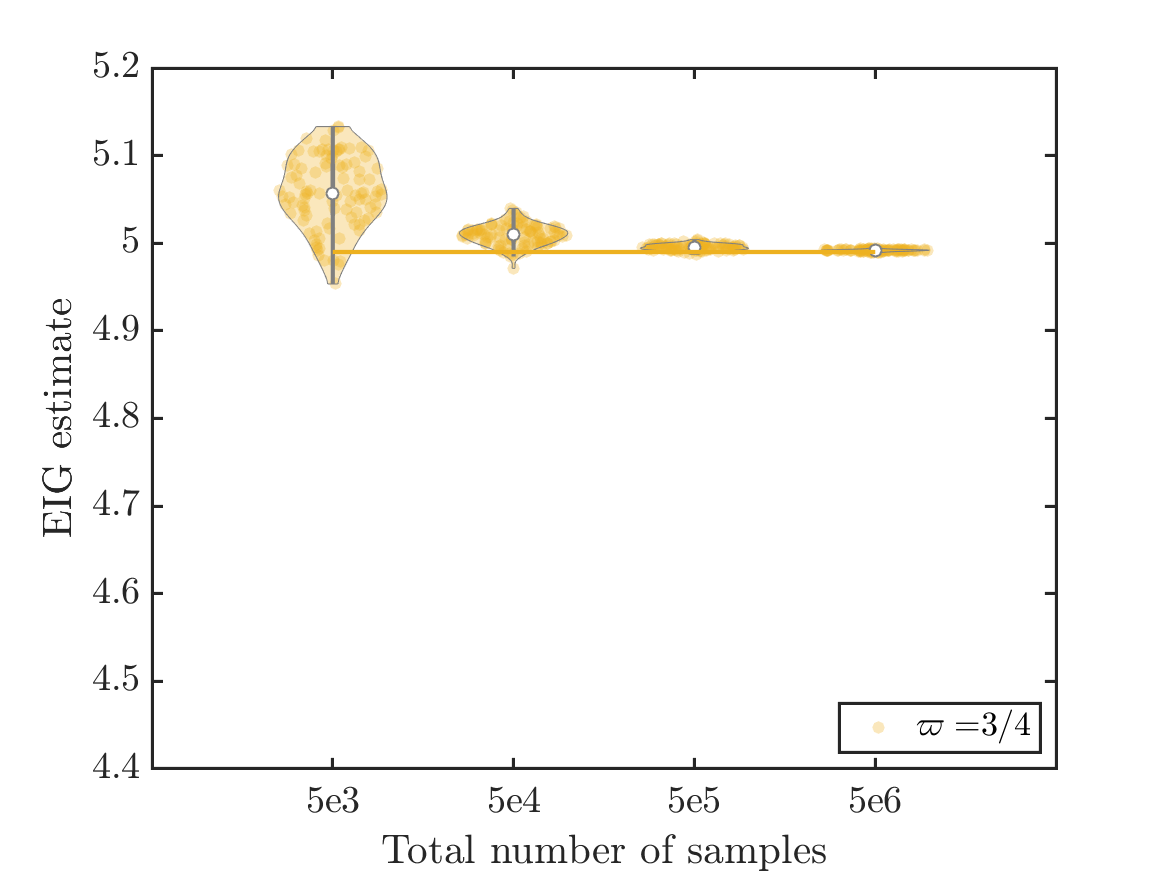}
\end{subfigure}
\begin{subfigure}{.45\textwidth}
  \centering
  \includegraphics[width=\linewidth]{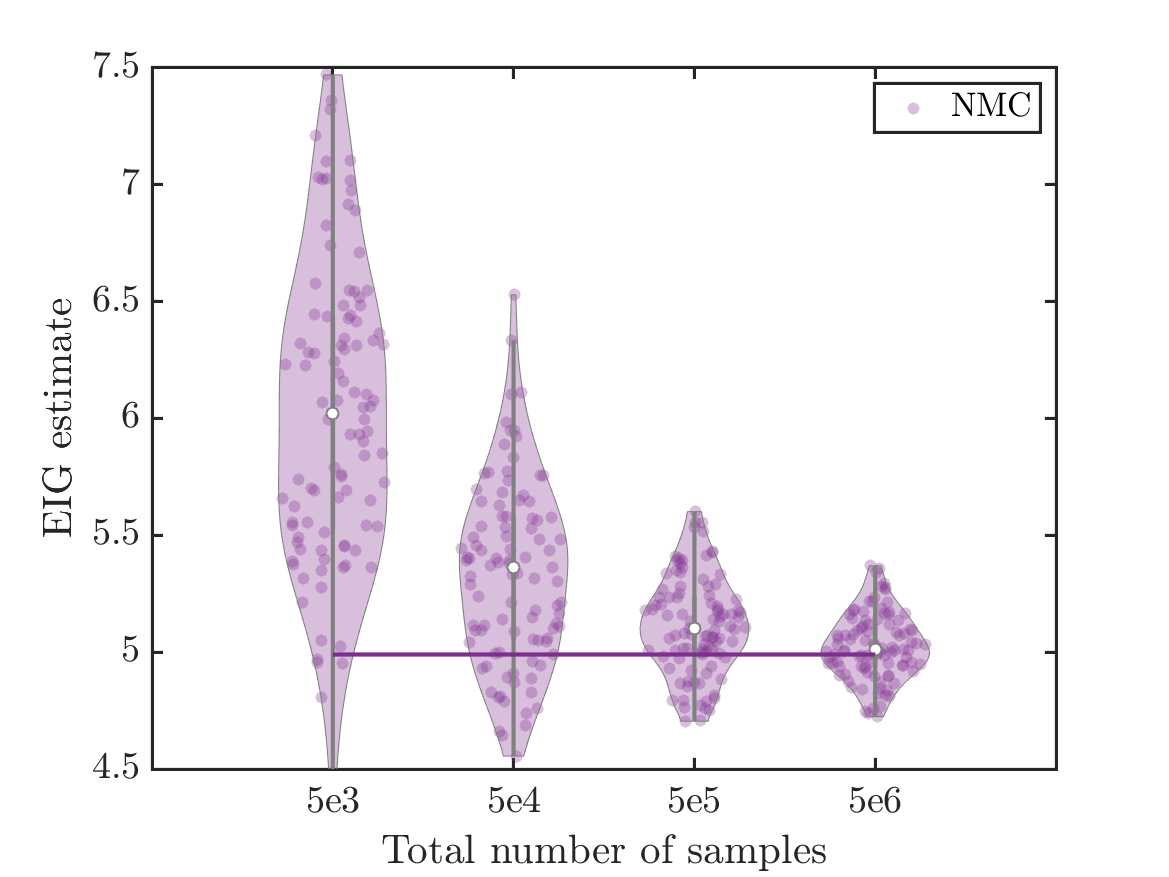}
\end{subfigure}
\caption{Violin plot of $\widehat{\mathrm{EIG}}_{\mathrm{m}}$, with different ratios between the number of training and evaluation samples, compared to NMC (bottom right). The solid lines represent the exact EIG value, computed via a closed-form expression.}
\label{fig:violin_marg}
\end{figure}

\begin{figure}[!ht]
\centering
\begin{subfigure}{.31\textwidth}
  \centering
  \includegraphics[width=\linewidth]{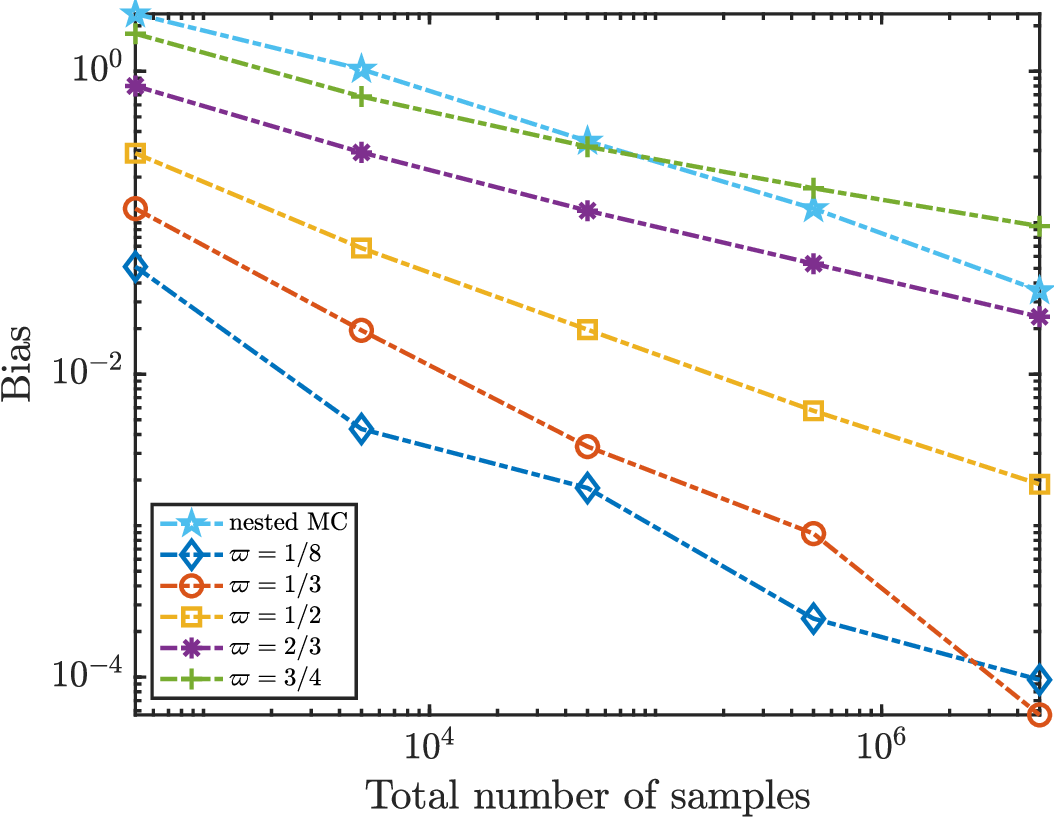}
\end{subfigure}%
\hspace{.5em}
\begin{subfigure}{.31\textwidth}
  \centering
  \includegraphics[width=\linewidth]{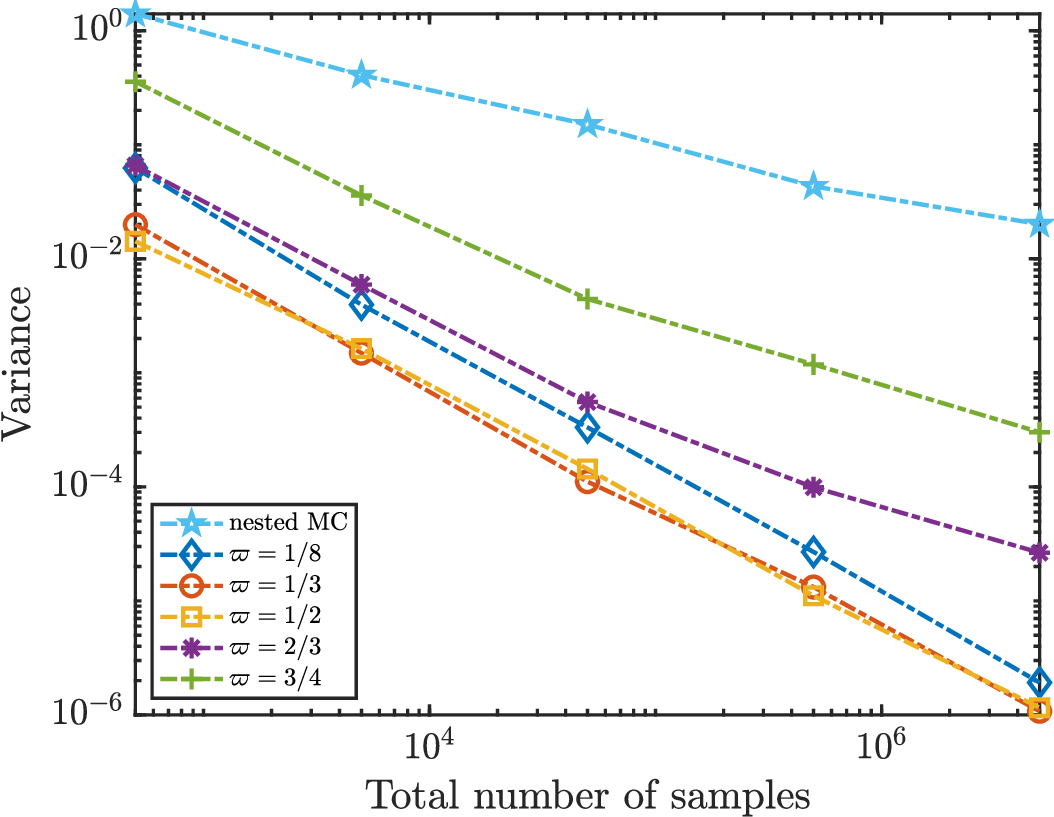}
\end{subfigure}
\hspace{.5em}
\begin{subfigure}{.31\textwidth}
  \centering
  \includegraphics[width=\linewidth]{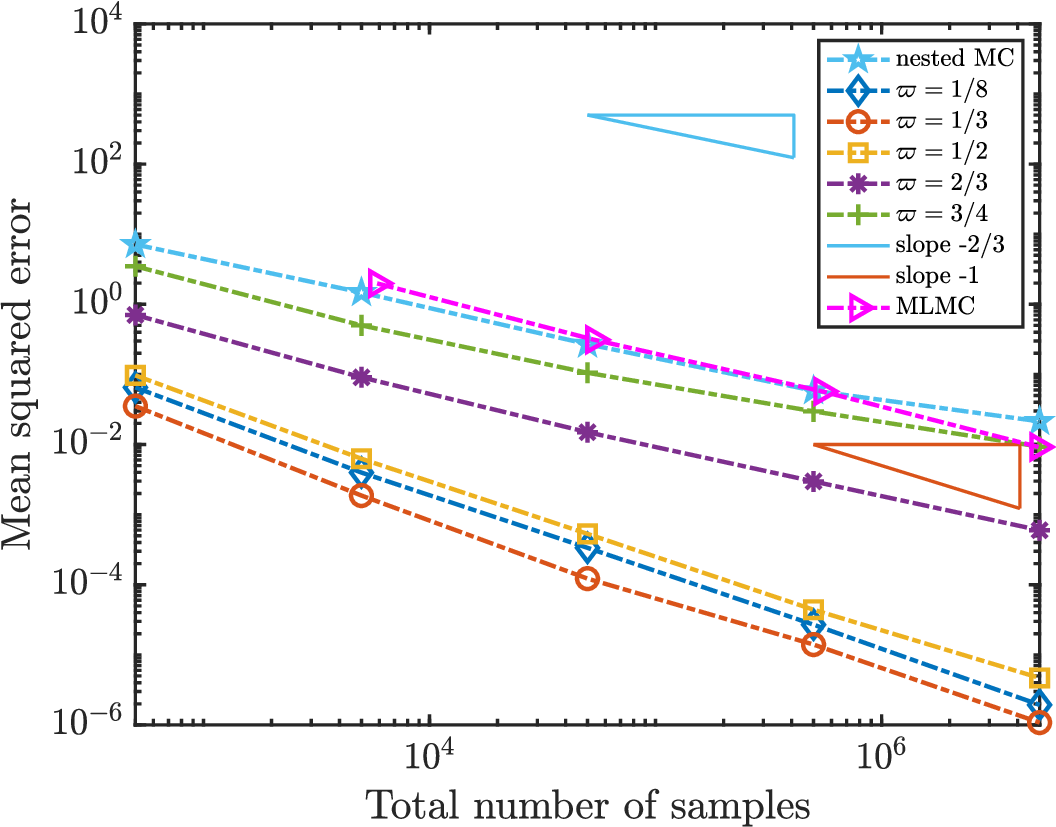}
\end{subfigure}
\caption{Convergence of the bias, variance, and MSE of $\widehat{\mathrm{EIG}}_{\mathrm{m}}$, also compared with NMC and MLMC.}
\label{fig:conv_marg}
\end{figure}

\begin{figure}[!ht]
\centering
\begin{subfigure}{.45\textwidth}
  \centering
  \includegraphics[width=\linewidth]{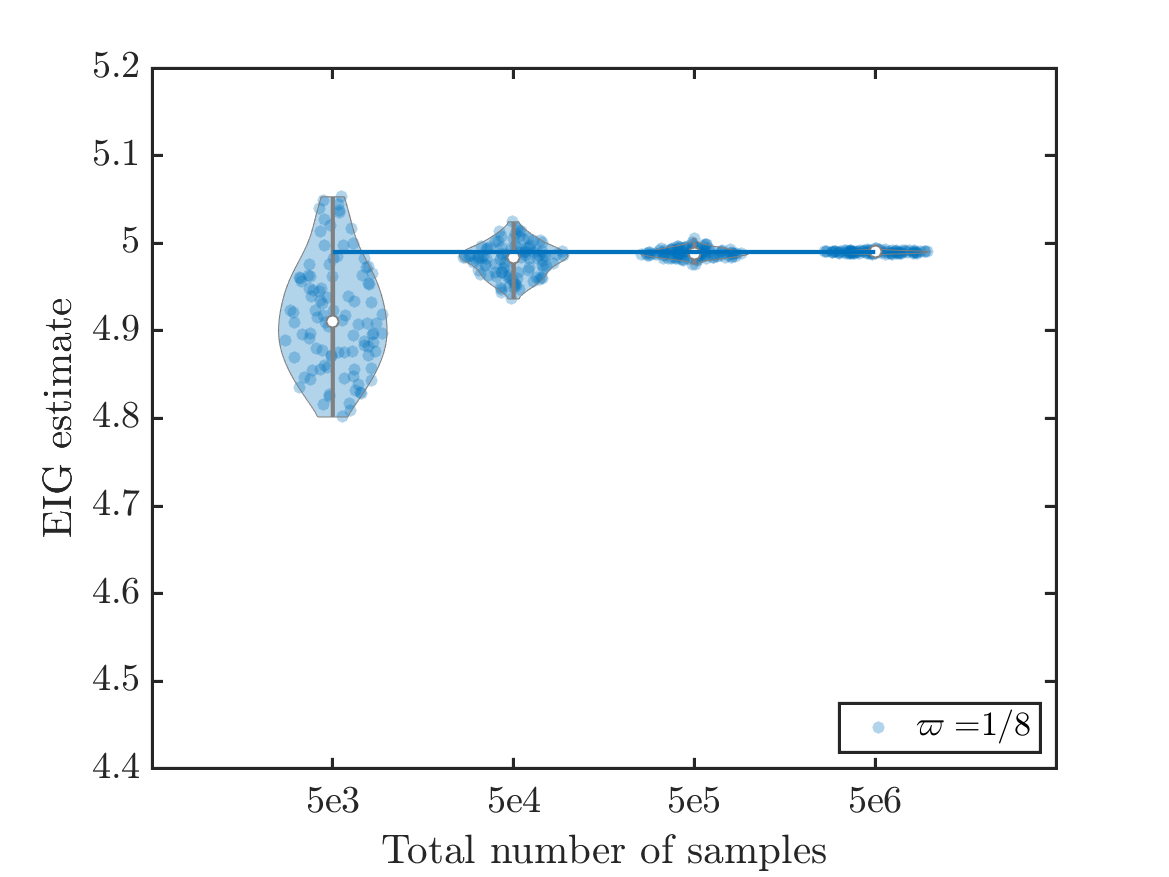}
\end{subfigure}%
\begin{subfigure}{.45\textwidth}
  \centering
  \includegraphics[width=\linewidth]{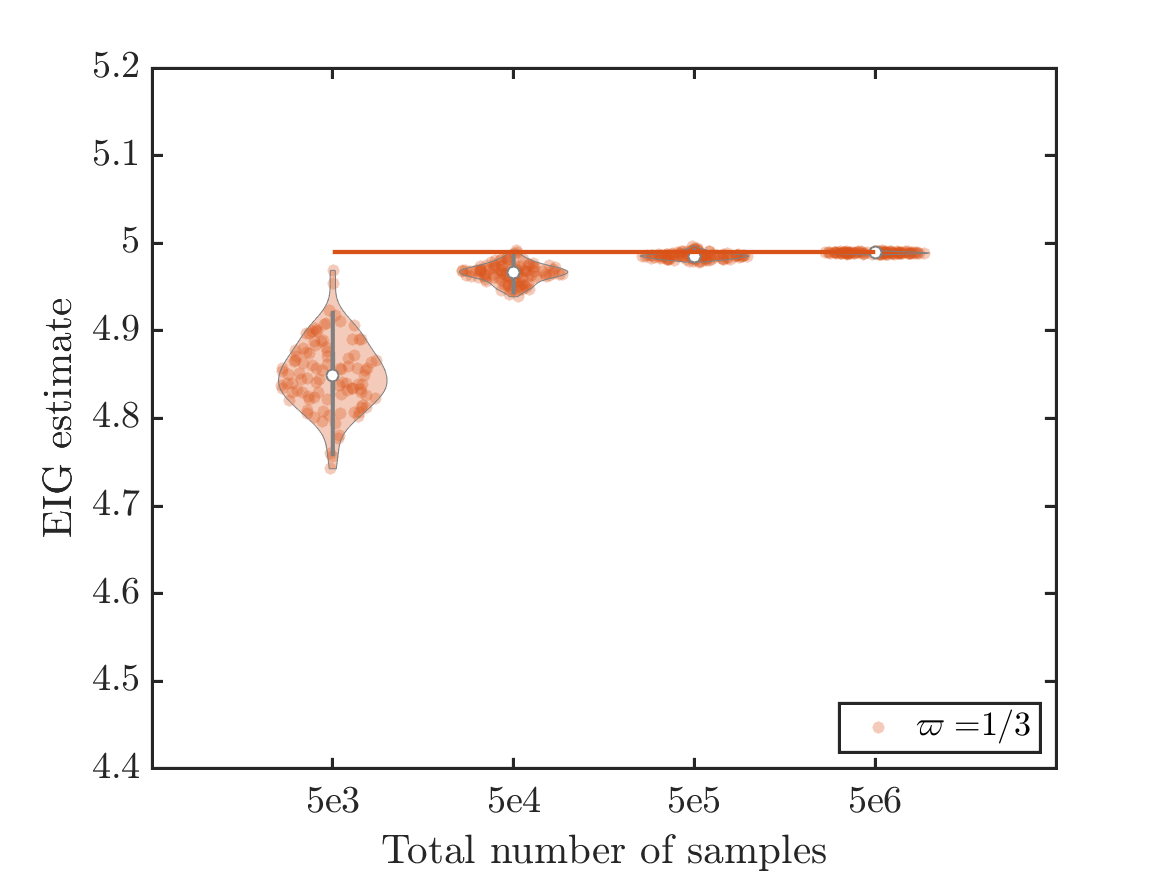}
\end{subfigure}
\\
\begin{subfigure}{.45\textwidth}
  \centering
  \includegraphics[width=\linewidth]{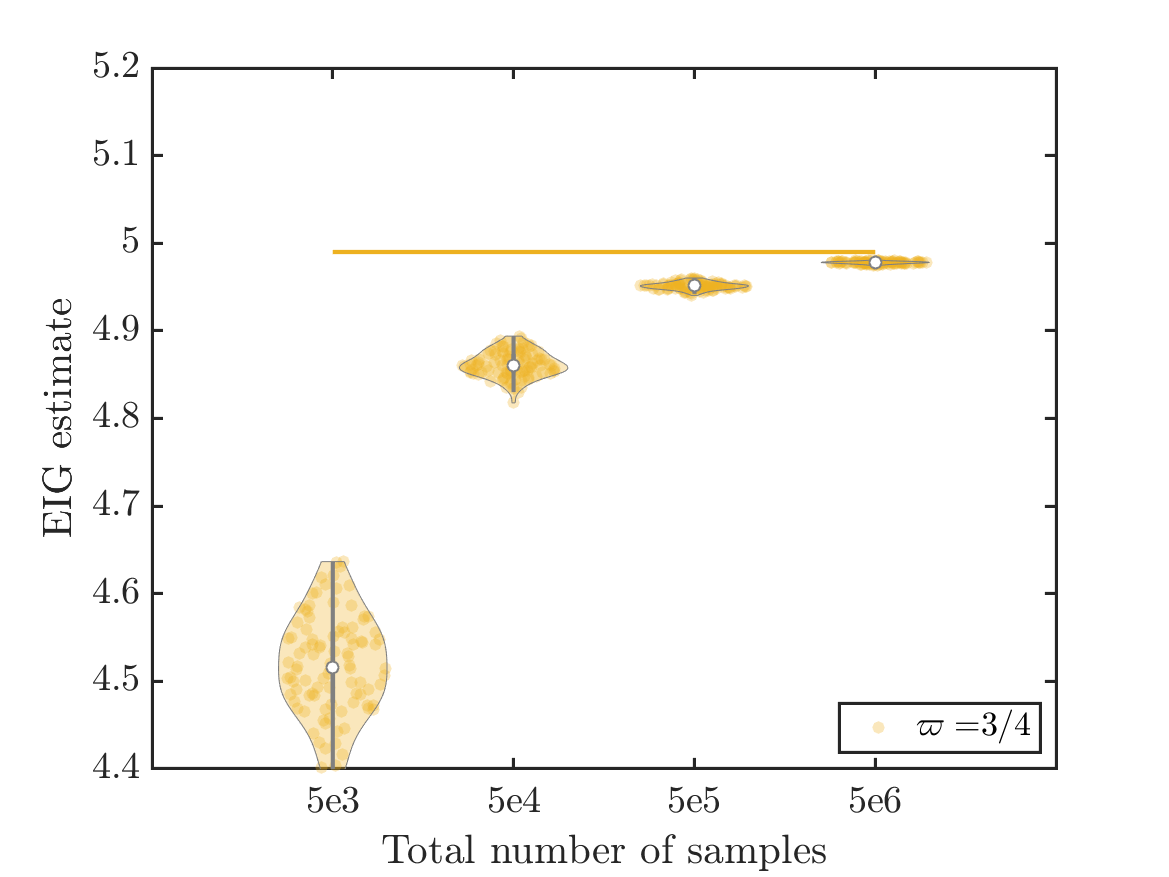}
\end{subfigure}
\begin{subfigure}{.45\textwidth}
  \centering
  \includegraphics[width=\linewidth]{violin_nmc.eps}
\end{subfigure}
\caption{Violin plot of  $\widehat{\mathrm{EIG}}_{\mathrm{pos}}$, with different ratios between the number of training and evaluation samples, compared to NMC (bottom right). The solid lines represent the exact EIG value.}
\label{fig:violin_post}
\end{figure}

\begin{figure}[!ht]
\begin{subfigure}{.31\textwidth}
  \centering
  \includegraphics[width=\linewidth]{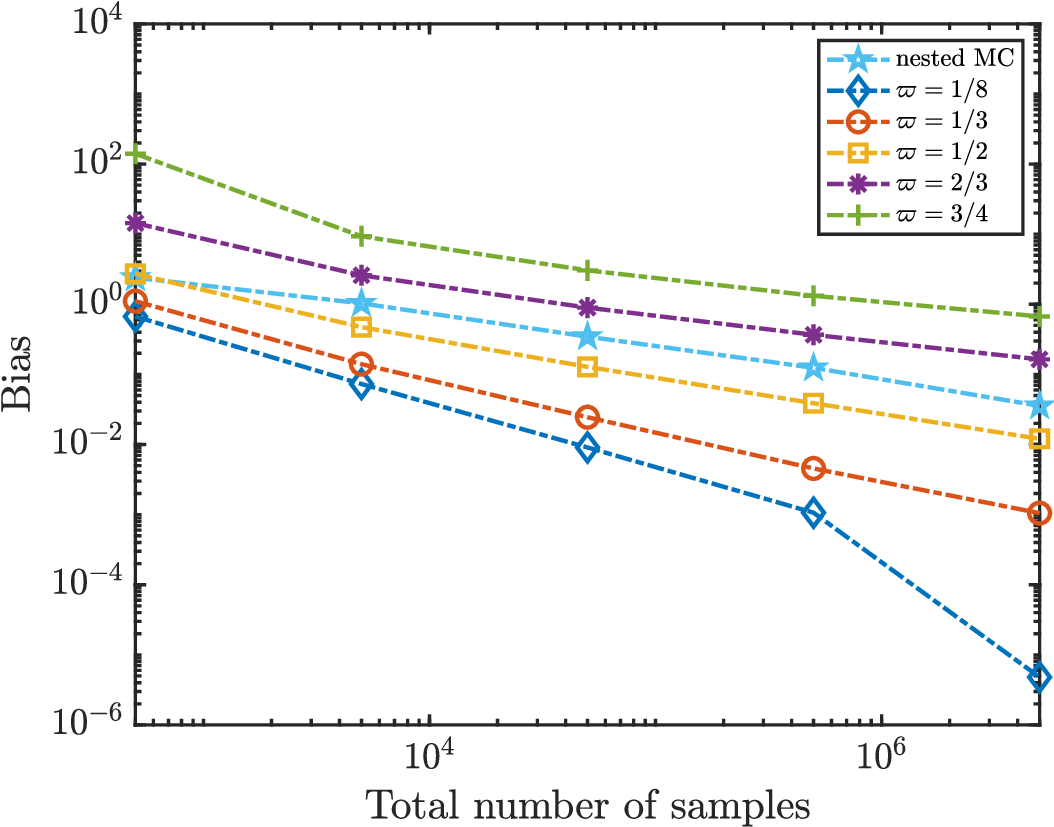}

\end{subfigure}%
\hspace{.5em}
\begin{subfigure}{.31\textwidth}
  \centering
  \includegraphics[width=\linewidth]{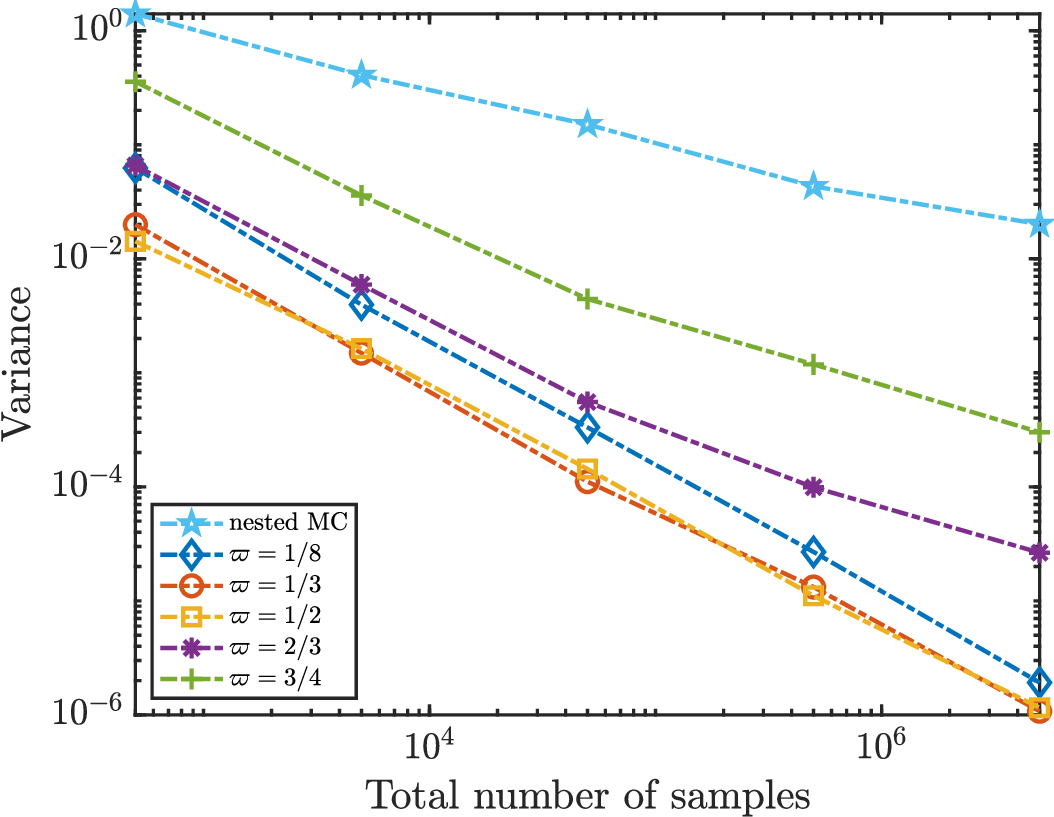}

\end{subfigure}
\hspace{.5em}
\begin{subfigure}{.31\textwidth}
  \centering
  \includegraphics[width=\linewidth]{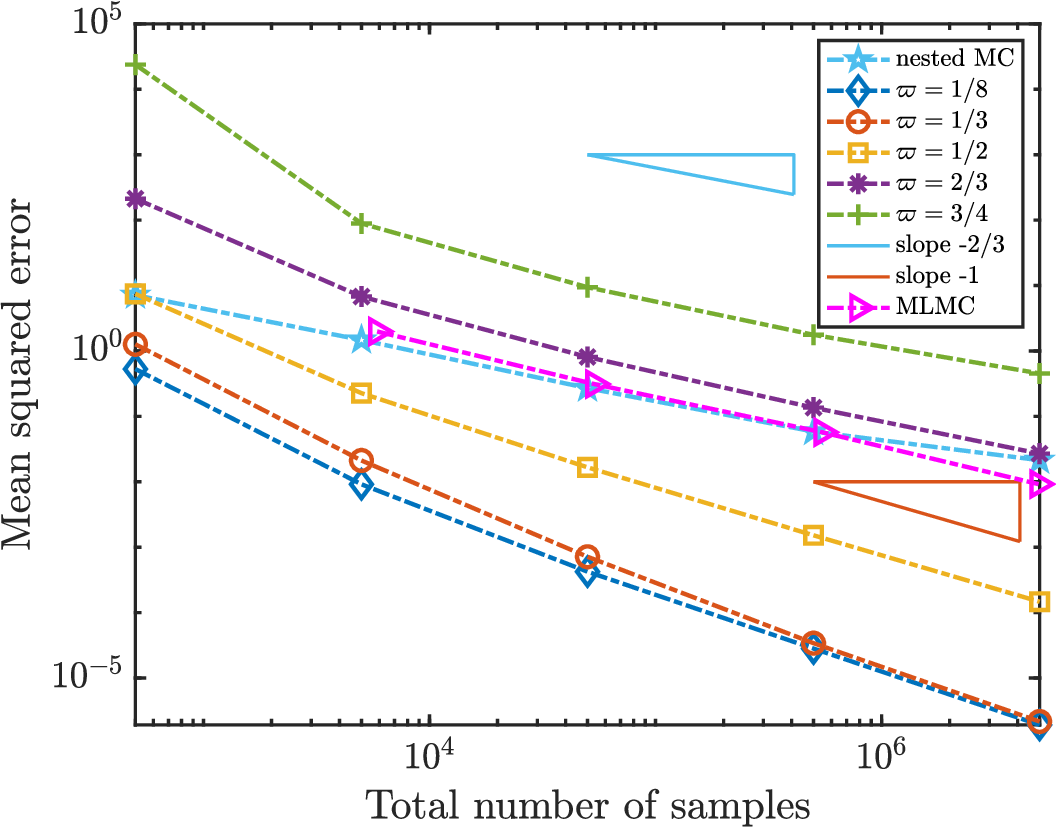}

\end{subfigure}
\caption{Convergence of the bias, variance, and MSE of $\widehat{\mathrm{EIG}}_{\mathrm{pos}}$, compared also with NMC and MLMC.}
\label{fig:conv_post}
\end{figure}

\begin{figure}[!ht]
\centering
\begin{subfigure}{.45\textwidth}
  \centering
  \includegraphics[width=\linewidth]{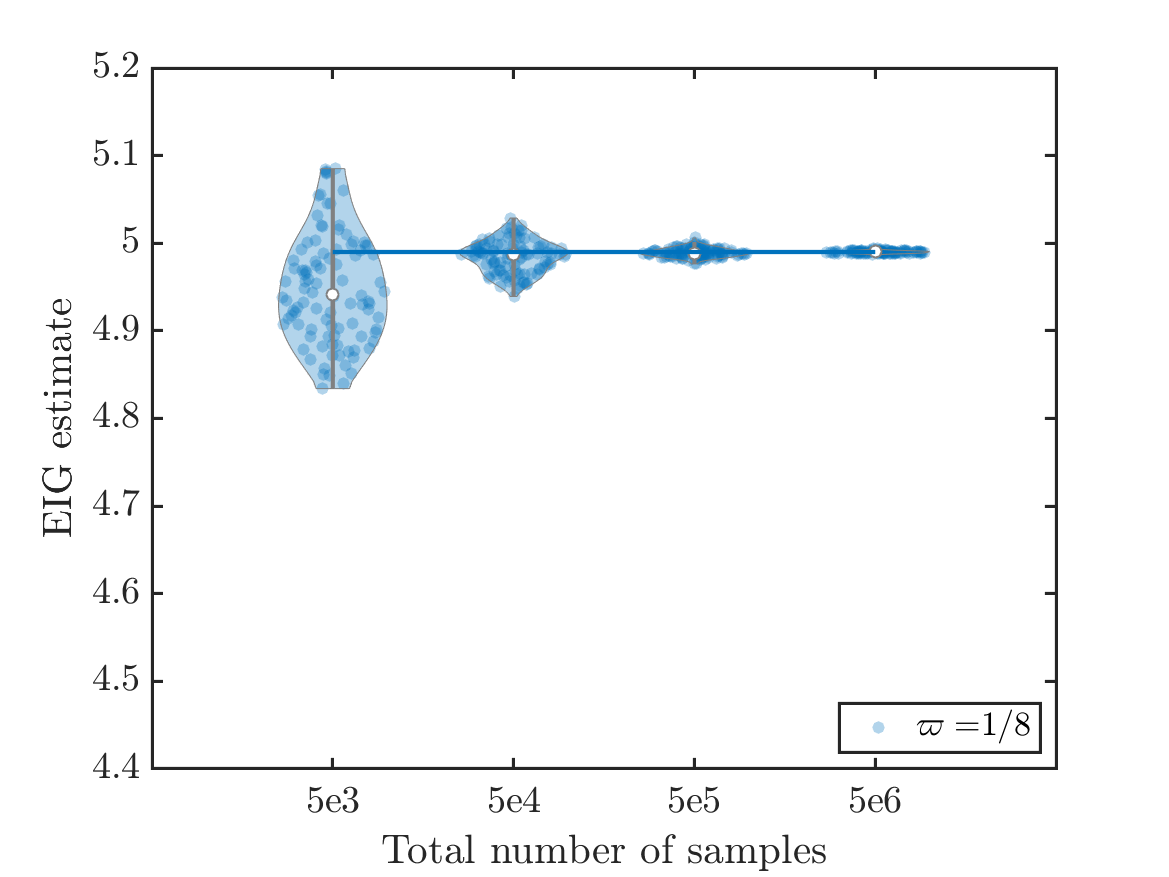}
\end{subfigure}%
\begin{subfigure}{.45\textwidth}
  \centering
  \includegraphics[width=\linewidth]{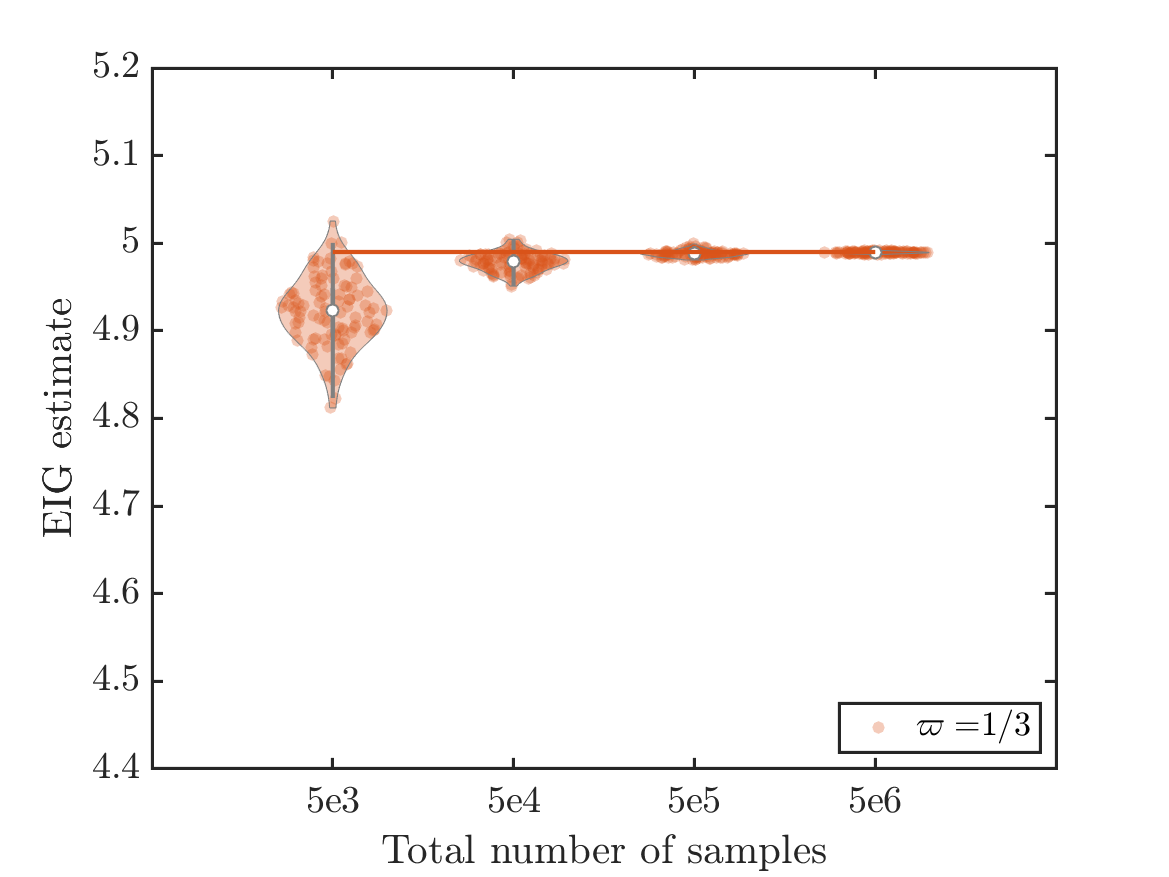}
\end{subfigure}
\\
\begin{subfigure}{.45\textwidth}
  \centering
  \includegraphics[width=\linewidth]{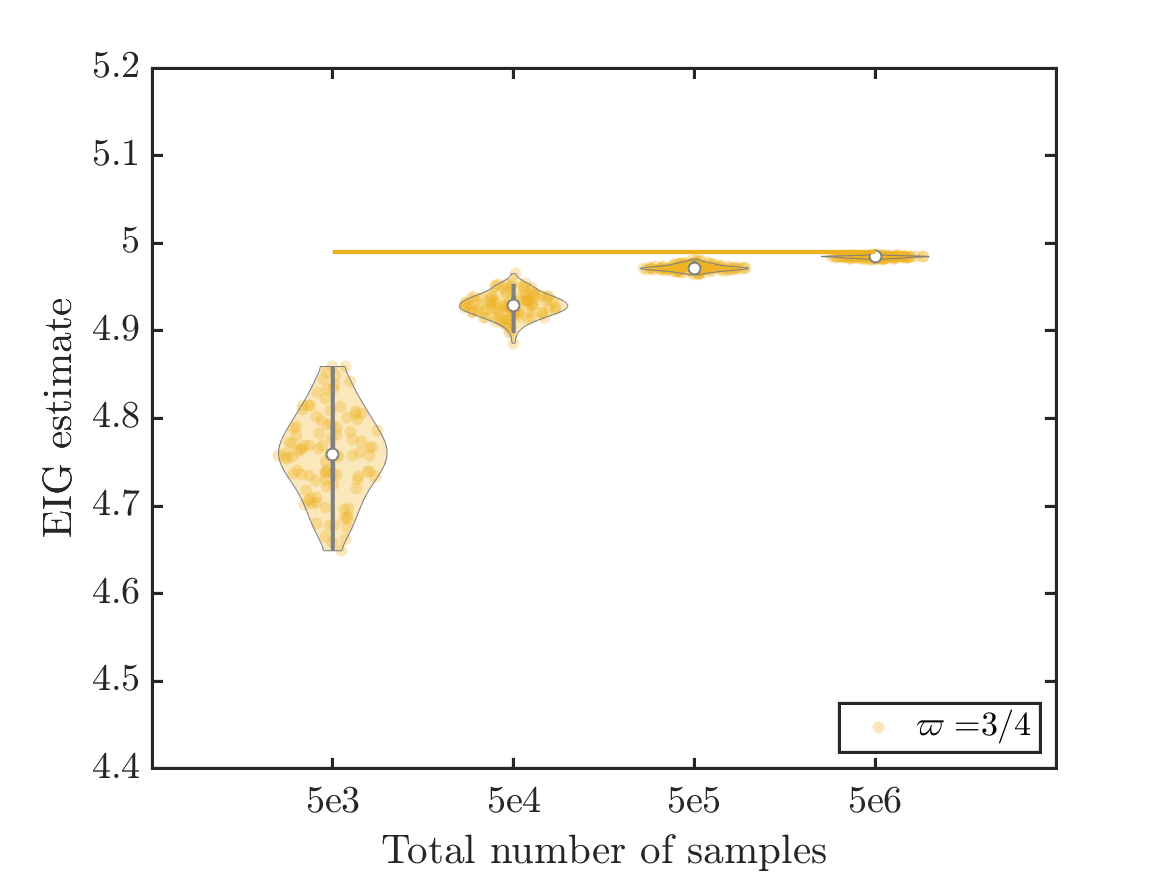}
\end{subfigure}
\begin{subfigure}{.45\textwidth}
  \centering
  \includegraphics[width=\linewidth]{violin_nmc.eps}
\end{subfigure}
\caption{Violin plot of $\widehat{\mathrm{EIG}}_{\mathrm{lik}}$, with different ratios between the number of training and evaluation samples, compared to NMC (bottom right). The solid lines represent the exact EIG value.}
\label{lik_over_marg_likfr}
\end{figure}

\begin{figure}[!ht]
\centering
\begin{subfigure}{.31\textwidth}
  \centering
  \includegraphics[width=\linewidth]{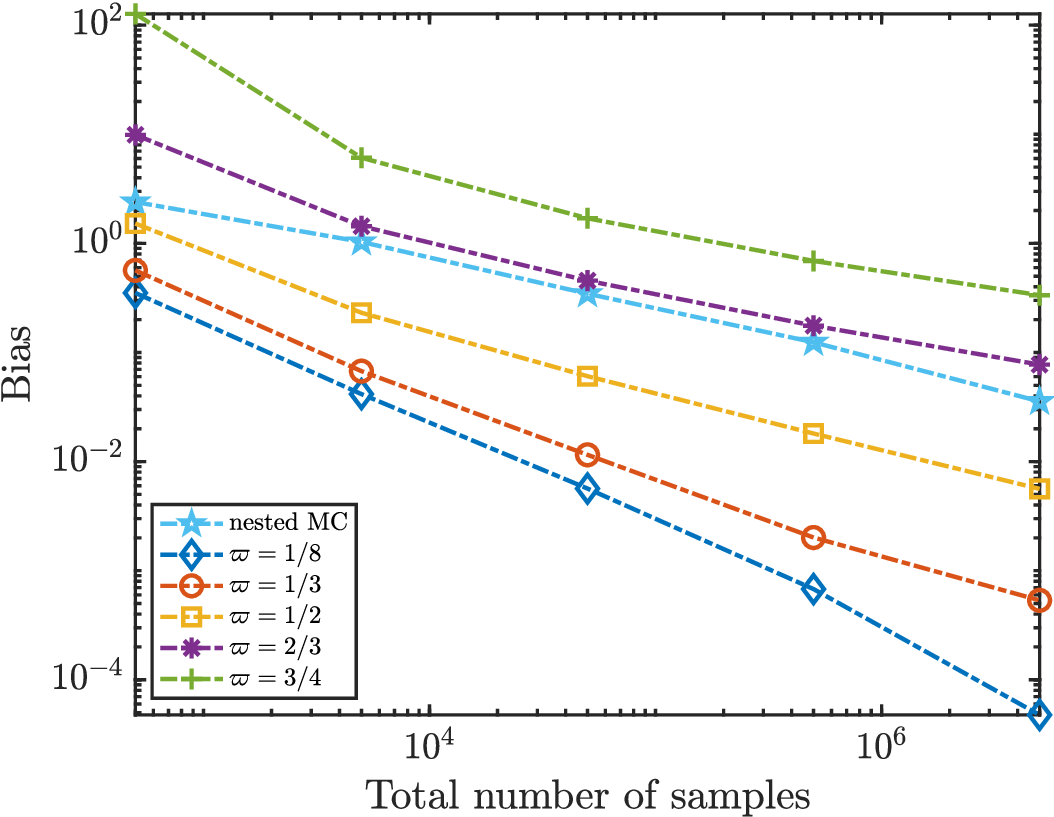}
  \caption{}
\end{subfigure}%
\hspace{.5em}
\begin{subfigure}{.31\textwidth}
  \centering
  \includegraphics[width=\linewidth]{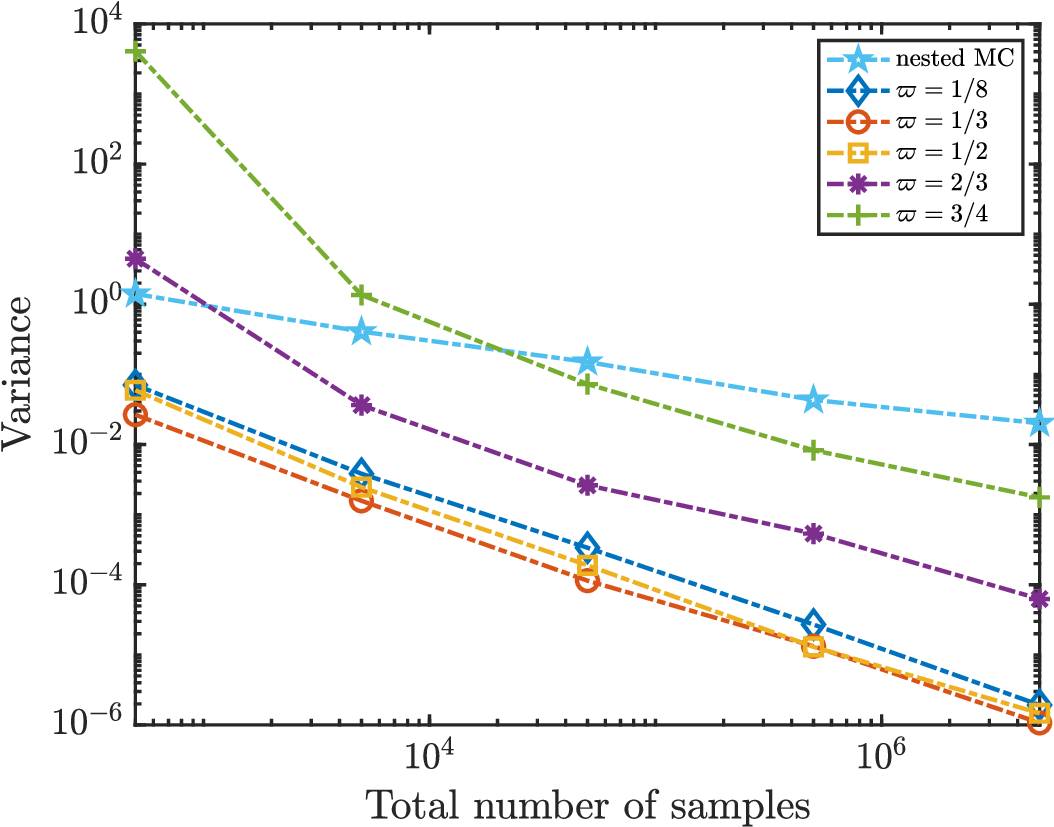}
  \caption{}
\end{subfigure}
\hspace{.5em}
\begin{subfigure}{.31\textwidth}
  \centering
  \includegraphics[width=\linewidth]{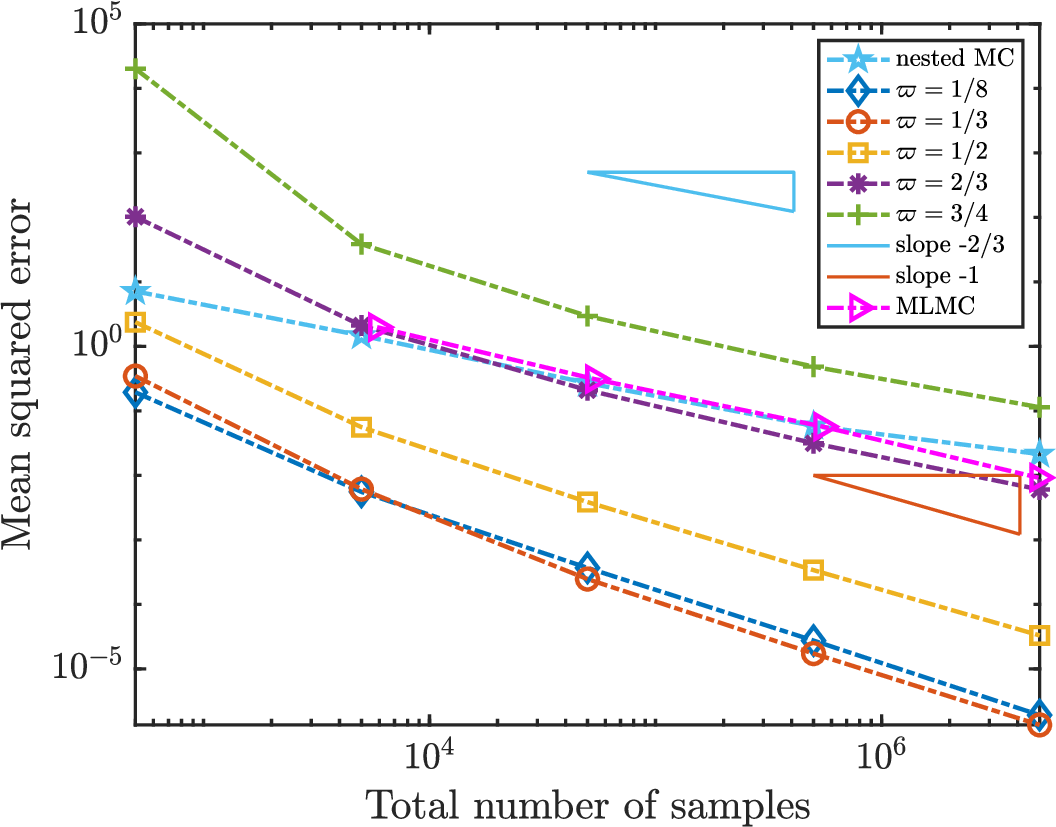}
  \caption{}
\end{subfigure}
\caption{Convergence of the bias, variance, and MSE of $\widehat{\mathrm{EIG}}_{\mathrm{lik}}$, compared also with NMC and MLMC.}
\label{fig:conv_lik_over_marg_likfr}
\end{figure}

\begin{figure}[!ht]
\centering
\begin{subfigure}{.45\textwidth}
  \centering
  \includegraphics[width=\linewidth]{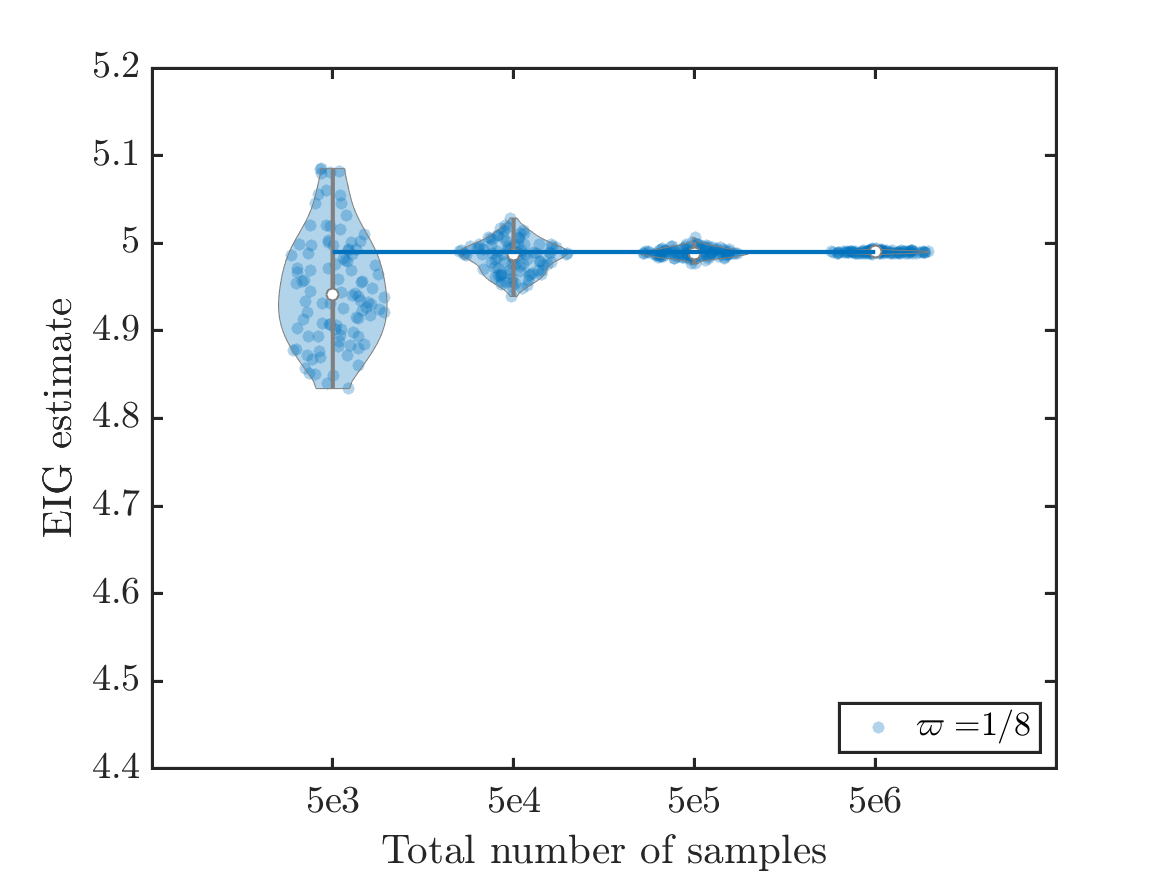}
\end{subfigure}%
\begin{subfigure}{.45\textwidth}
  \centering
  \includegraphics[width=\linewidth]{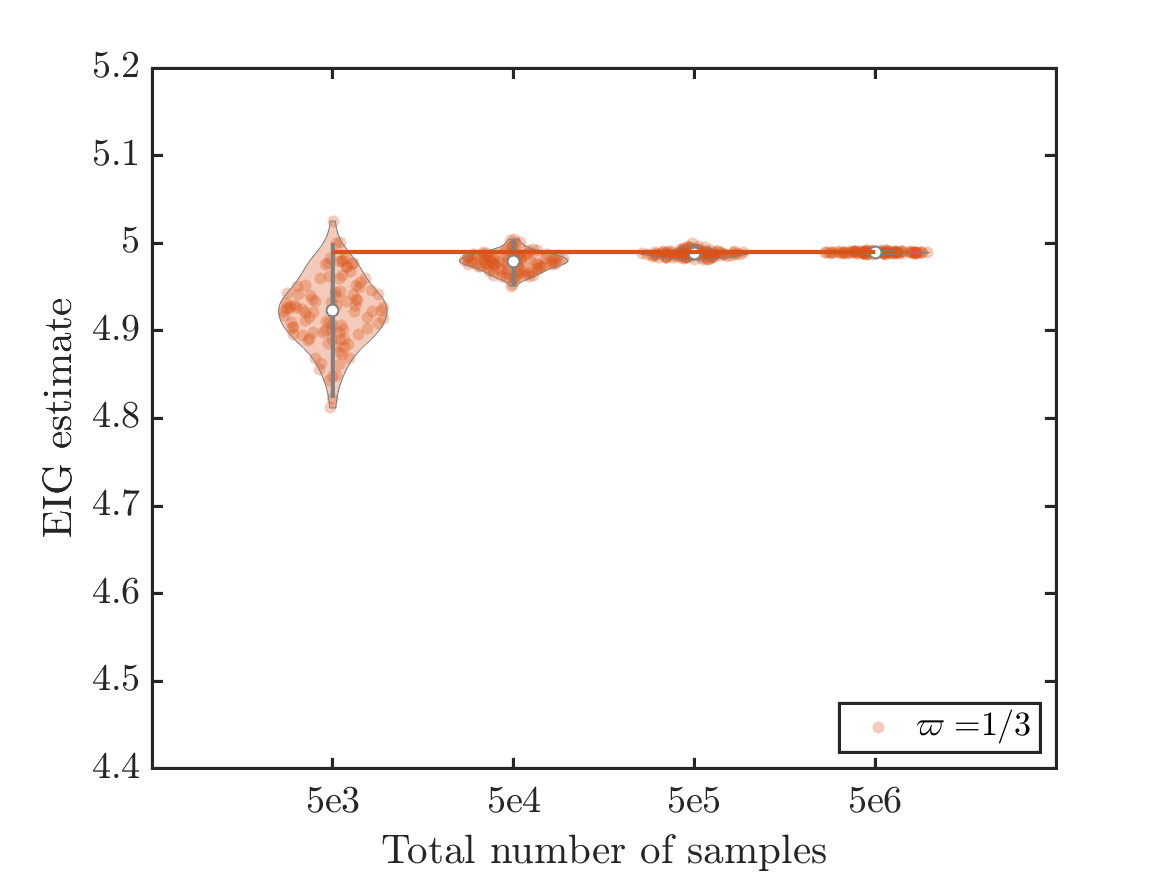}
\end{subfigure}
\\
\begin{subfigure}{.45\textwidth}
  \centering
  \includegraphics[width=\linewidth]{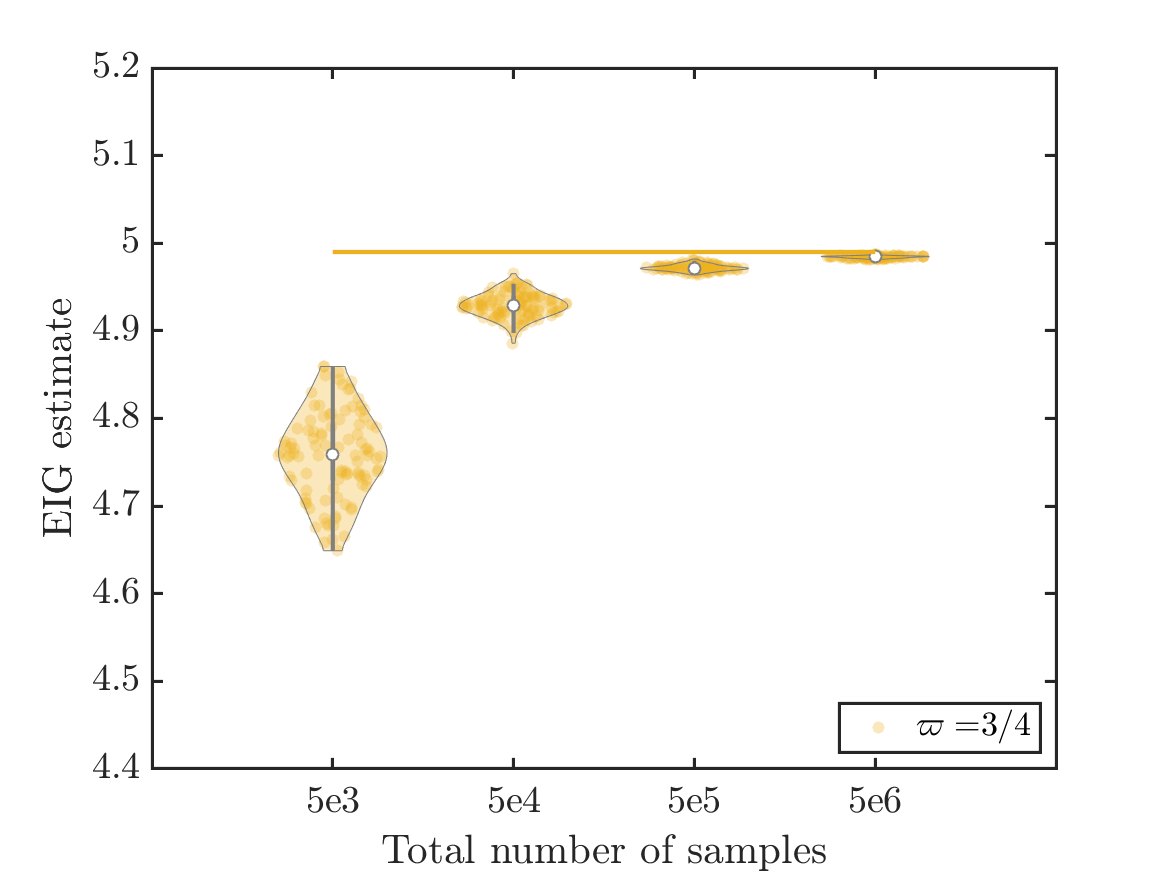}
\end{subfigure}
\begin{subfigure}{.45\textwidth}
  \centering
  \includegraphics[width=\linewidth]{violin_nmc.eps}
\end{subfigure}
\caption{Violin plot of $\widehat{\mathrm{EIG}}_{\mathrm{pr}}$, with different ratios between the number of training and evaluation samples, compared to NMC (bottom right). The solid lines represent the exact EIG value.}
      \label{post_over_pr_likfr}
\end{figure}

\begin{figure}[!ht]
\centering
\begin{subfigure}{.31\textwidth}
  \centering
  \includegraphics[width=\linewidth]{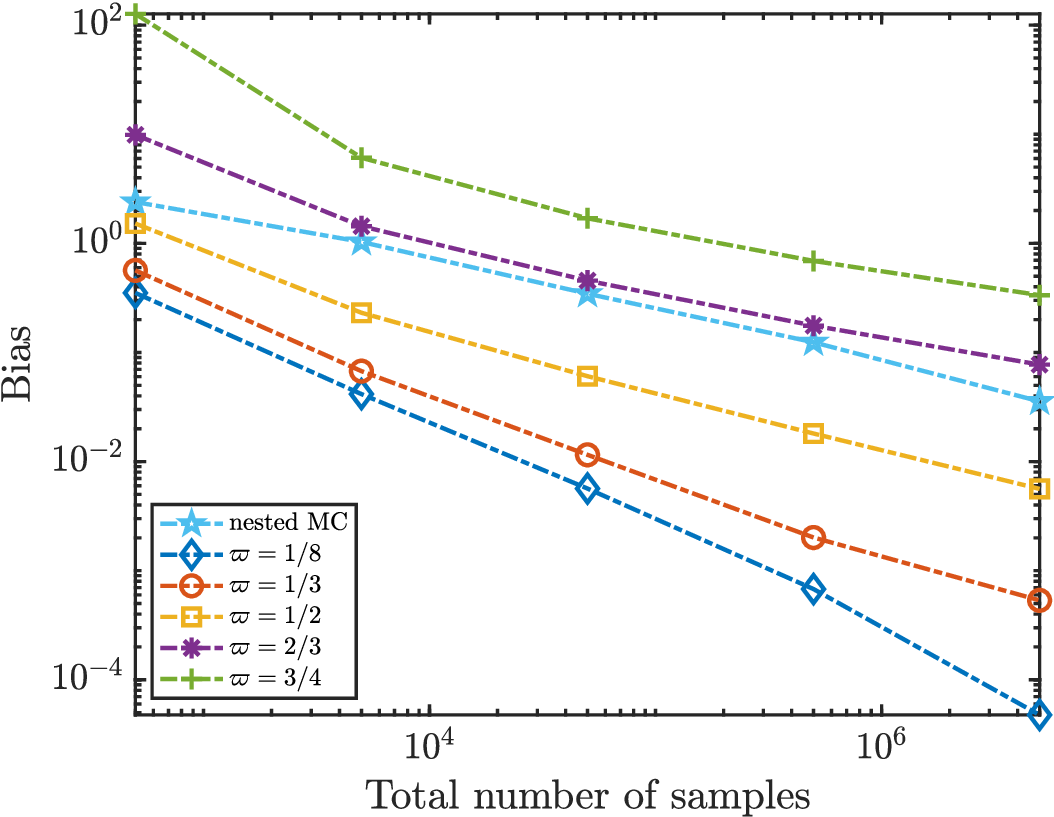}
  \caption{}
\end{subfigure}%
\hspace{.5em}
\begin{subfigure}{.31\textwidth}
  \centering
  \includegraphics[width=\linewidth]{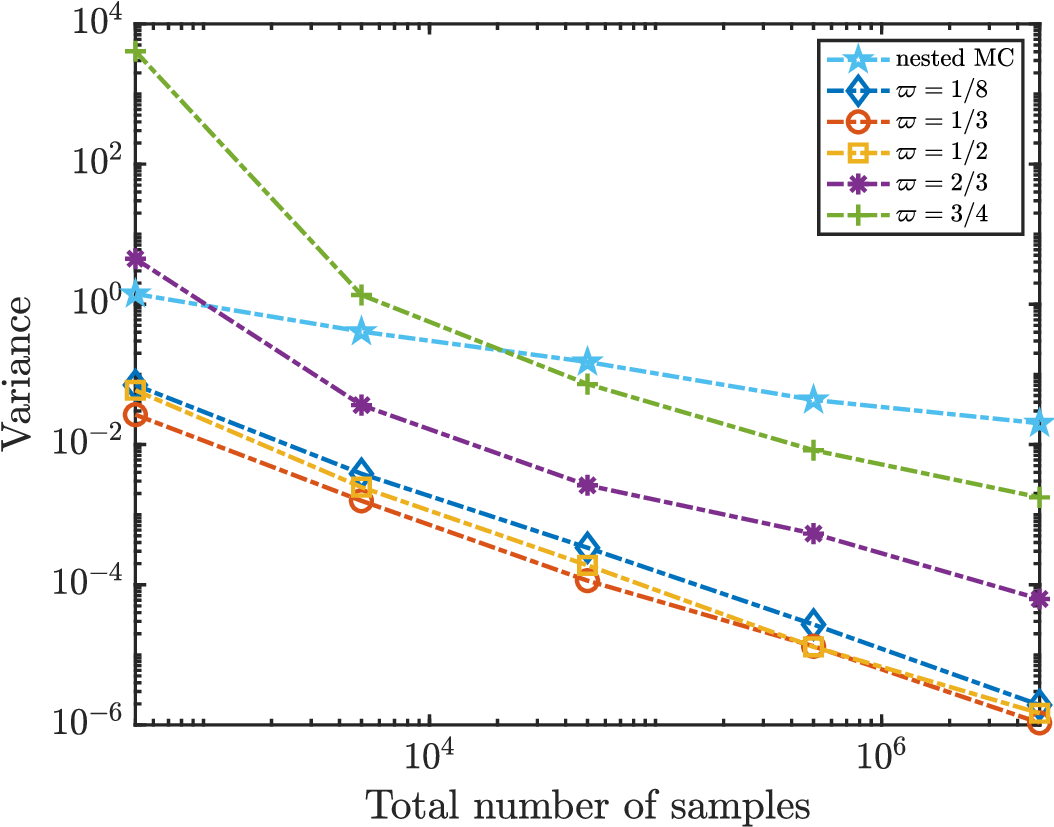}
  \caption{}
\end{subfigure}
\hspace{.5em}
\begin{subfigure}{.31\textwidth}
  \centering
  \includegraphics[width=\linewidth]{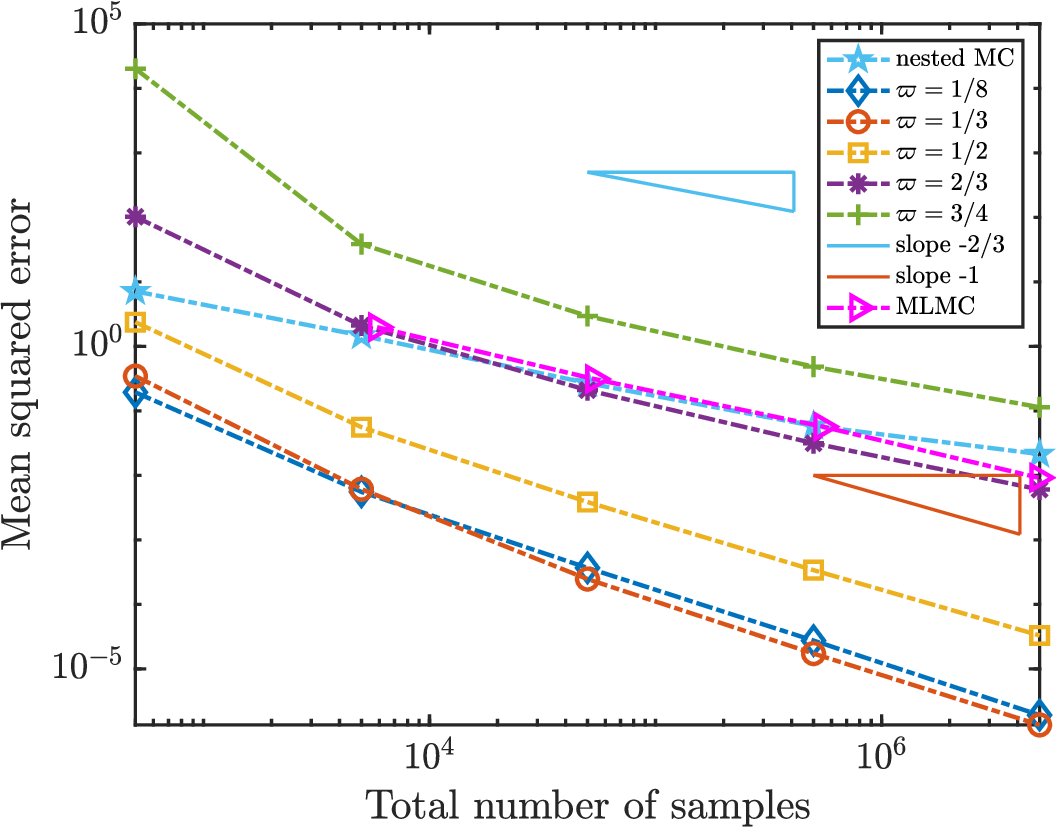}
  \caption{}
\end{subfigure}
\caption{Convergence of the bias, variance, and MSE of $\widehat{\mathrm{EIG}}_{\mathrm{pr}}$, compared also with NMC and MLMC.}
\label{fig:conv_post_over_pr_likfr}
\end{figure}

\subsection{Transport-based estimators of focused and full EIG}\label{subsec:nonlinear_ex}
We next study a nonlinear problem, based on M\"ossbauer spectroscopy, to demonstrate the flexibility of transport maps. M\"ossbauer spectroscopy is a technique that leverages the M\"ossbauer effect, which refers to the nearly recoil-free emission and absorption of gamma radiation in solids. This technique is used to discover isomer shifts, quadrupole splitting, and magnetic Zeeman splitting. A detailed description of the model can be found in~\cite{chi} and~\cite{mossbauer1958}. The parameters in this model are ``center,'' ``width,'' ``height,'' and ``offset.'' The absorption peak is modeled using a Lorentzian profile, as per the standard parameterization described in~\cite{mossbauer_model}. This allows the number of detector counts $y_i$ at a given velocity $d$ to be expressed by the following equation, 
\begin{align*}
    Y_i(d;\cdot) = \mathrm{offset} - \mathrm{height}\frac{\mathrm{width}^2}{\mathrm{width}^2 + (\mathrm{center} - d)^2} + \mathcal{E}_i,
\end{align*}
and the parameters of the model are distributed according to the following priors:
\begin{align*}
    \mathrm{center} &\sim \mathcal N(0,1), \\ \log(\mathrm{width}) &\sim \mathcal N(0,0.3^2),\\
    \log(\mathrm{height}) &\sim \mathcal N(0,0.3^2),\\ \log(\mathrm{offset}) &\sim \mathcal N(1.0,0.2^2).
\end{align*}
The additive noise is assumed to be Gaussian, with $\mathcal{E}_i \sim \mathcal{N}(0, 0.1^2)$, and $d = [-1.3, 0, 1.3]$. For convenience, we rename the center, width, height, and offset parameters as $X_1$, $X_2$, $X_3$, and $X_4$, respectively. First, we note that this model is highly nonlinear, as evidenced by the joint distribution of the samples. While all the $X_i$ (or their logarithms) are jointly Gaussian, the joint distribution of the $Y_i$ and the joint distribution of the $Y_i$ and $X_i$ depart significantly from Gaussianity (see Figure~\ref{original_data}).

We consider the problem of estimating the following two quantities: $\I(X_1; Y)$ and $\I(X; Y)$, which we refer to as the ``focused'' and ``full'' cases, respectively. In each case, we apply both likelihood-based and likelihood-free methods.
\fl{For the focused case, we study the following three EIG estimators:
\begin{align*}    \widehat{\mathrm{EIG}}^{X_1}_{\mathrm{pos}}&=\frac 1M \sum_{i = 1}^M\log \frac{\widehat \pi_{X_1|Y}( x_1^i| y^i)}{\pi_{X_1}(x_1^i)},\\
\widehat{\mathrm{EIG}}^{X_1}_{\mathrm{pr}} &= \frac 1M \sum_{i = 1}^M\log \frac{\widehat \pi_{X_1|Y}(x_1^i| y^i)}{\widehat \pi_{X_1}(x_1^i)},\\
\widehat{\mathrm{EIG}}^{X_1}_{\mathrm{lik}} &= \frac 1M \sum_{i = 1}^M\log \frac{\widehat \pi_{Y|X_1}( y^i|x^i_1)}{\widehat \pi_Y( y^i)}.
\end{align*}}
       

All densities are estimated using $N$ pairs of samples, using the adaptive transport procedure described in Section~\ref{subsubsec:ATM}. We set the ``ground truth'' to be the EIG estimate obtained using NMC. To obtain a consistent estimator of EIG using Monte Carlo for the focused case, we employed layered multiple importance sampling (LMIS) from~\cite{chi} with $1.58 \times 10^6$ samples (5623 for the outer loop and 281 for the inner loops). For the full case, we used the NMC algorithm with $10^{10}$ samples, where the allocation of samples between the inner and outer loops is determined using the optimal allocation proposed in~\cite{chi}. In this case, we have $3.83\times 10^6$ samples for the inner loop and 2610 for the outer loop. We repeated the procedure 10 times to calculate the mean and variance of the EIG estimates obtained using LMIS and NMC. The means of the Monte Carlo estimates in the focused and full cases are $1.54$ and $4.52$ respectively, with corresponding variances of $4.73\times 10^{-5}$ and $4.59 \times 10^{-4}$.
We then obtained EIG estimators using transport maps with an increasing total number of samples: 500, 5000, and 50000, and with the number of training and evaluation samples split according to the optimal sample allocation $M/N \sim \mathcal O(L^{1/3})$, as derived in the previous section. We repeated the experiments and obtained $30$ transport map-based EIG estimates. We present the numerical results of both cases in Figures~\ref{nonlinear1} and~\ref{nonlinear2}.

From the plots, we observe that as the number of samples increases, the EIG estimates converge to the ground truth and the variance decreases considerably in both cases. The performance in the focused case is generally better since we only need to estimate a 4-dimensional map, whereas in the full case, we need to estimate a 7-dimensional map. 
\revise{When estimating the EIG using Monte Carlo methods, the focused case is more computationally expensive, as additional Monte Carlo estimation is required to approximate the conditional likelihood.} In Figure~\ref{comparison}, we see that both TM estimators converge. Since NMC is positively biased, and since $\widehat{\mathrm{EIG}}_{\mathrm{m}}$ and $\widehat{\mathrm{EIG}}_{\mathrm{pos}}$ are unbiased estimates of upper and lower bounds, respectively, the true EIG is likely between $4.1$ and $4.5$.

Finally, to illustrate the accuracy of the triangular transport maps yielding the density estimates used in the EIG calculations for this example, Figure~\ref{fig:push_forward_normal} illustrates the pushforward of the original joint distribution $\pi_{X,Y}$ by an estimated transport map $\widehat{S}$ (where the samples used to estimate the map are independent from the samples that are pushed forward). 
This pushforward distribution appears very close to the standard normal reference distribution $\eta$. To further illustrate the Gaussianity of the pushforward samples,  Figure~\ref{fig:gaussian_quantiles} shows Gaussian quantile-quantile (Q–Q) plots of projected pushforward samples. To be precise, we generate 20 random vectors of unit length in $\R^{n_x + n_y}$ (where $n_x + n_y = 7$), and project samples from the pushforward distribution onto each of these 20 random directions. We then plot the quantiles of these random projections against those of the standard normal distribution. If the pushforward distribution were exactly Gaussian, the Q–Q plots of all such random projections would align with the black line (see Box 2.4 in \cite{Santambrogio}). As shown in Figure~\ref{fig:gaussian_quantiles}, only the tails of the push-forward samples deviate slightly, demonstrating the effectiveness of transport maps in transforming highly non-Gaussian samples into a standard Gaussian distribution.

\begin{figure}
\begin{center}
      \includegraphics[width = 4in]{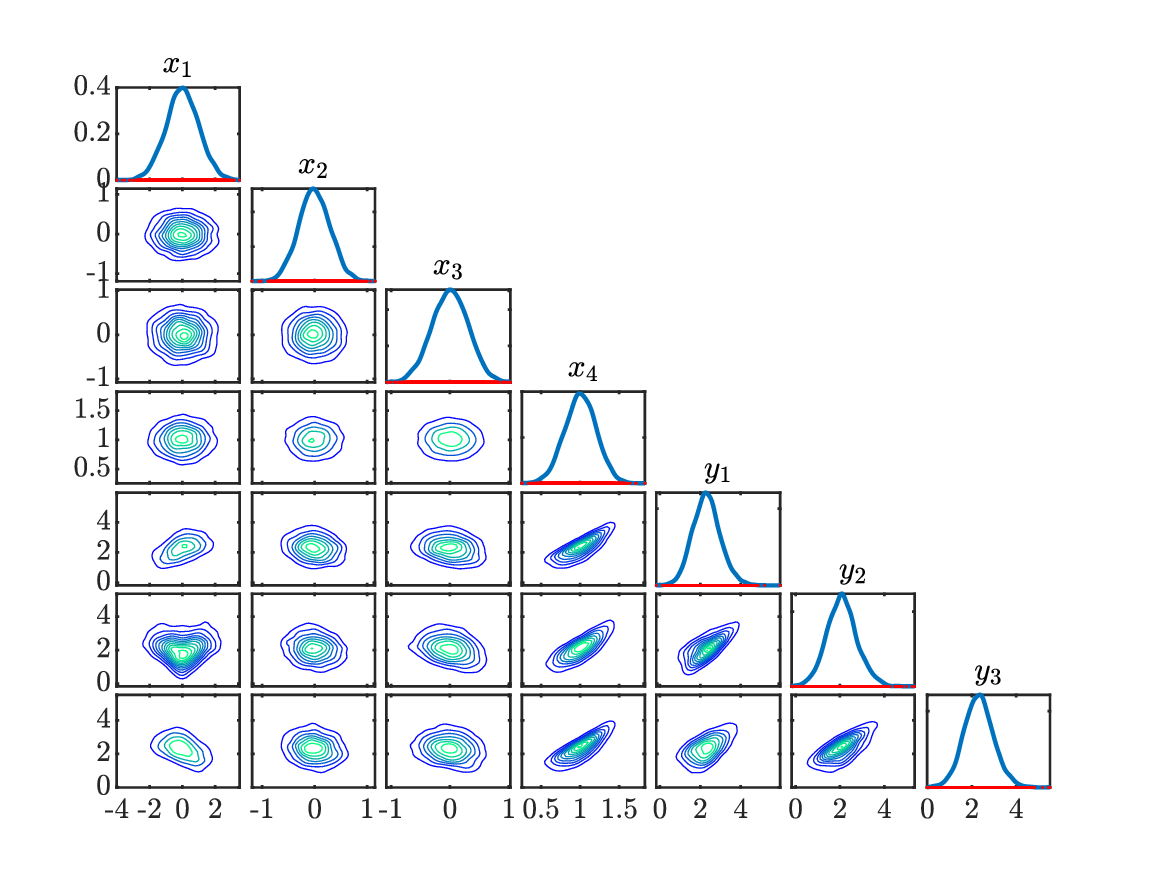}
      \caption{M\"ossbauer example. Samples from the joint distribution $\pi_{X,Y}$.}
      \label{original_data}
\end{center}
\end{figure}

\begin{figure}[!ht]
\centering
\begin{subfigure}{.33\textwidth}
  \centering
  \includegraphics[width=\linewidth]{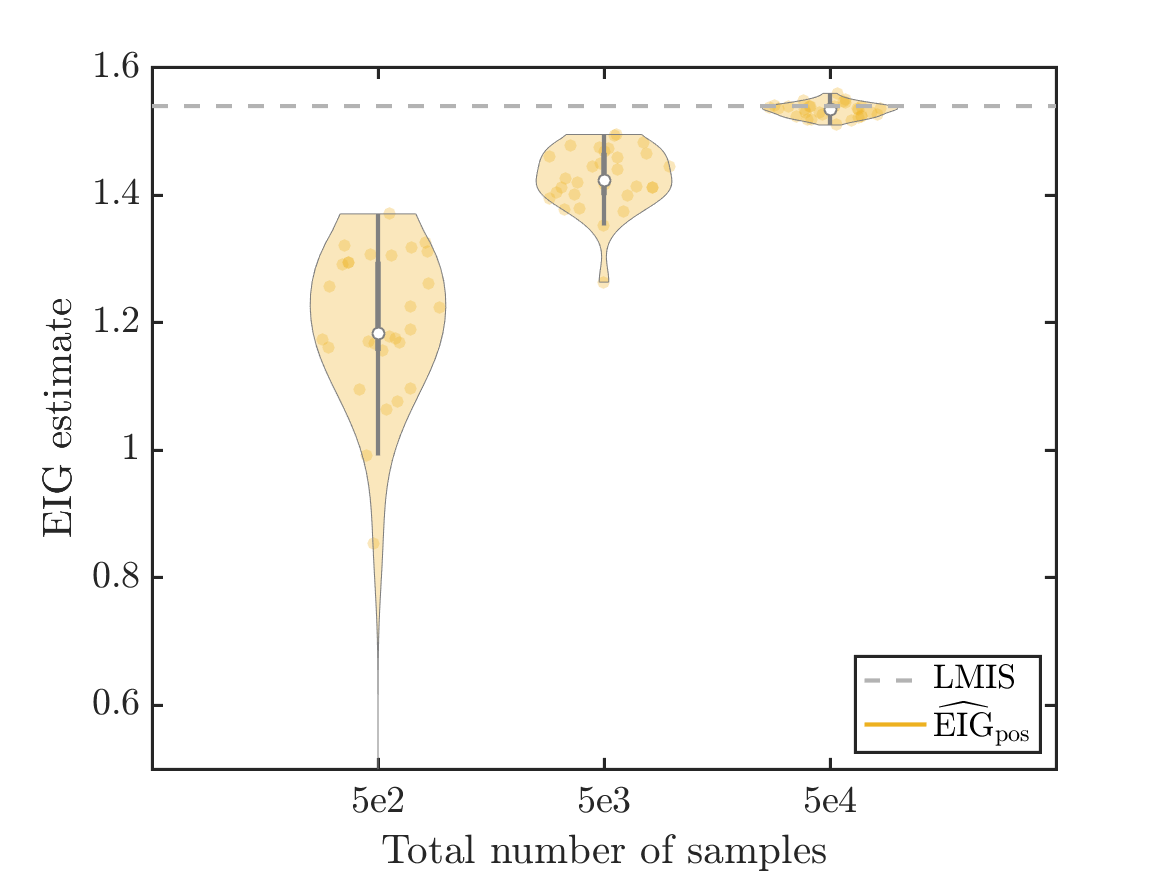}
  \caption{}
\end{subfigure}%
\begin{subfigure}{.33\textwidth}
  \centering  \includegraphics[width=\linewidth]{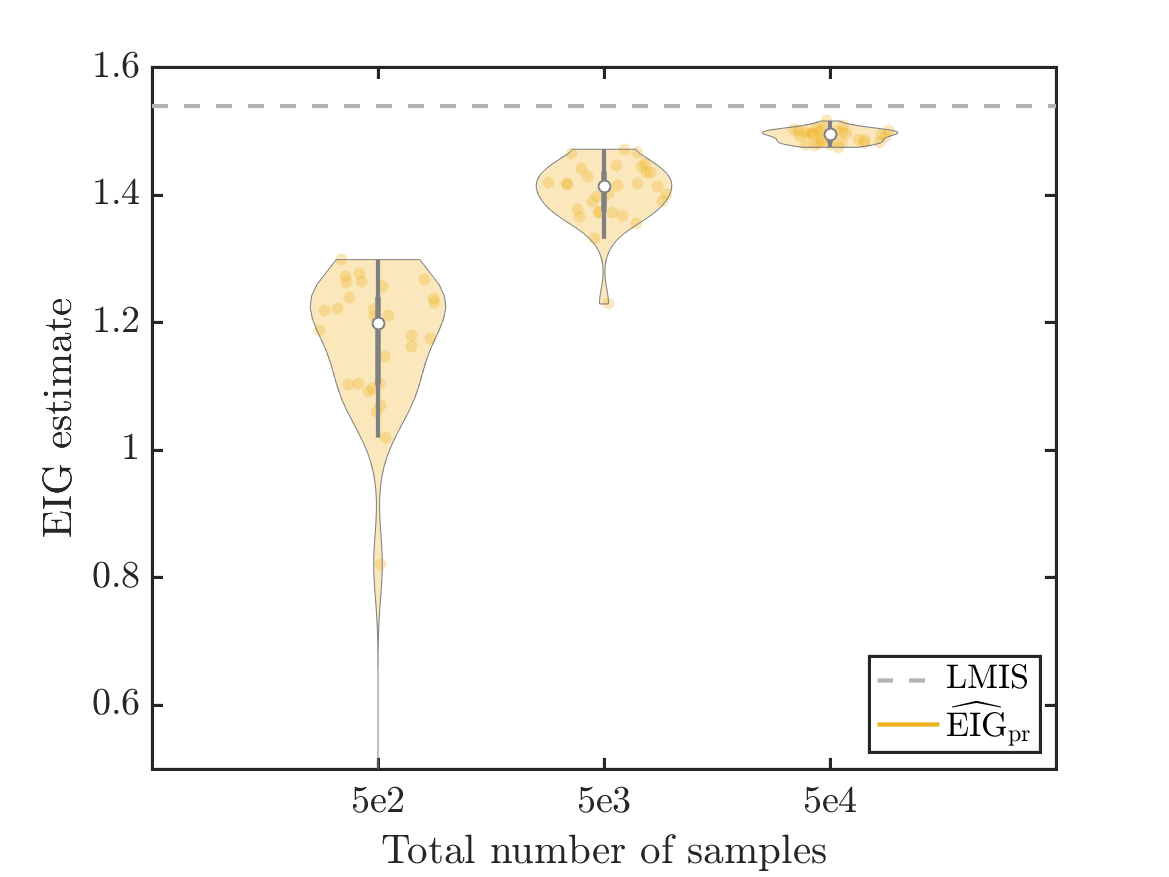}
  \caption{}
\end{subfigure}
\begin{subfigure}{.33\textwidth}
  \centering
 \includegraphics[width=\linewidth]{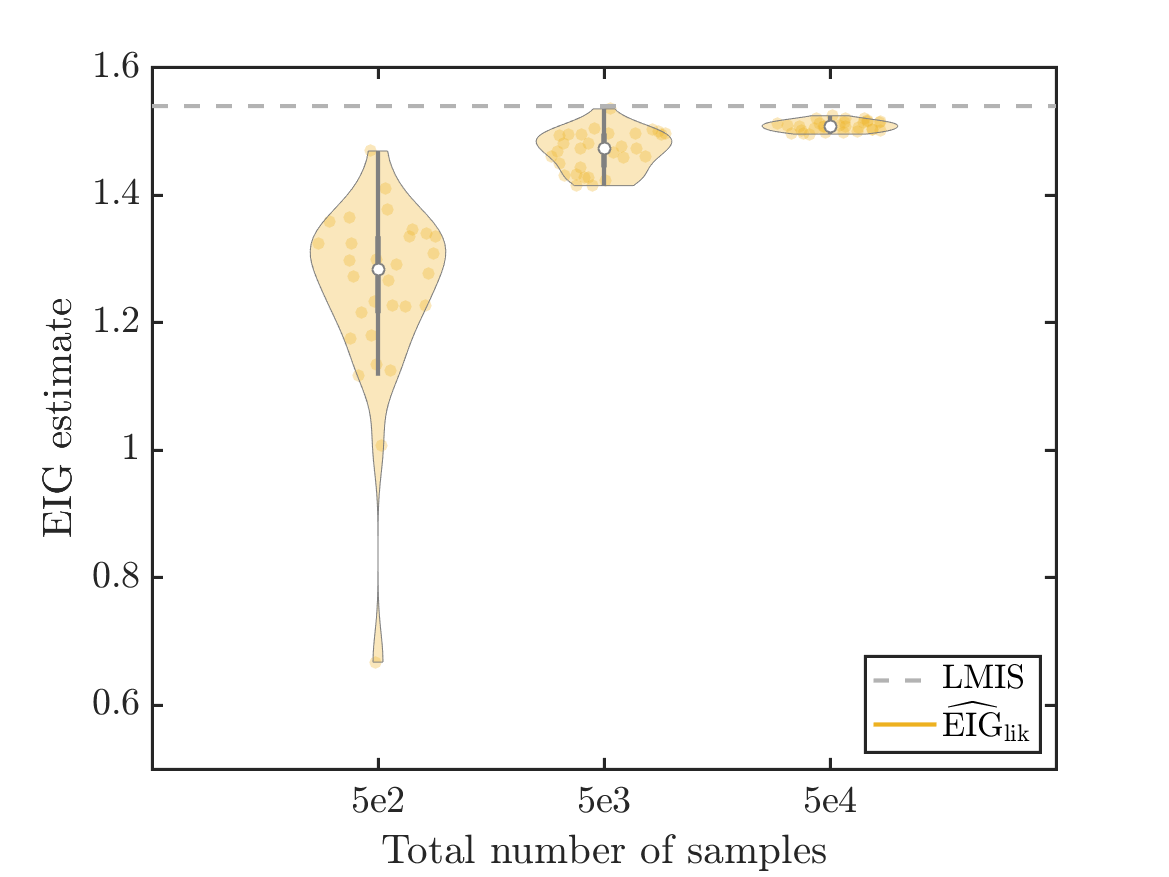}
  \caption{}
\end{subfigure}
\caption{M\"ossbauer example. Convergence of different EIG estimators in the focused case: (a) $\widehat{\mathrm{EIG}}^{X_1}_{\mathrm{pos}}$; (b) $\widehat{\mathrm{EIG}}^{X_1}_{\mathrm{pr}}$; (c) $\widehat{\mathrm{EIG}}^{X_1}_{\mathrm{lik}}$. The dashed line represents the ``ground truth'' obtained with LMIS.}
\label{nonlinear1}
\end{figure}

\begin{figure}[!ht]
\centering
\begin{subfigure}{.33\textwidth}
  \centering
  \includegraphics[width=\linewidth]{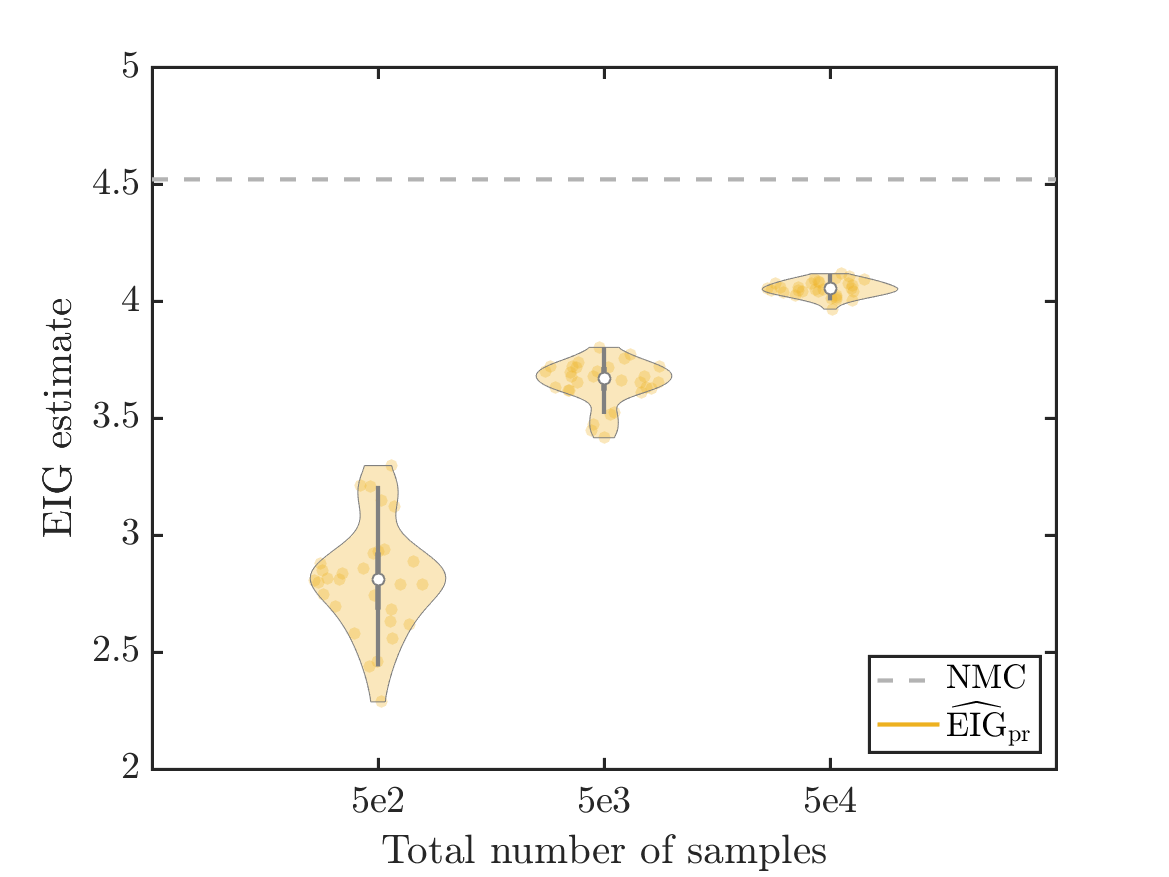}
  \caption{}
\end{subfigure}%
\begin{subfigure}{.33\textwidth}
  \centering
  \includegraphics[width=\linewidth]{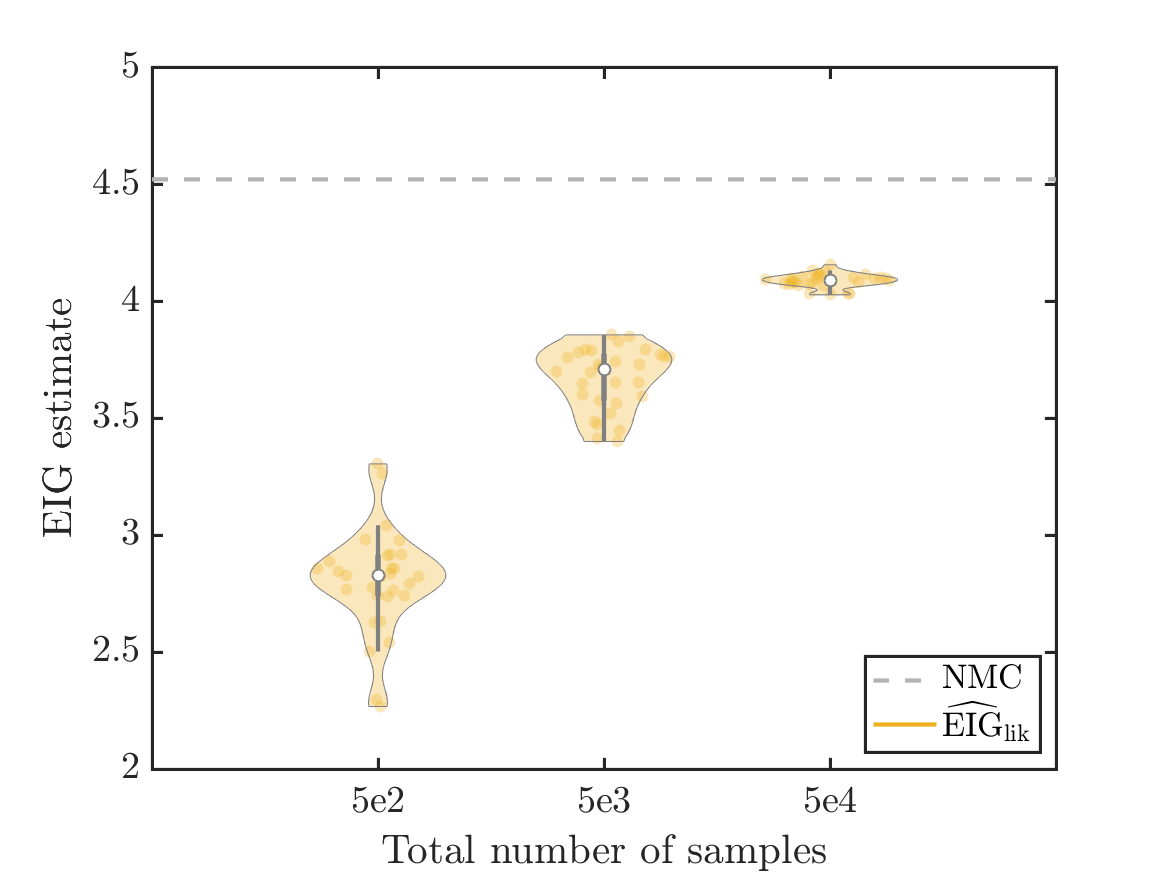}
  \caption{}
\end{subfigure}
\begin{subfigure}{.33\textwidth}
  \centering
  \includegraphics[width=\linewidth]{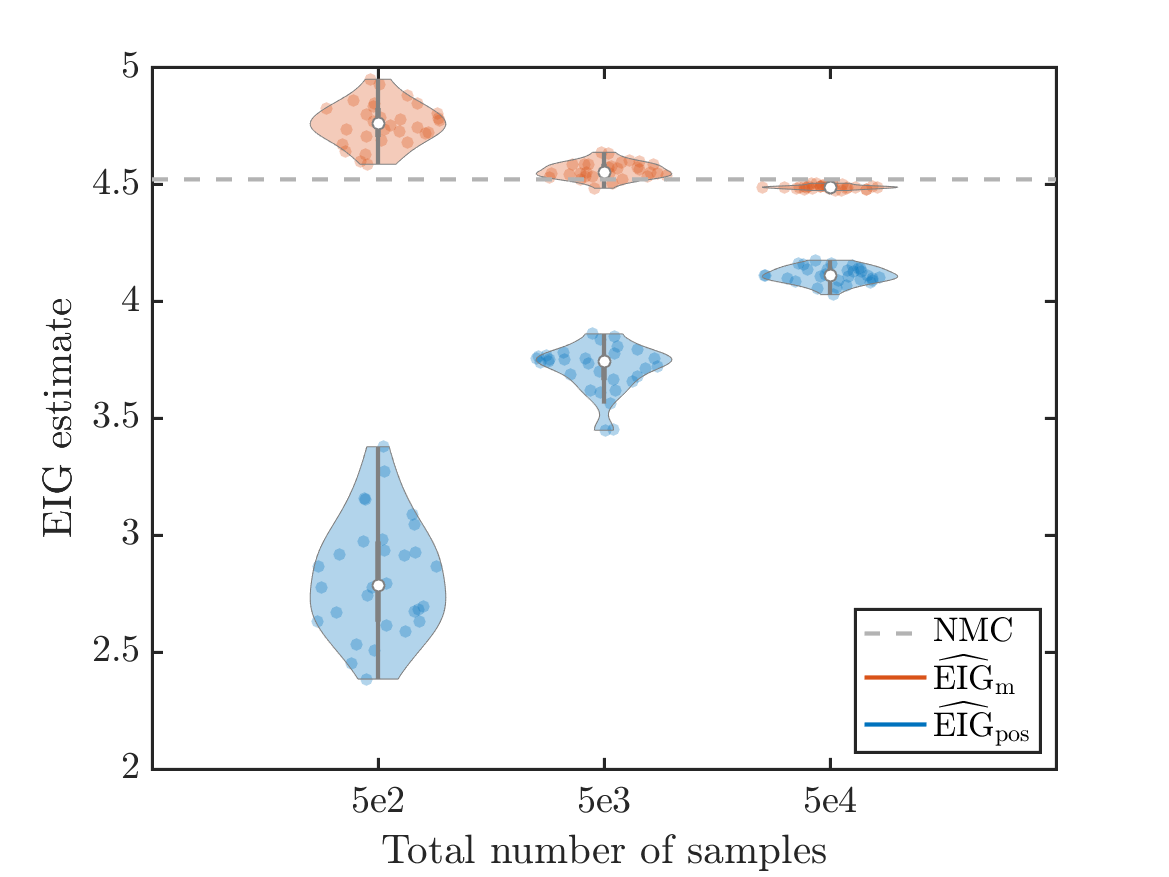}
  \caption{}
  \label{comparison}
\end{subfigure}
\caption{M\"ossbauer example. Convergence of different EIG estimators in the full case: (a)  $\widehat{\mathrm{EIG}}_{\mathrm{pr}}$; (b) $\widehat{\mathrm{EIG}}_{\mathrm{lik}}$; (c) $\widehat{\mathrm{EIG}}_{\mathrm{m}}$ and $\widehat{\mathrm{EIG}}_{\mathrm{pos}}$. The dashed line represents the (biased) ``ground truth'' obtained with NMC.}
\label{nonlinear2}
\end{figure}

\begin{figure}
\begin{subfigure}{.5\textwidth}
  \centering
  \includegraphics[width=\linewidth]{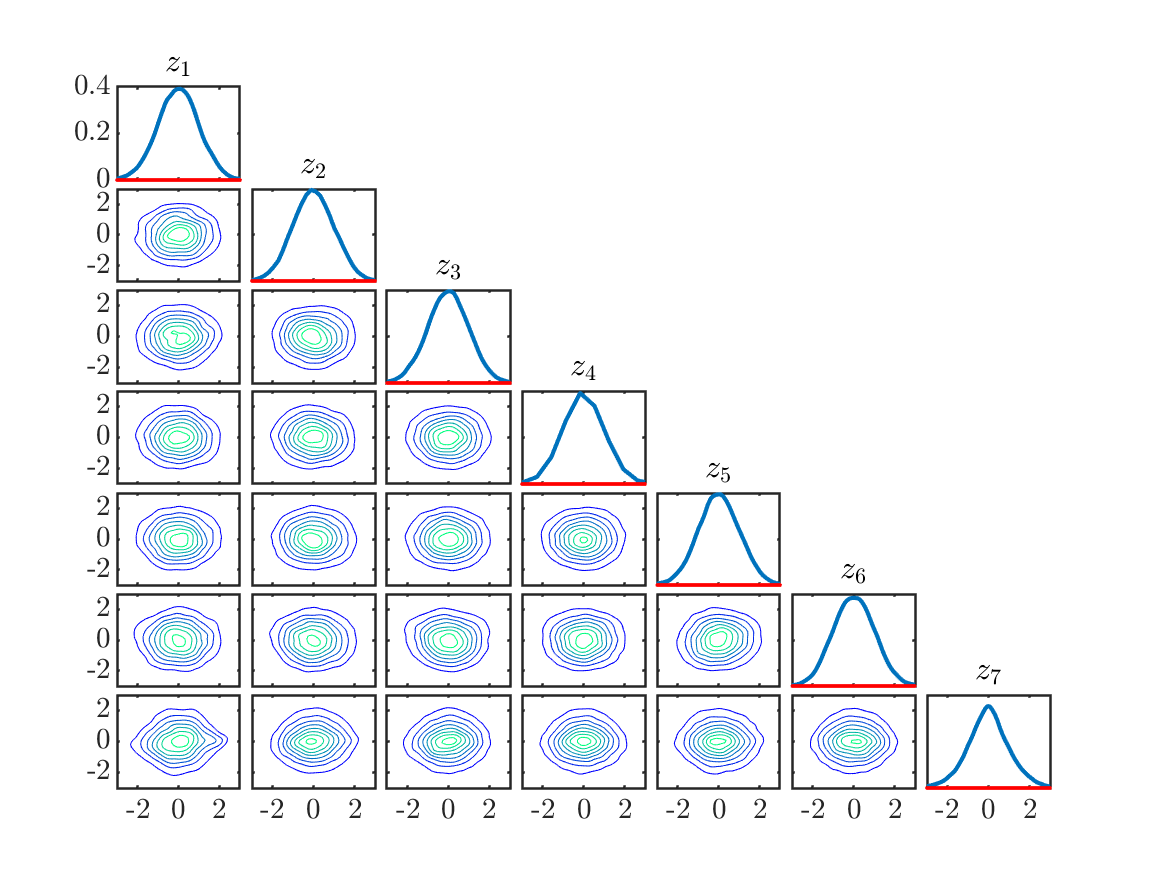}
  \caption{}
  \label{fig:push_forward_normal}
\end{subfigure}%
\begin{subfigure}{.5\textwidth}
  \centering
  \includegraphics[width=\linewidth]{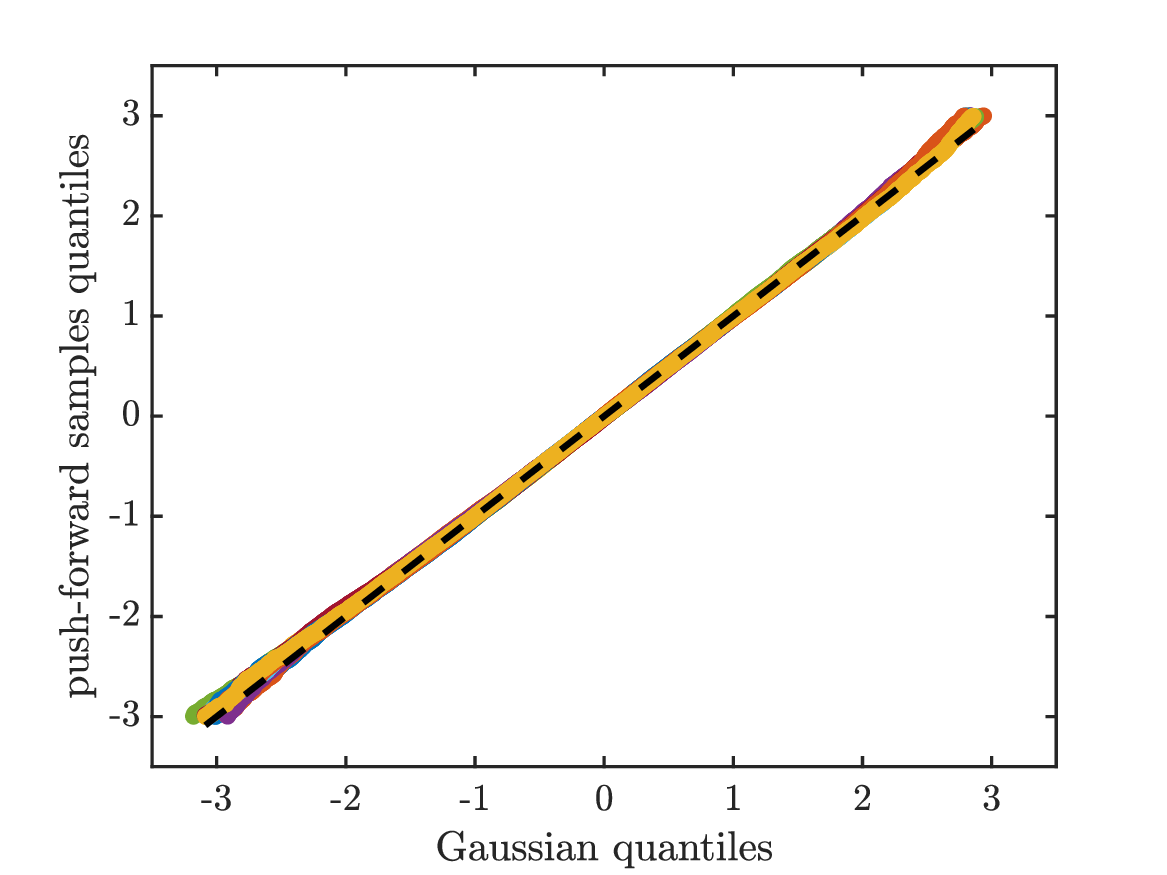}
  \caption{}
  \label{fig:gaussian_quantiles}
\end{subfigure}
\caption{M\"ossbauer example. (a) Density of the pushforward of the joint prior distribution under an estimated transport map, $\widehat{S}_\sharp \pi_{X,Y}$. (b) Q-Q plot of randomly projected pushforward marginals.}
\label{fig:push_forward_samples}
\end{figure}

\subsection{EIG estimation in high dimensions: linear elasticity inverse problem}\label{subsec:wrench_mark}
In this problem, we compare the proposed dimension reduction method for EIG estimation, referred to as conditional mutual information (CMI) dimension reduction, to alternative methods based on principal component analysis (PCA)~\cite{HotellingPCA} and canonical correlation analysis (CCA)~\cite{Hotelling36}. 
We consider the problem of estimating \revise{MI} between the Young's modulus field in a wrench-shaped domain and a noisy observation of its displacement on the upper boundary after applying an external force to a subset of the boundary $\partial \Omega$. Following the setup in~\cite{Ricarod_dimRed, remi_dimred, olivier_wrench}, the displacement $u: \Omega \rightarrow \mathbb{R}^2$ is related to the Young's modulus $E: \Omega \rightarrow \mathbb{R}_{\geq 0}$ through the partial differential equation
\begin{align*}
\nabla \cdot (K:\epsilon(u)) = 0,
\end{align*}
where
\begin{align*}
K:\epsilon(u) = \frac{E}{1 + \nu} \epsilon(u) + \frac{\nu E}{1 - \nu^2} \text{Tr}(\epsilon(u)) I_2,
\end{align*}
and $\nu = 0.3$ is the Poisson's ratio. The wrench is clamped at its right end, experiencing no displacement, i.e., $u = 0$. The Young's modulus field is assumed to have a log-normal prior distribution, specifically $\log E \sim \mathcal{N}(0, C)$, where $C$ is a squared exponential covariance kernel given by $C(z_1, z_2) = \exp(-\left\Vert z_1 - z_2\right\Vert^2)$ with $z_1, z_2\in \Omega$. The domain is discretized with 925 elements using the finite element method, resulting in $X\in \mathbb R^{925}$. The observation operator $D$ maps the Galerkin solution $u^h$ to a subset of 26 dimensions. We then write $Y = G(X) + \mathcal{E}$, where $G = Du^h(X)$, and $\mathcal{E} \sim \mathcal{N}(0, DR^{-1}D^\top)$, with $R$ being the Riesz map associated with the $H^1(\Omega)$ norm, such that $\Vert u^h \Vert_R^2 = \int_{\Omega} (u^h(s))^2 + \Vert \nabla u^h(s) \Vert^2 ds$~\cite{olivier_thesis}. Figure~\ref{wrench_para_sol} displays a realization of the parameter and the solution field.

\begin{figure}
\begin{subfigure}{.5\textwidth}
  \centering
  \includegraphics[trim={1cm 4cm 1cm 4cm},clip,width=\linewidth]{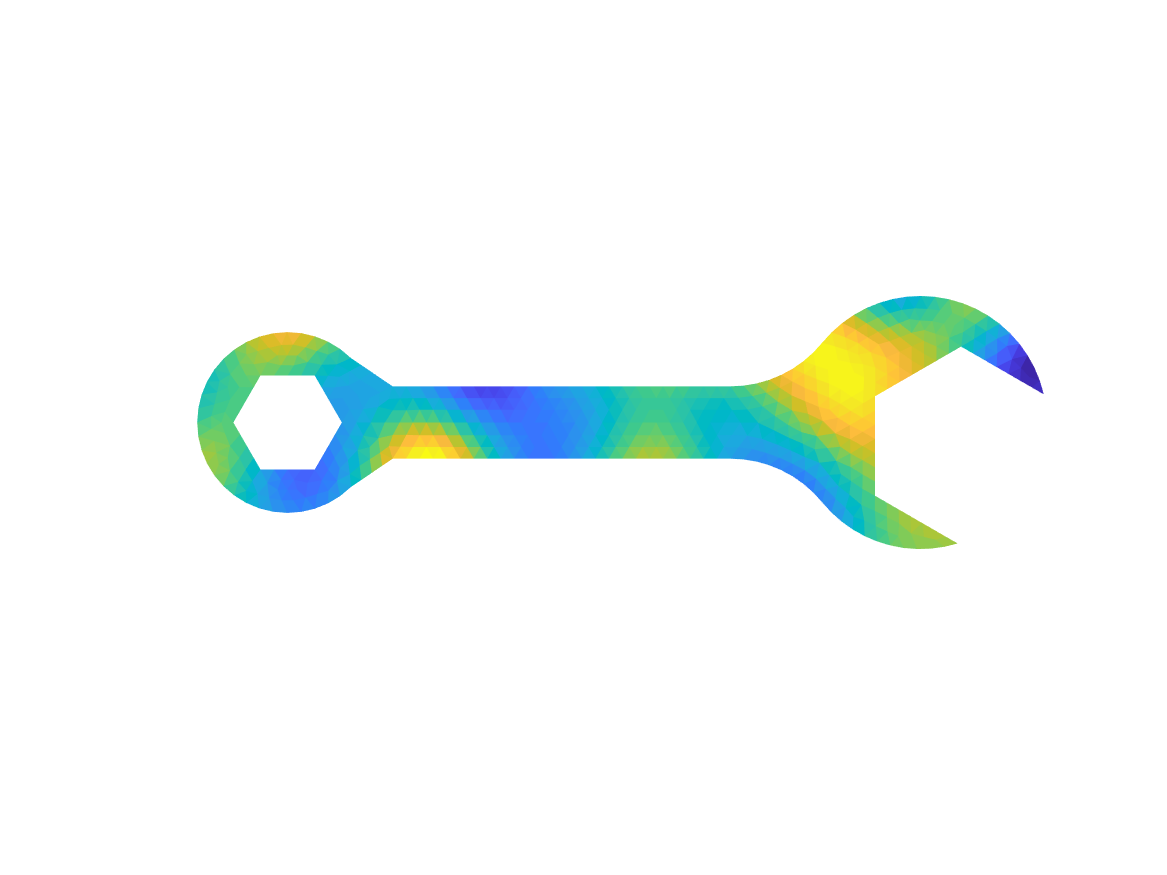}
  \caption{}
\end{subfigure}%
\begin{subfigure}{.5\textwidth}
  \centering
  \includegraphics[trim={1cm 4cm 1cm 4cm},clip,width=\linewidth]{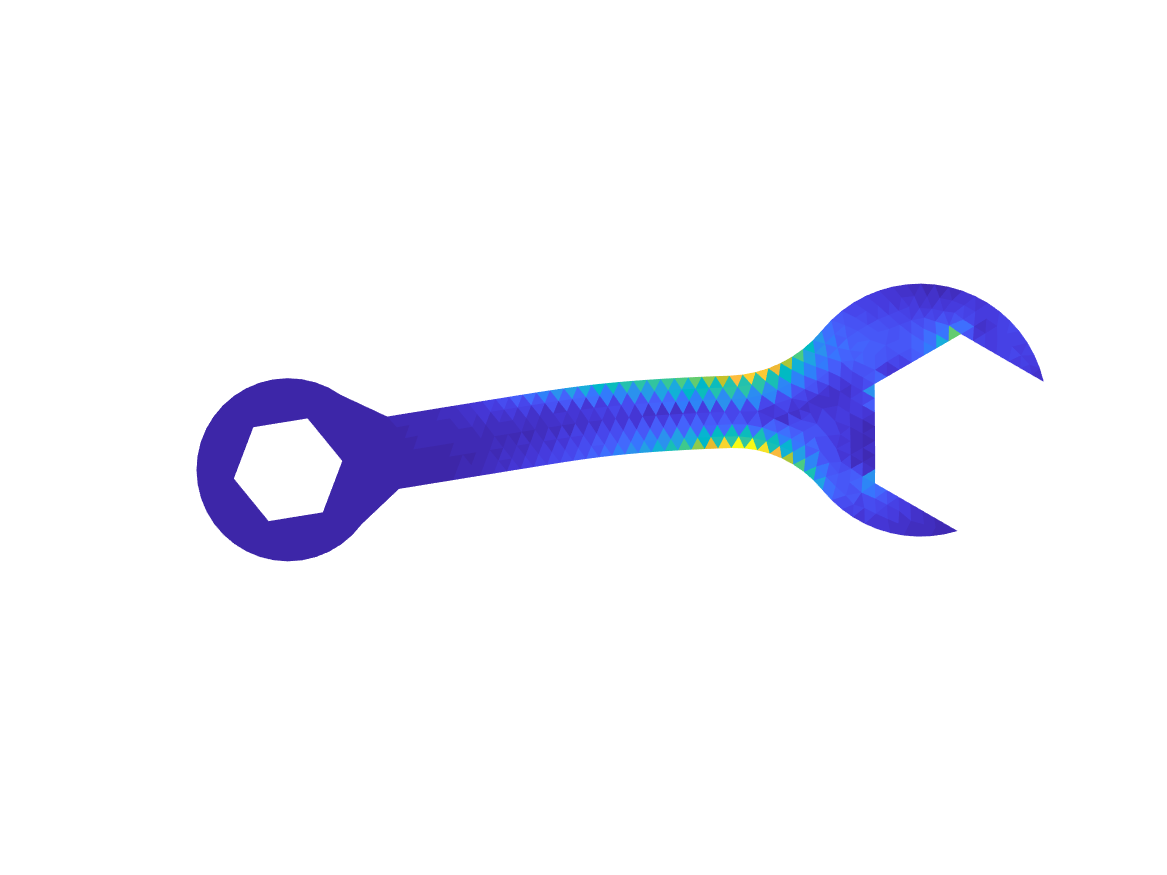}
  \caption{}
\end{subfigure}
\caption{Elasticity example. A realization of the parameter and solution field}
\label{wrench_para_sol}
\end{figure}
To compare the EIG estimators using different dimension reduction schemes, we use 500 samples to estimate $H_X$ and $H_Y$, and then compute the eigenvectors associated with these two matrices. Additionally, we compute the eigenvectors obtained through PCA and CCA. Recall that the principal components associated with PCA solve the following eigenvalue problems,
\begin{align*}
    \Sigma_X u^{\PCA} &= \lambda_X u^{\PCA}\\
    \Sigma_Y v^{\PCA} &= \lambda_Y v^{\PCA}.
\end{align*}
One the other hand, CCA finds the maximal correlation over all
possible linear transformations of the two random vectors $X$ and $Y$. The corresponding eigenvalue problems are:
\begin{align*}
    \Sigma_{XY} \Sigma_Y^{-1} \Sigma_{YX} u^{\CCA} = \rho \Sigma_X u ^{\CCA}\\
    \Sigma_{YX} \Sigma_X^{-1} \Sigma_{YX} v^{\CCA} = \rho \Sigma_Y v^{\CCA}.
\end{align*}
We then transform the samples $x^i$ and $y^i$ using these eigenvectors 
$\{u^{\CMI}, v^{\CMI} \}, \{u^{\PCA}, v^{\PCA} \}$, and $\{u^{\CCA}, v^{\CCA} \}$, respectively, where $\{u^{\CMI}, v^{\CMI} \}$ is computed following Algorithm~\ref{alg:EIG_alg_highDim}. Transport maps of input dimension $r+s$ are trained using $N = 2500$ samples. We then use $M = 10000$ evaluation samples, also projected as described in Section~\ref{sec:high_dim}, to compute the estimator $\widehat{\textrm{EIG}}_{\mathrm{pos}}(r,s) = \frac{1}{M} \sum_{i=1}^M \frac{\widehat \pi_{X_r|Y_s}(x_r^i|y_s^i)}{\pi_{X_r}(x_r^i)}$. To illustrate the effectiveness of our proposed nonlinear dimension reduction method, we also construct a simple ``Gaussian approximation'' of the EIG, i.e., an EIG estimate that would be consistent if $X$ and $Y$ were jointly Gaussian. Specifically, we approximate the EIG using the formula $\frac 12 \log \frac{\det(\widehat \Sigma_X)\det(  \widehat \Sigma_Y)}{\det( \widehat \Sigma_{X,Y})}$, where $\widehat \Sigma_X$, $\widehat \Sigma_Y$, and $\widehat \Sigma_{X,Y}$
denote the marginal and joint sample covariances of $X$ and $Y$. These sample covariance matrices are constructed from $10000$ samples. For reproducibility, the same experiment is repeated $10$ times and we show the error bars, representing two standard errors of the mean, in Figure~\ref{EIG_compare}.

The ground truth EIG value in this case is unknown. However, we know that the EIG estimators used here are \textit{lower bounds} to the true EIG. Therefore, the best EIG estimator should dominate the others---i.e., larger values are closer to the truth. We report the mean of each EIG estimator and its associated error in Appendix~\ref{wrench_table}. In Figure~\ref{EIG_compare}, we plot results with $r = s = 1$ to $r = s = 8$. As we can see, using CMI for dimension reduction yields the highest values and hence the best performance. At any given level of dimension reduction and for all three reduction schemes, results obtained with the Gaussian approximation are worse than their counterparts obtained using transport---illustrating the importance of treating non-Gaussianity and the ability of transport maps to do so. 

To expand this comparison, we also compute the EIG using the Gaussian approximation while holding the dimension of $X_r$ constant at $r = 300$ and varying the dimension of $Y_s$ from $s = 1$ to $s = 26$. Then we hold the dimension of $Y_s$ constant at $s = 26$ and vary the dimension of $X_r$ from $r = 1$ to $r = 300$. These results are presented in Figure~\ref{gaussian_1_26_1_300}, with the shaded region indicating two standard errors of the mean. We see that the performance of the Gaussian approximation remains inferior, even when using CMI and the retained dimension is high: the highest values in Figure~\ref{gaussian_1_26_1_300} are roughly 3.0, whereas Figure~\ref{EIG_compare} shows that the true EIG must be at least 5.0. This comparison further suggests the necessity of combining good dimension reduction techniques with non-Gaussian methods for EIG estimation: \textit{one can achieve better performance with a more expressive distribution representation and significantly more dimension truncation, than with a cruder information gain approximation and very little truncation.} It is also interesting to note that the benefit of the non-Gaussian transport representation is greatest when the subspaces for dimension reduction are obtained via CMI. This suggests that the non-Gaussian interaction of the parameters and data is well captured by such subspaces; see \cite{brennan2020greedy,li2025sharp} for comments on this point.


\begin{figure}[!ht]
\begin{center}
      \includegraphics[width = 4in]{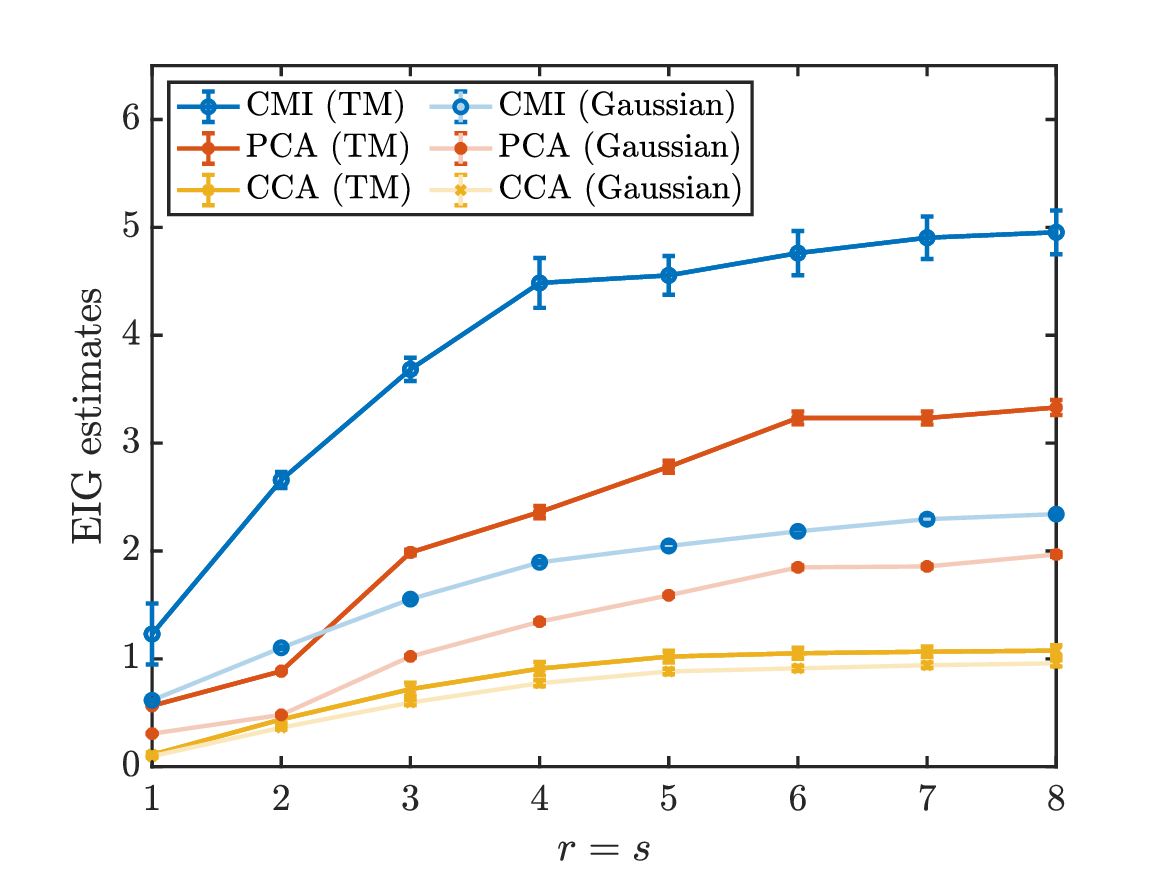}
      \caption{Elasticity example. EIG estimates using transport maps (darker lines) or a Gaussian approximation (lighter lines), with parameter and data dimensions truncated using CMI, PCA, or CCA with $(r = s) \in \{1, \ldots, 8\}$.}
      \label{EIG_compare}
\end{center}
\end{figure}

\begin{figure}[!ht]
\centering
\begin{subfigure}{.48\textwidth}
  \centering
  \includegraphics[width=\linewidth]{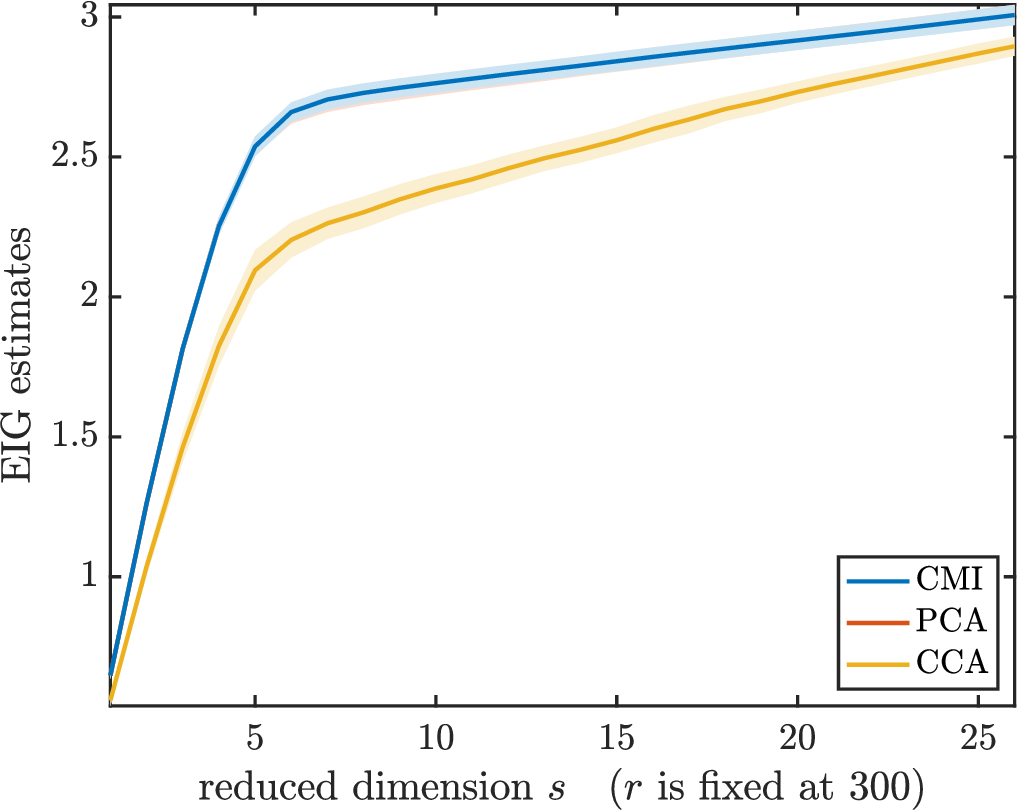}
  \caption{}
\end{subfigure}%
\hspace{5pt}
\begin{subfigure}{.48\textwidth}
  \centering
  \includegraphics[width=\linewidth]{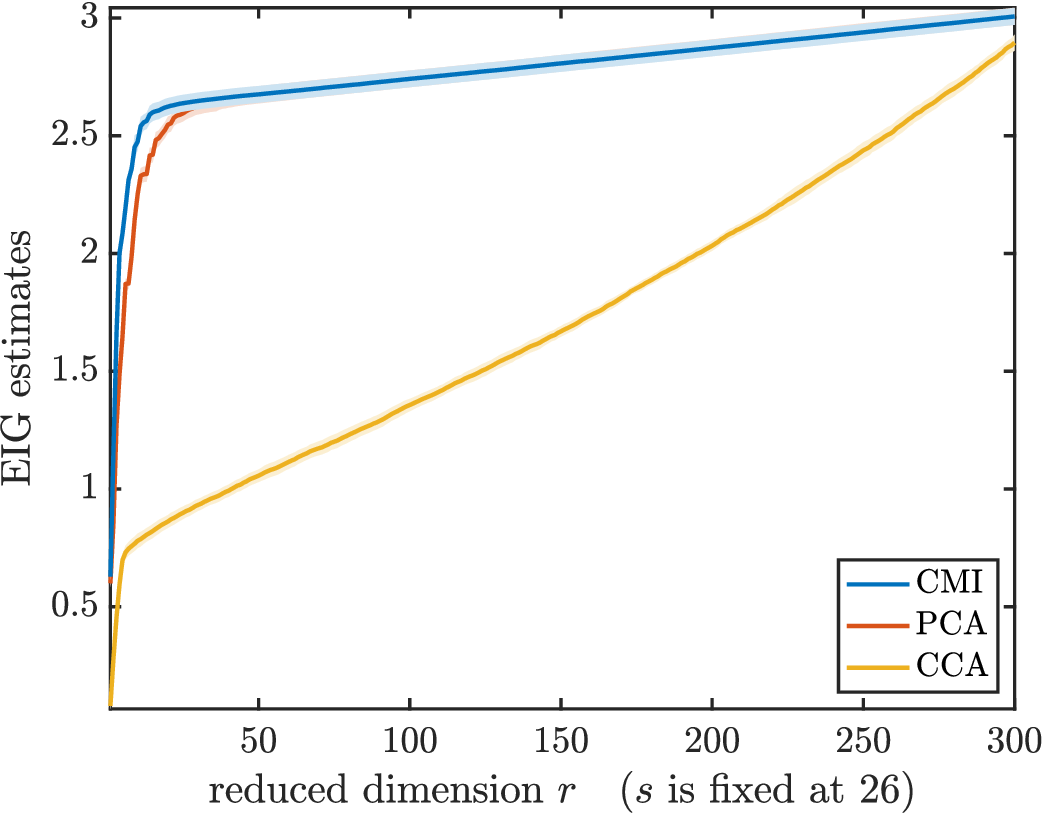}
  \caption{}
\end{subfigure}
\vspace{-0.4cm}
\caption{Elasticity example. EIG estimates using a Gaussian approximation, with parameter dimension $r$ or data dimension $s$ truncated using CMI, PCA, or CCA. \revise{(The PCA and CMI results are nearly identical in (a), with an absolute difference of less than 0.01.)}}
\label{gaussian_1_26_1_300}
\end{figure}

\section{Conclusion} \label{sec:conclusion}
\revise{We have studied a key problem in nonlinear Bayesian optimal experimental design: EIG estimation. We introduce a general transport-based framework for EIG estimation, which encompasses a wide variety of block-triangular map parameterizations (ranging from strictly triangular Knothe--Rosenblatt maps to conditional optimal transport to conditional normalizing flows) and enables estimation of EIG with both explicit and implicit models. Within this framework, we study the optimal allocation between training samples (used for marginal and/or conditional density estimation) and evaluation samples (used to approximate an outer prior expectation) and identify an asymptotic scaling that minimizes the MSE of the resulting EIG estimates. To improve the estimation of EIG in high dimensions, we use the gradient-based dimension reduction techniques of~\cite{Ricarod_dimRed} to construct ``projected'' estimators that can be understood as principled and efficient EIG lower bounds. Numerical examples verify the optimality of our proposed sample allocation and show improved efficiency over standard NMC and MLMC estimators of EIG. Numerical experiments also demonstrate the importance of capturing non-Gaussianity and the ability of transport to do so. We see that in high-dimensional problems, the combination of transport-based density estimation and our dimension reduction method significantly outperforms dimension reduction with PCA or CCA, as well as full-dimensional Gaussian approximations.}



\revise{For future work, we aim to extend our analysis in two key directions. First, we plan to investigate the non-asymptotic error of EIG estimates based on {nonparametric} estimation of transport maps, marrying the approaches introduced here with results in, e.g., \cite{wang2022minimax,irons2022triangular}, and again elucidating optimal sample allocations and even penalization schemes. Second, we will develop transport maps that explicitly incorporate the design variable~$d$. This extension would enable the construction of design-dependent maps that directly account for how the data-generating process varies with experimental conditions---though for the purposes of design optimization, it will be important to devise learning methods that \textit{control} how error in the EIG estimate depends on $d$. Such an approach could unify the map learning and design optimization steps, leading to more efficient end-to-end algorithms for Bayesian optimal experimental design.}

\section*{Acknowledgements}
FL and YM are grateful for support from the US Department of Energy (DOE), Office of Advanced Scientific Computing Research (ASCR), under grant numbers DE-SC0023188 and DE-SC0026245.
RB is grateful for support from the von K\'{a}rm\'{a}n instructorship at Caltech, the Air Force Office of Scientific
Research MURI on ``Machine Learning and Physics-Based Modeling and Simulation'' (award FA9550-20-1-0358), and a Department of Defense (DoD) Vannevar Bush Faculty Fellowship (award N00014-22-1-2790) held by Andrew M.~Stuart.

\bibliographystyle{plain}
\bibliography{references}

\appendix
\section{MLE and the second order delta method}\label{app_second_order_delta}
For completeness, we first state the multivariate second order delta method.
\begin{theorem}[Second Order Delta Method]\label{thm:second_order_delta}
Let $g \colon\mathbb R^{p} \to \mathbb R$ be a twice differentiable function at $\alpha^*$ that satisfies $\nabla g(\alpha^*) = 0$. If $r_N (\widehat{\alpha}_N - \alpha^\ast) \xrightarrow{d} Z$, where $r_N \rightarrow \infty$ as $N \rightarrow \infty$ and $Z$ is a random variable in $R^{p}$ whose distribution is independent of $N$, then we have
\begin{align*}
    r_N^2 (g(\widehat{\alpha}_N) - g(\alpha^\ast) )\xrightarrow{d} \frac{1}{2}Z^\top \nabla^2 g( \alpha^*) Z.
\end{align*}
\end{theorem}

From the theory of MLE and asymptotic normality, we have
\begin{align*}
    \sqrt{N} \left( \widehat{\alpha}_N -  \alpha^* \right) \xrightarrow{d} \mathcal N(0, I^{-1}( \alpha^*)).
\end{align*}
In addition, recall that $g$ is defined as $g( \alpha) = \mathbb E_{\pi_Y} \left [ \log q_Y( y; \alpha) \right]$ and that $\nabla g(\alpha^*) = 0$. Applying Theorem~\ref{thm:second_order_delta}, we have that 
\begin{align*}
    N\left( g( \widehat{\alpha}_N) - g( \alpha^*)\right) \xrightarrow{d} \frac{1}{2} Z^\top \nabla^2 g( \alpha^*) Z,
\end{align*}
where $Z\sim \mathcal N(0, I^{-1}( \alpha^*))$ in this case. We further note that the right hand side is the sum of Chi-square distributions. To be precise, let $\{\lambda_i\}_{i=1}^{p}$ be the generalized eigenvalues of the matrix pencil $\left(I( \alpha^*)^{-1}, \left(\nabla_{ \alpha}^2 g( \alpha^*)\right)\right)$.

Then $\mathbb E\left[N(g( \widehat{\alpha}_N) - g( \alpha^*)) \right]$ converges to $\sum_{i=1}^{p}\lambda_i$, and $\mathbb V\left[ N(g( \widehat{\alpha}_N) - g( \alpha^*)) \right]$ converges to $2\sum_{i=1}^{p} \lambda_i$ with the assumption that $\mathbb E\left[ N(g( \widehat{\alpha}_N) - g( \alpha^*))^2\right] \leq C$ for some $C$. Please refer to Appendix~\ref{app_chisqure} for more details.

\section{Proofs of the main results}\label{proof_of_thms}

\begin{proof}[Proof of Theorem~\ref{thm:EIG_rate}]
\label{proof:EIG_rate}
Recall that EIG can be computed as 
\begin{align*}
    \mathrm{EIG} = \mathbb E_{\pi_{X,Y}} \left[ \log \frac{\pi_{Y|X}( y| x)}{\pi_Y( y)}\right].
\end{align*}
This quantity can be estimated using samples, i.e., $\frac{1}{M}\sum_{i = 1}^M \log \frac{\pi_{Y|X}( y^i| x^i)}{\pi_Y( y^i)}$.

Suppose we can evaluate $\pi_{Y|X}( y| x)$ exactly for any $ x$ and $ y$, and suppose that the distribution is parametrized by $ \widehat{\alpha}_N$. Then the EIG estimator can be written as 
\begin{align*}
    \widehat{\mathrm{EIG}}_{M,N} = \frac{1}{M}\sum_{i = 1}^M \log \frac{\pi_{Y|X}( y^i| x^i)}{q_Y( y^i;  \widehat{\alpha}_N)},
\end{align*}
whose bias is 
\begin{align*}
   \mathbb E\left[  \widehat{\mathrm{EIG}}_{M,N} - \mathrm{EIG} \right] 
   &= \mathbb E_{\widehat{\alpha}_N} \mathbb E_{\pi_Y} \left[  \log \pi_Y( y) -  \log q_Y( y;  \widehat{\alpha}_N)\right] \\ &= \mathbb E_{\pi_Y} \left[  \log \pi_Y( y) \right] -\mathbb E_{\widehat{\alpha}_N}  \mathbb E_{\pi_Y}\left[ \log q_Y( y;  \widehat{\alpha}_N)\right].
\end{align*}
Let $\alpha^*$ be the true parameter that generates the data. That is $\pi_Y(y) = q_Y(y;\alpha^*)$. We define 
\begin{align*}
    g( \alpha) := \mathbb E_{a^*} \left [ \log q_Y( y; \alpha) \right] = \mathbb E_{q_{Y}(y ;a^*)} \left [ \log \revise{q_Y( y; \alpha) }\right] =  \mathbb E_{\pi_Y} \left [ \log \revise{q_Y( y; \alpha)} \right],
\end{align*}
which is a function of $\alpha$ only. Using Taylor expansion, we have 
\begin{align*}
    g( \widehat{\alpha}_N) = g( \alpha^*) + \nabla g( \alpha^*) ( \widehat{\alpha}_N - \alpha^*) + \frac 12 ( \widehat{\alpha}_N -  \alpha^*)^\top \nabla^2 g( \alpha^*) ( \widehat{\alpha}_N -  \alpha^*) + o_p(\left \Vert  \widehat{\alpha}_N -  \alpha^*\right\Vert^2).
\end{align*}
Here we also have $ \alpha^* =  \arg\max_{ \alpha} g( \alpha)$, assuming that $\left| \partial^3 g( \alpha) \right| \leq C$ for all $ \alpha$ \cite{taylor}. Note that by the property of MLE, 
\begin{align*}
    \nabla g( \alpha^*) = 0.
\end{align*}
Then using the asymptotic normality of MLE, we have that 
\begin{align*}
    \sqrt N ( \widehat{\alpha}_N -  \alpha^*) \rightarrow \mathcal N(0, I( \alpha^*)^{-1})
\end{align*}
where $I( \alpha^*)$ is the Fisher information matrix, evaluated at $ \alpha^*$. Since $\nabla_{\alpha} g( \alpha^*) = 0$, the Taylor expansion above can be written as 
\begin{align*}
     g( \widehat{\alpha}_N) = g( \alpha^*) + \frac 12 (  \widehat{\alpha}_N -  \alpha^*)^\top \nabla^2_{\alpha} g( \alpha^*) (\widehat{\alpha}_N - \alpha^*) + o_p(\left \Vert \widehat{\alpha}_N - \alpha^*\right\Vert^2).
\end{align*}
Rearrange we have
\begin{align*}
    N(g(\widehat{\alpha}_N) - g(\alpha^*)) &= \frac N2 ( \widehat{\alpha}_N - \alpha^*)^\top \nabla^2_\alpha g(\alpha^*) ( \widehat{\alpha}_N - \alpha^*) +  o_p\left(\left \Vert \sqrt N \left(\widehat{\alpha}_N -  \alpha^*\right)\right\Vert^2\right)\\
    &=\frac N2 ( \widehat{\alpha}_N - \alpha^*)^\top \nabla^2_\alpha g(\alpha^*) ( \widehat{\alpha}_N - \alpha^*) +  o_p\left(1\right),
\end{align*}
where we have used that $\sqrt{N}\left(\widehat{\alpha}_N -  \alpha^*\right) $ converges in distribution, so $\left \Vert \sqrt N \left(\widehat{\alpha}_N -  \alpha^*\right)\right\Vert^2 = O_p(1)$ and that $o_p\left( O_p(1)\right) = o_p(1)$.
The first term on the right hand side converges to the weighted sum of Chi-square distributions with one degree of freedom $\sum_{i = 1}^{p} c_i \chi_{1_i}^2$ as $N \rightarrow \infty$, where we assume $\alpha \in \mathbb R^{p}$, and the limit is independent of $N$. (See Theorem~\ref{thm:second_order_delta} and Appendix~\ref{app_chisqure}). Note that the limiting distribution  depends only on the dimension of $\alpha^*$, not the value of $ \alpha^*$, thus fixed. We then have
\begin{align*}
    g( \widehat{\alpha}_N) - g( \alpha^*) \xrightarrow{d} \frac 1N \sum_{i = 1}^{p} c_i \chi_{1_i}^2.
\end{align*}
Let $\mathbb E\left[\sum_{i = 1}^{p} c_i \chi_{1_i}^2 \right] = C_{p}$, and assume $g$ is a continuous and bounded function. Using Portmanteau theorem \cite{portmanteau}, we have that 
\begin{align*}
    \mathbb{E}_{\widehat{\alpha}_N} \left[g(\widehat{\alpha}_N)\right] - g(\alpha^*)\sim \mathcal O(1/N).
\end{align*}
We thus conclude that the bias converges at a rate of $ \mathcal O(1/N)$. 

We then study the variance of $\widehat{\mathrm{EIG}}_{M,N}$.
\begin{align*}
    \mathbb V \left[ \widehat{\mathrm{EIG}}_{M,N}\right] &= \mathbb V \left[ \frac 1M \sum_{i = 1}^M \log \pi_{Y|X}( y^i| x^i) - \frac 1M \sum_{i = 1}^M \log q_Y( y^i; \widehat{\alpha}_N)\right]\\
    &=\mathbb V \left[ \frac 1M \sum_{i = 1}^M \log \pi_{Y|X}( y^i| x^i) \right]+ \mathbb V \left[\frac 1M \sum_{i = 1}^M \log q_Y( y^i;  \widehat{\alpha}_N)\right] \\
    & - 2\mathbb{C}\text{ov} \left[\frac 1M \sum_{i = 1}^M \log \pi_{Y|X}( y^i| x^i),  \frac 1M \sum_{i = 1}^M \log q_Y( y^i;  \widehat{\alpha}_N)\right].
\end{align*}
We then analyze each of the three terms separately. The term
\begin{align*}
    \mathbb V \left[ \frac 1M \sum_{i = 1}^M \log \pi_{Y|X}( y^i| x^i) \right] = \frac 1M \mathbb V \left[\log \pi_{Y|X}( y| x) \right] \sim \mathcal O(1/M),
\end{align*}
as the samples are independent and the variance of $\log \pi_{Y|X}( y| x)$ is a fixed quantity. To analyze $\mathbb V \left[\frac 1M \sum_{i = 1}^M \log q_Y( y^i;  \widehat{\alpha}_N)\right]$, we use the law of total variance:
\begin{align*}
    &\mathbb V \left[ \frac 1M \sum_{i = 1}^M \log q_Y( y^i;  \widehat{\alpha}_N) \right]
    = \frac {1}{M^2} \mathbb V \left[  \sum_{i = 1}^M \log q_Y( y^i;  \widehat{\alpha}_N ) \right]\\
    = &\frac {1}{M^2} \left(\mathbb E_{  \widehat{\alpha}_N} \left[\mathbb V_{ \pi_Y} \left[ \sum_{i = 1}^M \log q_Y( y^i;  \widehat{\alpha}_N ) \right] \right] +  \mathbb V_{  \widehat{\alpha}_N} \left[\mathbb E_{ \pi_Y} \left[ \sum_{i = 1}^M \log q_Y( y^i;  \widehat{\alpha}_N ) \right] \right]\right)\\
    = &\frac {1}{M^2} \left(M\mathbb E_{  \widehat{\alpha}_N} \left[\mathbb V_{ \pi_Y} \left[ \log q_Y( y;   \widehat{\alpha}_N ) \right] \right] + M^2 \mathbb V_{  \widehat{\alpha}_N} \left[\mathbb E_{\pi_Y} \left[  \log q_Y( y;   \widehat{\alpha}_N ) \right] \right]\right)\\
    = &\left(\frac {1}{M} \mathbb E_{  \widehat{\alpha}_N} \left[\mathbb V_{\pi_Y} \left[ \log q_Y( y;   \widehat{\alpha}_N ) \right] \right]\right) + \left(\mathbb V_{  \widehat{\alpha}_N} \left[\mathbb E_{\pi_Y} \left[  \log q_Y( y;   \widehat{\alpha}_N ) \right] \right]\right),
\end{align*}
where we have used the fact that $\{\log q_Y( y^i;  \widehat{\alpha}_N )\}_{i=1}^M$ are independent given $\widehat{\alpha}_N$.
We then study two terms in the last equality separately. Define
\begin{align*}
    f( \alpha) = \mathbb V_{\pi_Y} \left[ \log q_Y( y;  \alpha) \right],
\end{align*}
then 
\begin{align*}
    f( \widehat{\alpha}_N) = f( \alpha^*) + \nabla_{ \alpha} f( \alpha^*) (  \widehat{\alpha}_N -  \alpha^*) + o_p\left(\left\Vert \widehat{\alpha}_N - \alpha^* \right\Vert^2\right)
\end{align*}
and 
\begin{align*}
    \mathbb E_{\widehat{\alpha}_N}\left[ f( \widehat{\alpha}_N)\right] = f( \alpha^*) + \nabla_{ \alpha} f( \alpha^*) \mathbb E_{\widehat{\alpha}_N}\left[  \widehat{\alpha}_N -  \alpha^*\right] + o_p\left(\mathbb E_{\widehat{\alpha}_N}\left\Vert \widehat{\alpha}_N - \alpha^* \right\Vert^2\right) \sim \mathcal O(1).
\end{align*}
For the second term, recall that $g(\alpha)$ was previously defined as
\begin{align*}
    g( \alpha)  = \mathbb E_{\pi_Y} \left [ \log q_Y( y; \alpha) \right] = \mathbb E_{a^*} \left [ \log q_Y( y; \alpha) \right].
\end{align*}
Using Portmanteau theorem again \cite{portmanteau}, we have that the second term 
\begin{align*}
    \mathbb V_{  \widehat{\alpha}_N} \left[\mathbb E_{\pi_Y
    } \left[  \log q_Y( y;  \widehat{\alpha}_N ) \right] \right]\sim \mathcal O\left(\frac{1}{N^2}\right),
\end{align*}
as we know that $\mathbb V[N (g(  \widehat{\alpha}_N) - g( \alpha^*))] = N^2 \mathbb V(g( \widehat{\alpha}_N))$ converge to the variance of the weighted sum of chi-square distribution as $N$ goes to infinity, which is a constant. Therefore, 
\begin{align*}
    \mathbb V \left[ \frac 1M \sum_{i = 1}^M \log q_Y( y^i;  \widehat{\alpha}_N) \right] \sim \mathcal O\left(\frac{1}{M} + \frac{1}{N^2}\right).
\end{align*}
We then study the covariance term $\mathbb{C}\text{ov} \left(\frac 1M \sum_{i = 1}^M \log \pi_{Y|X}( y^i| x^i),  \frac 1M \sum_{i = 1}^M \log q_Y( y^i; \widehat{\alpha}_N)\right)$. Note that $\mathbb{C}\text{ov}  \left( \log \pi_{Y|X}( y^i| x^i), \log q_Y( y^j;   \widehat{\alpha}_N)\right) = 0$ for $i\neq j$. Therefore the expression can be simplified to
\begin{align*}
    \frac {1}{M^2} \sum_{i=1}^M\mathbb{C}\text{ov}  \left( \log \pi_{Y|X}( y^i| x^i), \log q_Y( y^i;   \widehat{\alpha}_N)\right).
\end{align*}
Using Cauchy-Schwarz, we have
\begin{align*}
     &\mathbb{C}\text{ov}  \left( \log \pi_{Y|X}( y^i| x^i), \log q_Y( y^i;   \widehat{\alpha}_N)\right) \\
     \leq & \sqrt{\mathbb V\left[\log \pi_{Y|X}( y^i| x^i) \right] \mathbb V\left[\log q_Y( y^i; \widehat{\alpha}_N)\right]}.
\end{align*}
From previous analysis, the first term $\mathbb V\left[\log \pi_{Y|X}( y^i| x^i) \right]$ is a constant, and the second term $\mathbb V\left[\log q_Y( y^i; \widehat{\alpha}_N)\right]$ converges at a rate of $\mathcal O(1+\frac{1}{N^2})$ using the law of total variance. Therefore, 
\begin{align*}
    \frac {1}{M^2} \sum_{i=1}^M\mathbb{C}\text{ov}  \left( \log \pi_{Y|X}( y^i| x^i), \log q_Y( y^i;   \widehat{\alpha}_N)\right) \sim \frac 1M \mathcal O\left(1+\frac{1}{N^2}\right)=\mathcal O\left(\frac 1M \right ).
\end{align*}

To conclude, we have
\begin{align*}
 &\mathbb V \left[ \frac 1M \sum_{i = 1}^M \log \pi_{Y|X}( y^i| x^i) \right]  \sim \mathcal O\left(\frac{1}{M} \right),\\
&\mathbb V \left[ \frac 1M \sum_{i = 1}^M \log q_Y( y^i;  \widehat{\alpha}_N) \right] \sim  \mathcal O\left(\frac{1}{M} + \frac{1}{N^2}\right),\\
&\mathbb{C}\text{ov} \left(\frac 1M \sum_{i = 1}^M \log \pi_{Y|X}( y^i| x^i),  \frac 1M \sum_{i = 1}^M \log q_Y( y^i; \widehat{\alpha}_N)\right) \sim \mathcal O\left(\frac{1}{M} \right).
\end{align*}
Consequently, $\mathbb V \left[ \widehat{\mathrm{EIG}}_{M,N}\right] \sim \mathcal O\left(\frac{1}{M} + \frac{1}{N^2}\right)$.
\end{proof}

\begin{proof}[Proof of Theorem~\ref{thm:EIG_rate_2}]
    This theorem can be proved using the same techniques as in this \hyperref[proof:EIG_rate]{proof}, thus omitted. 
\end{proof}
\begin{proof}[Proof of Theorem~\ref{cor:MSE_rate}]
We first analyze the bias term,
\begin{align*}
    &\mathbb E\left[\widehat{\mathrm{EIG}}_{M,N} - {\mathrm{EIG}}\right]\\
    = &\revise{ \mathbb E_{\widehat{\alpha}_N} \mathbb E_{\pi_Y}}\left[\log \pi_Y( y) - \log q_Y( y; \widehat{\alpha}_N)\right] + \revise{ \mathbb E_{\widehat{\beta}_N} \mathbb E_{ \pi_{X,Y}}}\left[\log q_{Y|X}( y| x; \widehat{\beta}_N) - \log \pi_{Y|X}( y| x)  \right].
\end{align*}
Similar to the analysis in shown in the \hyperref[proof:EIG_rate]{proof} of Theorem~\ref{thm:EIG_rate}, we have that the bias term converges at a rate of $\mathcal O(1/N)$. We then analyze the variance term. 
\begin{align*}
    \mathbb V\left[\widehat{\mathrm{EIG}}_{M,N}\right]
    =& \mathbb V\left[\frac 1M \sum_{i = 1}^M \log q_{Y|X}( y^i| x^i;\widehat{\beta}_N) - \frac 1M \sum_{i = 1}^M \log q_Y( y^i; \widehat{\alpha}_N)\right]\\
    =& \mathbb V\left[ \frac 1M \sum_{i = 1}^M \log q_{Y|X}( y^i| x^i;\widehat{\beta}_N)\right] + \mathbb V\left[\frac 1M \sum_{i = 1}^M \log q_Y( y^i; \widehat{\alpha}_N)\right] \\
    - &2\mathbb{C}\text{ov} \left(\frac 1M \sum_{i = 1}^M \log q_{Y|X}( y^i| x^i; \widehat{\beta}_N),\frac 1M \sum_{i = 1}^M \log q_Y( y^i; \widehat{\alpha}_N)\right).
\end{align*}
Recall from the previous \hyperref[proof:EIG_rate]{proof} that the first and the second term converge at a rate of $\mathcal O(1/M + 1/N^2)$. For the covariance term, we use Cauchy-Schwarz and obtain
\begin{align*}
    &\mathbb{C}\text{ov} \left( \frac 1M\sum_{i = 1}^M \log q_{Y|X}( y^i| x^i; \widehat{\beta}_N), \frac 1M\sum_{i = 1}^M \log q_Y( y^i; \widehat{\alpha}_N)\right)\\
    \leq & \sqrt{\mathbb V\left[ \frac 1M\sum_{i = 1}^M \log q_{Y|X}( y^i| x^i; \widehat{\beta}_N) \right] \mathbb V\left[\frac 1M\sum_{i = 1}^M \log q_Y( y^i; \widehat{\alpha}_N)\right]}.
\end{align*}
We then know from the previous \hyperref[proof:EIG_rate]{proof} that both terms converge at a rate of $\mathcal O(1/M + 1/N^2)$. Therefore, we conclude that 
\begin{align*}
    \mathbb V\left[ \widehat{\mathrm{EIG}}_{M,N}\right] \sim \mathcal O\left(\frac 1M + \frac{1}{N^2}\right).
\end{align*}
\end{proof}

\begin{proof}[Proof of Corollary~\ref{cor:opt_allocation}]\label{proof:opt_allocation}
Suppose we have $L$ samples in total, and let $\gamma = \frac MN$. We then get $M = \frac{\gamma}{\gamma + 1} L$ and $N = \frac{1}{\gamma + 1}L$. We would then minimize $\frac{A}{\frac{\gamma}{\gamma + 1}L} + \frac{B}{\frac{L^2}{(\gamma + 1)^2}}$ over $\gamma$, where $A, B$ are constants. We set the derivative of the above expression to $0$ and get $\frac{2\gamma^3 B + 2\gamma^2 B - AL}{\gamma^2L^2} = 0$. Solving for $\gamma$, we get $\gamma = \frac 13 (\zeta + \frac{1}{\zeta} - 1)$, where $\zeta = \frac{(3 \sqrt 3 \sqrt{27A^2 B^4L^2 - 8AB^5L} + 27AB^2L - 4B^3)^{1/3}}{2^{2/3}B} \sim\mathcal O(L^{1/3})$. We see that if $L$ is large, then $\gamma \approx \frac{1}{3}(\zeta - 1)$. Therefore, if we let $L$ represent the total number of samples, the ratio of $M$ to $N$ should grow proportional to $L^{1/3}$. Using this optimal sample allocation, we obtain $M = \frac{L}{L^{-1/3} + 1}$ and $N = \frac{L}{L^{1/3} + 1}$. Consequently, the MSE decays at a rate of $\frac{2L^{2/3} + L + 2L^{1/3} + 1}{L^2}$, with the dominant term being of the order $L^{-1}$. 
\end{proof}

\section{Chi-square distribution}\label{app_chisqure}
In this section we derive the mean and variance of the quadratic form $(X - \mu)^\top \Sigma_2^{-1} (X - \mu)$ given a multivariate Gaussian random variable $X\sim \mathcal N(\mu, \Sigma_1)$.
\begin{theorem}
    Suppose $X \in\mathbb R^n$ is normally distributed, i.e., $X \sim \mathcal N(\mu, \Sigma_1)$, then the random variable $Y=(X - \mu)^\top \Sigma_2^{-1} (X - \mu)$ is the sum of Chi-square random variables, i.e., $(X - \mu)^\top \Sigma_2^{-1}(X - \mu) \sim \sum_{i=1}^n \lambda_i \chi_{1_i}^2$, where $\{\chi^2_{1_i}\}_{i=1}^n$ are independent Chi-square distributions with one degree of freedom 
   and $\lambda_i$'s are the generalized eigenvalues of the matrix pencil $(\Sigma_1, \Sigma_2)$. Consequently, the mean and the variance of $Y$ are $\sum_{i=1}^n \lambda_i$ and $2 \sum_{i=1}^n \lambda_i$, respectively.
\end{theorem}
\begin{proof}
    Consider the matrix pencil $(\Sigma_1, \Sigma_2)$, with generalized eigenvalue $\Lambda$ and eigenvector $\Phi$, where $\Lambda$ is a diagonal matrix with eigenvalues being its diagonal entries and the corresponding eigenvector being the column of $\Phi$. We then have $\Phi^\top \Sigma_1 \Phi = \Lambda, \Phi^\top \Sigma_2 \Phi = I$. Therefore, 
\begin{align*}
    \Sigma_1^{1/2} &= \Lambda^{1/2} \Phi^{-1}\\
    \Sigma_2^{1/2} &= \Phi^{-1}.
\end{align*}
Since $\Sigma_1^{-\top/2}(X - \mu) \sim \mathcal N(0, I)$, plugging in $\Sigma_1^{-\top/2} = \Lambda^{-1/2} \Phi^{\top}$, we have 
\begin{align*}
     \Phi^{\top}(X - \mu) \sim \mathcal N(0, \Lambda),
\end{align*}
from which we get 
\begin{align*}
    (X - \mu)^\top \Phi \Phi^{\top}(X - \mu)  =  (X - \mu)^\top \Sigma_2^{-1}(X - \mu) \sim \sum_{i=1}^n \lambda_i \chi_{1_i}^2.
\end{align*}
The mean is $\sum_{i=1}^n \lambda_i$, and the variance is $2 \sum_{i=1}^n \lambda_i$. 
\end{proof}

\section{EIG estimates results of the linear elasticity inverse problem \ref{subsec:wrench_mark}}\label{wrench_table}
\begin{table}[!ht]
\caption{Mean $\pm$ $2$ standard error using CMI for dimension reduction}
\centering
\begin{tabular}{|m{1cm}
| m{1.2cm}| m{1.2cm}| m{1.2cm}| m{1.2cm}| m{1.2cm}| m{1.2cm}| m{1.2cm}| m{1.2cm}|}\hline
& $s = 1$& $s = 2$& $s = 3$& $s = 4$& $s = 5$& $s = 6$& $s = 7$& $s = 8$\\
\hline
$r = 1$&$1.192 \pm 0.294$&$1.786 \pm 0.033$&$1.963 \pm 0.039$&$2.005 \pm 0.094$&$1.950 \pm 0.087$&$1.941 \pm 0.081$&$1.941 \pm 0.081$&$1.941 \pm 0.081$ \\
\hline
$r = 2$&$1.761 \pm 0.063$&$2.642 \pm 0.073$&$2.942 \pm 0.073$&$3.006 \pm 0.154$&$2.956 \pm 0.161$&$2.950 \pm 0.172$&$2.952 \pm 0.174$&$2.952 \pm 0.174$ \\
\hline
$r = 3$&$1.891 \pm 0.341$&$3.050 \pm 0.072$&$3.723 \pm 0.111$&$3.946 \pm 0.138$&$3.868 \pm 0.167$&$3.864 \pm 0.176$&$3.866 \pm 0.179$&$3.866 \pm 0.179$ \\
\hline
$r = 4$&$1.864 \pm 0.309$&$3.532 \pm 0.123$&$4.325 \pm 0.174$&$4.524 \pm 0.239$&$4.557 \pm 0.190$&$4.527 \pm 0.178$&$4.504 \pm 0.199$&$4.504 \pm 0.199$ \\
\hline
$r = 5$&$1.865 \pm 0.313$&$3.552 \pm 0.114$&$4.421 \pm 0.179$&$4.625 \pm 0.258$&$4.672 \pm 0.202$&$4.651 \pm 0.190$&$4.628 \pm 0.209$&$4.627 \pm 0.208$ \\
\hline
$r = 6$&$1.871 \pm 0.317$&$3.558 \pm 0.109$&$4.644 \pm 0.169$&$4.832 \pm 0.239$&$4.894 \pm 0.213$&$4.876 \pm 0.203$&$4.854 \pm 0.220$&$4.854 \pm 0.220$ \\
\hline
$r = 7$&$1.869 \pm 0.317$&$3.569 \pm 0.106$&$4.690 \pm 0.166$&$4.935 \pm 0.229$&$5.044 \pm 0.184$&$5.029 \pm 0.180$&$5.007 \pm 0.198$&$5.007 \pm 0.197$ \\
\hline
$r = 8$&$2.042 \pm 0.270$&$3.573 \pm 0.106$&$4.691 \pm 0.165$&$4.982 \pm 0.237$&$5.094 \pm 0.194$&$5.077 \pm 0.187$&$5.050 \pm 0.202$&$5.055 \pm 0.205$ \\
\hline
\end{tabular} 
\end{table}

\begin{table}[!ht]
\caption{Mean $\pm$ $2$ standard error using PCA for dimension reduction}
\centering
\begin{tabular}{|m{1cm}
| m{1.2cm}| m{1.2cm}| m{1.2cm}| m{1.2cm}| m{1.2cm}| m{1.2cm}| m{1.2cm}| m{1.2cm}|}\hline
& $s = 1$& $s = 2$& $s = 3$& $s = 4$& $s = 5$& $s = 6$& $s = 7$& $s = 8$\\
\hline
$r = 1$&$0.575 \pm 0.003$&$0.752 \pm 0.002$&$1.014 \pm 0.011$&$1.050 \pm 0.014$&$1.082 \pm 0.013$&$1.082 \pm 0.013$&$1.083 \pm 0.013$&$1.083 \pm 0.013$ \\
\hline
$r = 2$&$0.586 \pm 0.003$&$0.898 \pm 0.004$&$1.193 \pm 0.011$&$1.294 \pm 0.015$&$1.335 \pm 0.019$&$1.343 \pm 0.023$&$1.344 \pm 0.023$&$1.344 \pm 0.023$ \\
\hline
$r = 3$&$0.844 \pm 0.004$&$1.562 \pm 0.008$&$1.986 \pm 0.025$&$2.097 \pm 0.053$&$2.138 \pm 0.047$&$2.146 \pm 0.054$&$2.147 \pm 0.053$&$2.147 \pm 0.053$ \\
\hline
$r = 4$&$1.097 \pm 0.004$&$1.816 \pm 0.009$&$2.293 \pm 0.030$&$2.362 \pm 0.053$&$2.410 \pm 0.052$&$2.418 \pm 0.059$&$2.419 \pm 0.059$&$2.419 \pm 0.059$ \\
\hline
$r = 5$&$1.142 \pm 0.009$&$2.164 \pm 0.014$&$2.631 \pm 0.043$&$2.739 \pm 0.056$&$2.787 \pm 0.054$&$2.795 \pm 0.060$&$2.796 \pm 0.060$&$2.796 \pm 0.060$ \\
\hline
$r = 6$&$1.376 \pm 0.013$&$2.413 \pm 0.024$&$3.159 \pm 0.049$&$3.179 \pm 0.055$&$3.227 \pm 0.052$&$3.234 \pm 0.059$&$3.234 \pm 0.058$&$3.234 \pm 0.058$ \\
\hline
$r = 7$&$1.375 \pm 0.013$&$2.411 \pm 0.024$&$3.158 \pm 0.049$&$3.178 \pm 0.055$&$3.226 \pm 0.051$&$3.233 \pm 0.059$&$3.233 \pm 0.059$&$3.233 \pm 0.059$ \\
\hline
$r = 8$&$1.373 \pm 0.013$&$2.414 \pm 0.025$&$3.160 \pm 0.050$&$3.260 \pm 0.062$&$3.322 \pm 0.063$&$3.329 \pm 0.071$&$3.329 \pm 0.070$&$3.329 \pm 0.070$ \\
\hline
\end{tabular} 
\end{table}

\begin{table}[!ht]
\caption{Mean $\pm$ $2$ standard error using CCA for dimension reduction}
\centering
\begin{tabular}{|m{1cm}
| m{1.2cm}| m{1.2cm}| m{1.2cm}| m{1.2cm}| m{1.2cm}| m{1.2cm}| m{1.2cm}| m{1.2cm}|}\hline
& $s = 1$& $s = 2$& $s = 3$& $s = 4$& $s = 5$& $s = 6$& $s = 7$& $s = 8$\\
\hline
$r = 1$&$0.110 \pm 0.021$&$0.114 \pm 0.020$&$0.118 \pm 0.020$&$0.118 \pm 0.020$&$0.121 \pm 0.019$&$0.121 \pm 0.019$&$0.123 \pm 0.019$&$0.123 \pm 0.019$ \\
\hline
$r = 2$&$0.111 \pm 0.023$&$0.431 \pm 0.035$&$0.434 \pm 0.035$&$0.432 \pm 0.035$&$0.439 \pm 0.033$&$0.442 \pm 0.033$&$0.445 \pm 0.031$&$0.442 \pm 0.032$ \\
\hline
$r = 3$&$0.111 \pm 0.023$&$0.440 \pm 0.035$&$0.712 \pm 0.058$&$0.712 \pm 0.056$&$0.722 \pm 0.055$&$0.728 \pm 0.051$&$0.733 \pm 0.049$&$0.735 \pm 0.051$ \\
\hline
$r = 4$&$0.111 \pm 0.023$&$0.444 \pm 0.035$&$0.715 \pm 0.057$&$0.903 \pm 0.060$&$0.915 \pm 0.062$&$0.926 \pm 0.056$&$0.933 \pm 0.055$&$0.939 \pm 0.055$ \\
\hline
$r = 5$&$0.112 \pm 0.023$&$0.445 \pm 0.035$&$0.720 \pm 0.058$&$0.910 \pm 0.061$&$1.014 \pm 0.054$&$1.030 \pm 0.049$&$1.037 \pm 0.048$&$1.041 \pm 0.049$ \\
\hline
$r = 6$&$0.110 \pm 0.023$&$0.443 \pm 0.035$&$0.719 \pm 0.058$&$0.911 \pm 0.061$&$1.018 \pm 0.053$&$1.045 \pm 0.051$&$1.052 \pm 0.049$&$1.057 \pm 0.050$ \\
\hline
$r = 7$&$0.110 \pm 0.022$&$0.442 \pm 0.035$&$0.720 \pm 0.057$&$0.912 \pm 0.061$&$1.020 \pm 0.053$&$1.048 \pm 0.051$&$1.059 \pm 0.050$&$1.063 \pm 0.051$ \\
\hline
$r = 8$&$0.110 \pm 0.022$&$0.443 \pm 0.034$&$0.720 \pm 0.057$&$0.914 \pm 0.062$&$1.021 \pm 0.054$&$1.049 \pm 0.052$&$1.060 \pm 0.050$&$1.067 \pm 0.052$ \\
\hline
\end{tabular} 
\end{table}

\end{document}